\documentclass[12pt]{article}

\def\squarebox#1{\hbox to #1{\hfill\vbox to #1{\vfill}}}
\def\boxit#1{\vbox{\hrule\hbox{\vrule\kern6pt
          \vbox{\kern6pt#1\kern6pt}\kern6pt\vrule}\hrule}}

\usepackage{epsfig, adjustbox, caption,array}
\usepackage{amsmath,verbatim,amssymb}
\usepackage{verbatim,color,amsmath,amsthm}
\usepackage{hyperref}
\usepackage{lscape}
\usepackage{mathrsfs}
 \usepackage{graphicx,algorithm,float}
\usepackage{ulem}
\usepackage[absolute,overlay]{textpos}
\usepackage{lipsum} 
\newcommand{\wt}{\widetilde}
\renewcommand{\baselinestretch} {1.3}
\makeatletter 
\def\singlespace{\def\baselinestretch{1}\@normalsize}

\date{}

\renewcommand{\theequation} {\arabic{section}.\arabic{equation}}
\@addtoreset{equation}{section}

\newtheorem{theorem}{Theorem}[section]
\newtheorem{proposition}{Proposition}[section]
\newtheorem{lemma}{Lemma}[section]
\newtheorem{remark}{Remark}[section]
\newtheorem{corollary}{Corollary}[section]
\newtheorem{assumption}{Assumption}
\newtheorem{condition}{Condition}

\usepackage{amsmath,amssymb,mathrsfs,graphicx, color}
\usepackage{booktabs}

\def\bSig\mathbf{\Sigma}

\newcommand{\nl}{\langle}
\newcommand{\nr}{\rangle}
\newcommand{\ch}{{\cal H}}

\newcommand{\bs}{\boldsymbol}

\newcommand{\ms}{\mathscr}
\def\mb#1{\mathbb{#1}}
\def\wh{\widehat}
\def\wt{\widetilde}

\newcommand{\bse}{\begin{eqnarray*}}
\newcommand{\ese}{\end{eqnarray*}}
\newcommand{\bsee}{\begin{eqnarray}}
\newcommand{\esee}{\end{eqnarray}}

\begin{document}
\begin{textblock}{15}(1,1)
    \small
    The paper has been accepted by Biometrics
\end{textblock}

\begin{center}
{\Large\bf Statistical Inference for Heterogeneous Treatment Effect with Right-censored Data from Synthesizing Randomized Clinical Trials and Real-world Data}
\end{center}
\vspace{0.15in}
\begin{center}
{\bf Guangcai Mao} \\
Department of Biostatistics and Bioinformatics, Duke University\\
 Durham, North Carolina, 27710, U.S.A.\\
{{\it email}: maoguangcai@ccnu.edu.cn}
\end{center}
\begin{center}
{\bf Shu Yang}$^*$ \\
Department of Statistics, North Carolina State University\\
Raleigh, North Carolina 27695, U.S.A. \\
{{\it email}: syang24@ncsu.edu}
\end{center}
\begin{center}
{\bf Xiaofei Wang} \\
Department of Biostatistics and Bioinformatics, Duke University\\
 Durham, North Carolina, 27710, U.S.A.\\
{{\it email}: xiaofei.wang@duke.edu}
\end{center}
\vspace{0.15in}
\begin{center}
{\bf Abstract}
\end{center}
{
The heterogeneous treatment effect plays a crucial role in precision medicine.
There is evidence that real-world data,  even subject to biases,
can be employed as supplementary evidence for randomized clinical trials to improve the statistical efficiency of the heterogeneous treatment effect estimation.
In this paper, for survival data with right censoring, we consider estimating the heterogeneous treatment effect, defined as the difference of the treatment-specific conditional restricted mean survival times given covariates, by synthesizing evidence from randomized clinical trials and the real-world data with possible biases.
We define an omnibus bias function to characterize the effect of biases caused by  unmeasured confounders, censoring, and outcome heterogeneity,
and further, identify it by combining the trial and real-world data.
We propose a penalized sieve method to estimate the heterogeneous treatment effect and the bias function. We further study the theoretical properties of
the proposed integrative estimators based on the theory of reproducing kernel Hilbert space and empirical process.
The proposed methodology outperforms the approach solely based on the trial data through simulation studies and
an integrative analysis of the data from a randomized trial and a real-world registry on early-stage non-small-cell lung cancer.
}

\vspace{0.5cm}

\noindent{\bf KEY WORDS}:
Heterogeneous treatment effect;  
Inverse probability weighting; 
Nonparametric penalized estimation; 
Sieve approximation; 
Time-to-event endpoint.

\section{Introduction}\label{sec1}

The average treatment effect is commonly used to assess treatment effects, but it may not account for differences among individuals.
Precision medicine, considering individual characteristics, has spurred interest in heterogeneous treatment effect (HTE). 
Randomized clinical trials (RCTs) are reliable but limited in sample diversity and external validity.
Real-world data (RWD) can provide valuable complementary information but may be biased.
Addressing these biases and integrating the RWD with the RCT data enhances study efficiency and accuracy.

The integration of RCTs with the RWD is becoming increasingly common, particularly in oncology and rare-disease settings, where registry and electronic health record (EHR) systems can be readily linked to Phase III trials (e.g., FDA, 2021; Carrigan et al., 2022).
For example, several cancer RCTs have been successfully augmented with National Cancer Database data (Lee et al., 2024b) and with Flatiron Health EHR data to bolster subgroup analyses in older adults and smaller tumor-size strata (Ye et al., 2021). Other disease areas, such as diabetes, have likewise benefited from RCT–EHR integration (Kurki et al., 2024). 
Significant attention in the literature has been drawn to the integrative analysis method of the RCT data and the RWD.
Prentice et al. (2008) introduced joint analysis for pooled data, while
Soares et al. (2014) developed a hierarchical Bayesian model based on network meta-analysis.
Efthimiou et al. (2017) compared various approaches, including naive data synthesis, design-adjust synthesis, and a three-level hierarchical model using the RWD as prior information.
Verde and Ohmann (2015) summarized meta-analysis methods comprehensively, and
Wang and Rosner (2019) extended the propensity score adjustment to a multi-study setting, proposing a Bayesian nonparametric Dirichlet process mixture model.
Lee et al. (2023, 2024b)  introduced an integrative estimator of average treatment effect, and
Lee, Yang, and Wang (2022) and Lee et al. (2024a)  proposed doubly robust estimators for generalizing average treatment effects on survival outcomes from trial data to a target population. These methods, however, often assume no biases in the RWD, which is unlikely. Methods addressing biases include preliminary testing (e.g., Yang et al., 2023), instrumental variable methods (e.g., Angrist, Imbens, and Rubin, 1996), negative controls (e.g., Kuroki and Pearl, 2014), and sensitivity analysis (e.g., Robins, Rotnitzky, and Scharfstein, 1999).
Yang et al. (2025) introduced a confounding function approach to handle unmeasured confounder bias, leveraging transportability and trial treatment randomization to improve the HTE estimator.
Colnet et al. (2023) reviewed methods for combining the RCT data and the RWD for non-survival outcomes.

For censored survival data, the restricted mean survival time is an easily
interpretable, a clinically meaningful summary of survival function, which is defined as the
area under the survival curve up to a pre-specified time. A common way to measure the HTE is to define it as the difference in the conditional restricted mean survival time (CRMST)
between the treatment and control groups.
The existing estimations for the CRMST can be
roughly divided into two classes. One uses survival models for hazard rate. For example, Zucker (1998), Chen and Tsiatis (2001), and Zhang and Schaubel (2011) used the Cox
proportional hazards model (Cox, 1972) for hazard rate, deriving conditional survival function estimates from the relationship between hazard rate and survival function.
Another approach is to model the CRMST directly.
For example, Tian, Zhao, and Wei (2014) used a generalized linear regression model and estimated parameters through an inverse probability censoring weighted function, relying on the assumption that censoring time is independent of covariates.
Wang and Schaubel (2018) relaxed such an assumption and estimated the survival function of censoring time by the Cox model and further constructed the estimating equation.

This paper considers  statistical inference for the HTE with right-censored survival data by combining
the RCT data and the RWD  with possible biases.
The HTE is defined as the difference in the CRMST between the treatment and control groups.
Inspired by Wu and Yang (2022) and Yang et al. (2025), we define an omnibus bias function to summarize all sources of bias in the RWD.
This results in a flexible and generalizable framework that can be applied across a wide range of study designs and data sources.
As such, our approach is a valuable tool for researchers and practitioners working in various fields. Its capacity to handle biases in the RWD, without relying on unrealistic assumptions, makes it an essential tool for improving the accuracy and reliability of causal inference in real-world scenarios.
In our methodology,
the HTE and the bias function are modeled with fully nonparametric models and estimated by minimizing a proposed penalized loss function.
To implement the proposed estimators, a sieve method is utilized to approximate the HTE and the bias function.
Furthermore, we derive the convergence rates and local asymptotic normalities of the proposed estimators by reproducing kernel Hilbert space and empirical process theory.
Simulation studies demonstrate the excellent performance of the proposed method, and an illustrative application of the method to the real data reveals some intriguing findings.

The remainder of the paper is structured as follows:
Section \ref{sec2} introduces some preliminaries, such as notations and definitions of the HTE and  bias function.
Section \ref{sec3} introduces the penalized loss function and sieve method for estimating the unknown functional parameters.
Section \ref{sec4} presents the asymptotic properties of the proposed integrative estimators.
Section \ref{sec5} includes simulation studies and an application to a real non-small cell lung cancer dataset for finite sample performance evaluation of the proposed approach. Finally, in Section \ref{sec6}, we present some discussions, and Web Appendix E in the supplementary material, including the proofs of theoretical results, is also provided.

\section{Preliminaries}\label{sec2}

\subsection{Notations: HTE and data structure}\label{sub_sec2.1}
For two positive sequences $a_n$ and $b_n$, $a_n \asymp b_n$ means
$\lim_{n\rightarrow \infty}a_n / b_n = c $ for some constant $c > 0$.
For real numbers $a$ and $ b $, let $a\wedge b = \min\{ a , b \}$.
Let $T$ and $C$ denote the failure and censoring times, respectively.
Under right censoring, the observed variable is $(Y, \Delta)$, where $Y = T \wedge C $ is the observed time and $\Delta = I(T \le C )$ is the censoring indicator.
Let  $\bs X = (X_1, \ldots, X_p)^{\rm T}$ be the $p$-dimensional covariate and
$A \in\{ 0, 1\}$ be the binary treatment, where $A = 1$ and $A = 0 $ indicate the active and the control treatments, respectively.

We consider two independent data sources, the RCT data, and the RWD.
Let $S = 1$ denote the RCT participation and $S = 0 $ denote the RWD participation.
Therefore, the observed data structure for subject $i$ can be concluded as $(Y_i, \Delta_i, \bs X_i, A_i, S_i)$.
It is postulated that the data gathered from RCT comprises
${\cal V}_1 = \{(Y_i, \Delta_i, \bs X_{i}, A_i, S_i = 1): i = 1,\ldots, n_1\}$, with sample size $n_1$, which represents independent replications of $(Y, \Delta, \bs X , A , S = 1)$, while the data from RWD is represented by ${\cal V}_0 = \{(Y_i, \Delta_i, \bs X_i, A_i, S_i = 0): i = n_1 + 1,\ldots, n_1+ n_0 \}$ with a sample size of $n_0$,
which are independent copies of $(Y, \Delta, \bs X , A , S = 0)$.

This paper treats the RCT as the target population and utilizes
potential outcomes (Neyman, 1923; Rubin, 1974) as the framework to define causal effects. Specifically,
let $T(a)$ denote the potential outcome corresponding to the treatment $A = a$ for $a = 0 , 1$.
We make the causal consistency assumption that $T = A T(1) + (1 - A) T(0)$.
For a restricted time point $L$, the
HTE is defined as
\bse
\tau(\bs X) = E \left\{ T(1)\wedge L - T(0)\wedge L \mid \bs X, S = 1 \right\}.
\ese
Our research goal is to estimate $\tau(\bs X)$ based on the integrative dataset ${\cal V}_1 \cup {\cal V}_0$ containing $n = n_1 + n_0$ samples.

\subsection{Identifiability of the HTE}

We impose assumptions to explore the identiﬁability of the HTE from the observed data.
\begin{assumption}\label{A1}
   (i) $0 < c_1 \le P(A = 1 \mid \bs X , S = 1 ) \le c_2 < 1$, where $c_1$ and $c_2$ are some constants;
and (ii) $E\{T(a)\wedge L \mid  \bs X , A, S = 1\} = E\{T(a)\wedge L \mid \bs X , S = 1\}$ for $a = 0 ,1$.
\end{assumption}

\begin{assumption}\label{A2}
(i) $T(a) \perp C \mid (\bs X , A, S = 1 )$ for $a = 0 , 1 $;
and (ii) $P(Y \ge L \mid \bs X, A , S = 1 ) > 0$.  
\end{assumption}

Assumption \ref{A1}(i) implies that each subject in the RCT has a positive probability of receiving treatment. It is also considered a fundamental assumption.
Assumption \ref{A1}(ii) is satisfied by default for the RCT. Of note, this assumption is formally weaker than the strong  ignorability assumption on trial participation, i.e., $T(a) \perp A \mid (\bs X , S = 1)$ for $a = 0 ,1$, which is a traditional assumption in causal inference.
Assumption \ref{A2} is a standard assumption in survival analysis.
Assumption \ref{A2}(i) is also imposed in Zhang and Schaubel (2011).
Assumption \ref{A2}(ii) is a classical assumption in survival analysis and guarantees that the observed survival time can possess values within the vicinity of the restricted time point $L$.
Particularly, these Assumptions are exclusively imposed on the RCT and not on the RWD, thereby extensively broadening the scope of the proposed methodology.

Let
$T_L =  A\{ T(1) \wedge L \} + (1 - A)\{ T(0)\wedge L \}$,
$Y_L = Y \wedge L$,
$e(\bs X) = P(A = 1 \mid \bs X , S = 1)$ be the treatment propensity score, and
$\mu_a(\bs X) = E( T_L \mid \bs X , A = a, S = 1 ) $
with $a = 0 ,1$. We now deliberate on the  identifiability of the HTE from the RCT data.
If the failure time $T$ is precisely observed  for all subjects, to identify the HTE, we can use the method proposed by Lee, Okui, and Whang (2017), in which an augmented inverse probability weighting (AIPW) approach is proposed for complete data. To be specific, let
\bse
R_1 = \frac{A T_L}{e(\bs X)} - \frac{A - e(\bs X)}{e(\bs X)} \mu_1(\bs X), \quad
R_0 = \frac{(1 - A) T_L}{1 - e(\bs X)} + \frac{A - e(\bs X)}{1- e(\bs X)} \mu_0(\bs X), \quad R = R_1 - R_0.
\ese
Under Assumption 1,
Proposition S1 in Web Appendix C shows that
$
E( R \mid \bs X , S = 1 ) = \tau(\bs X)
$,
which  indicates  the identifiability of $\tau(\cdot)$ from the RCT data in the uncensored case.

To handle the right censoring,  an augmented inverse probability-of-censoring weighting (AIPCW) method (Zhao et al., 2015) is employed.
Let $G_T(t \mid \bs X , A, S = 1) = P(T \ge t \mid \bs X, A , S = 1)$ and $G_C(t \mid \bs X , A, S = 1) = P(C \ge t \mid \bs X, A , S = 1)$ be the conditional survival functions of $T$ and $C$ given $\bs X$, $A$ and $S = 1 $, respectively. Let
$\wt{\Delta} = I(T_L \le C)$, $N_C(t) = ( 1 - \wt{\Delta} ) I(Y_L \le t)$,
$Q_C(t) = \int_0^t I(Y_L \ge u ) G^{-1}_C(u \mid \bs X, A, S = 1){\rm d}G_C(u \mid \bs X, A, S = 1)$,
$M_C(t) = N_C(t) + Q_C(t)$,
$B(t) = E(T_L \mid T_L > t , \bs X, A, S = 1 ) = t + { \int_t^L G_T(u \mid \bs X, A , S = 1 ){\rm d}u } / {G_T(t \mid \bs X, A , S = 1 )}$,
and further
\bse
\wt{T}_L = {Y_L \wt{\Delta}}{G^{-1}_C( Y_L \mid \bs X , A, S = 1)} - \int_0^L B(t)G^{-1}_C(t \mid \bs X , A, S = 1){\rm d}M_C(t).
\ese
The first term of $\wt{T}_L$ is the familiar IPCW transformation, 
which is to convert the right-censored survival outcome into a complete outcome while preserving the conditional expectation of the outcome.
The second term of $\wt{T}_L$ is the augmented component, which enhances the efficiency and robustness of the transformation provided by the first term.
When implementing the transformation, one needs to use some survival models to estimate $G_C$ and $G_T$. Obviously, to ensure that the transformation is valid, the censoring time model must be correctly specified if there is no the second term. However, with the inclusion of the second half, either the failure time model or the censoring time model is correctly specified, not necessarily both. This is known as the double robustness property, which increases the tolerance for model misspecification.

Since $M_C(t)$ is a zero-mean martingale (Fleming and Harrington, 1991), the conditional mean of the second term of $\wt{T}_L$ is zero.
Intuitively, we have
$
E( \wt{T}_L \mid \bs X , A , S = 1 ) =
E( {T}_L \mid \bs X , A , S = 1 )
$. The detailed theoretical proof is provided in Proposition S2 in Web Appendix C. 
Consequently, defining $\wt{R}$ by replacing $T_L$ with $\wt{T}_L$ in the definition of $R$, we obtain
$
E( \wt{R} \mid \bs X , S = 1 ) = \tau(\bs X)
$.
We call $\wt{R}$ the pseudo-individual treatment effect (pseudo-ITE).

If $E(T(1)\wedge L \mid A = 1, \bs X, S = 0) - E(T(0)\wedge L \mid A = 0, \bs X, S = 0) = \tau(\bs X)$ and all trial assumptions are met for the RWD,
one can define the pseudo-ITE for the RWD in a similar manner to that for the RCT, denoted by $\wt{R}^*$.
In such a case, we have $E(\wt{R}^* \mid \bs X, S = 0 ) = \tau(\bs X)$.
Furthermore, let $D = S\,\wt{R} + (1 - S )\,\wt{R}^*$, then $E(D \mid \bs X, S) = \tau(\bs X)$,
which implies the identifiability of $\tau(\cdot)$ by combining the RCT data and the RWD.
However, these assumptions might not be true for the RWD due to the various biases in the RWD. Thus,
$\lambda(\bs X) = E( \wt{R}^* \mid \bs X , S = 0) - \tau(\bs X) = E( D \mid \bs X , S = 0) - \tau(\bs X) $, referred to as the bias function, captures the impact of any assumption violation in the RWD.
Therefore, we conclude that
\bsee\label{eq5}
E\left( D \mid \bs X , S \right) = \tau(\bs X) + (1 - S) \lambda(\bs X),
\esee
which implies the identifiability of $\tau(\cdot)$ and $\lambda(\cdot)$ from the integrative dataset ${\cal V}_1 \cup {\cal V}_0$.
In fact, the choice of the pseudo-ITE for the RWD is not unique; it can be any reasonable proxy.
In this paper, we opt to use $Y_L$ as a substitute for $\wt{R}^*$.

\section{Estimation methodology}\label{sec3}

Let $\sigma^2(\bs x , s) = \sigma^2_s(\bs x) = {\rm Var}(D \mid \bs X = \bs x, S = s)$.
Based on the integrative dataset ${\cal V}_1 \cup {\cal V}_0 = \{(Y_i, \Delta_i, \bs X_i, A_i, S_i), i = 1,\ldots, n  \}$,
equation (\ref{eq5}) enables us to construct the loss function for $\tau(\cdot)$ and $\lambda(\cdot)$ as
$
{\ell}_n(\tau, \lambda) = (2n)^{-1} \sum_{i = 1}^n \left\{\sigma^{2}(\bs X_i , S_i)\right\}^{-1} \left\{ {D}_i - \tau(\bs X_i) - (1 - S_i)\lambda(\bs X_i) \right\}^2
$.
Then, an estimated version of the loss function, $\wh{\ell}_n(\tau, \lambda)$, can be written as
$
\wh{\ell}_n(\tau, \lambda) = (2n)^{-1} \sum_{i = 1}^n \left\{\wh{\sigma}^{2}(\bs X_i , S_i)\right\}^{-1}\big\{ \wh{D}_i - \tau(\bs X_i) - (1 - S_i)\lambda(\bs X_i) \big\}^2
$,
where $\wh{D} = S\,\wh{R} + (1 - S)\,Y_L$, $\wh{R}$ is defined by replacing ${e}(\bs X)$, ${\mu}_a(\bs X)$, ${G}_T(t \mid \bs X , A, S = 1)$,
and ${G}_C(t \mid \bs X , A, S = 1)$ with the estimators $\wh{e}(\bs X)$, $\wh{\mu}_a(\bs X)$, $\wh{G}_T(t \mid \bs X , A, S = 1)$, and $\wh{G}_C(t \mid \bs X , A, S = 1)$ in the definition of $\wt{R}$, respectively, and $\wh{\sigma}^2(\bs X, S)$ is the estimator of ${\rm Var}({D} \mid \bs X, S)$.
In this paper, we propose
the penalized loss function for $\tau(\cdot)$ and $\lambda(\cdot)$ as
$
\wh{\ell}_{n, \gamma_1, \gamma_0}(\tau, \lambda) = \wh{\ell}_n(\tau, \lambda) + {\gamma_1}/{2}J_1(\tau, \tau) + {\gamma_0}/{2}J_0(\lambda, \lambda)
$,
where $\gamma_1$ and $\gamma_0$ are some positive penalized parameters and converge to zero as the sample size $n$ goes to infinity, 
$J_1(\cdot, \cdot)$ and
$J_0(\cdot, \cdot)$ are roughness penalties to avoid overfitting, encourage smoothness and information borrowing, and are defined in Web Appendix A.

Let $\tau_0(\cdot)$ and $\lambda_0(\cdot)$ be the true values of $\tau(\cdot)$ and $\lambda(\cdot)$. 
Throughout the paper, we assume that $\tau_0(\cdot)$ and $\lambda_0(\cdot)$ belong to the reproducing kernel Hilbert spaces ${\cal H}_1$ and ${\cal H}_0$, respectively, equipped with the norms $\| \cdot \|_{\ch_1}$ and $\| \cdot \|_{\ch_0}$. The spaces ${\cal H}_1$ and ${\cal H}_0$, as well as the norms $\| \cdot \|_{\ch_1}$ and $\| \cdot \|_{\ch_0}$,  are defined in Web Appendix A. 
For inferring $\tau_0(\cdot)$ and $\lambda_0(\cdot)$, we use sieve expansion to approximate $\tau(\cdot)$ and $\lambda(\cdot)$ in $\wh{\ell}_{n, \gamma_1, \gamma_0}(\tau, \lambda)$.
Specifically, let $\{\phi_1(\cdot), \ldots, \phi_{r_1}(\cdot)\}$ and $\{\psi_1(\cdot), \ldots, \psi_{r_0}(\cdot)\}$ denote two sets of sieve basis functions, and let $\Phi_n$ and  $\Psi_n$ be the corresponding spanned linear spaces,
respectively, where $r_1$  and $r_0$ are the numbers of basis functions and represent the complexities of the approximations.
Then, we propose the penalized nonparametric estimator of $(\tau_0, \lambda_0)$ as
$
(\wh{\tau}_{n}, \wh{\lambda}_{n}) =
\arg{\min}_{(\tau ,\lambda) \in \Phi_n \times \Psi_n}\wh{\ell}_{n, \gamma_1, \gamma_0}(\tau, \lambda)
$.

\section{Asymptotic properties}\label{sec4}

In this section, we delve into the asymptotic properties of the proposed estimators,
focusing on consistency, convergence rates, and asymptotic normality.

\begin{theorem}\label{th1}
Assuming that Assumptions 1--2 are met, as well as Conditions S1--S6 as outlined in Web Appendix B.
Then,
$\| \wh{\tau}_n - {\tau}_0 \|_{\ch_1} +  \|\wh{\lambda}_{n} - \lambda_0 \|_{\ch_0} = o_P(1)$.
\end{theorem}

Given that $\mu_a(\bs X) = E(T_L \mid \bs X, A = a, S = 1) = \int_0^L G_T(t \mid \bs X, A = a, S = 1 ){\rm d}t$, using $\wh{\mu}_a(\bs X) = \int_0^L \wh{G}_T(t \mid \bs X, A = a, S = 1 ){\rm d}t$ as an estimator for $\mu_a(\bs X)$ is appropriate.
In this context, Theorem \ref{th1} implies that if $\wh{G}_T(t \mid \bs X, A , S = 1)$ is consistent, or both $\wh{G}_C(t \mid \bs X, A , S = 1)$ and $\wh{e}(\bs X)$ are consistent, then $\wh{\tau}_n(\cdot)$ and $\wh{\lambda}_n(\cdot)$ are consistent estimators of $\tau_0(\cdot)$ and $\lambda_0(\cdot)$, respectively.
For the RCT, where the probability of receiving treatment is typically known, Theorem \ref{th1} establishes that if either $\wh{G}_T$ or $\wh{G}_C$ is consistent, $\wh{\tau}_n$ and $\wh{\lambda}_n$ remain consistent. In such a case,
Theorem \ref{th1} thus highlights the double robustness of the proposed HTE estimator, an important attribute that enhances its reliability.
Alternatively, regardless of the knowledge of $e(\bs X)$, one can use a nonparametric estimation approach for $\mu_a(\bs X)$. Employing such methods ensures consistency, thereby providing reliable estimators $\wh{\tau}_n$ and $\wh{\lambda}_n$.
Under stronger assumptions, Theorem \ref{th2} provides the convergence rates of the proposed estimators, making it a strengthened version of Theorem \ref{th1}.

\begin{theorem}\label{th2}
Assuming that Assumptions \ref{A1}--\ref{A2} are met, as well as Conditions S1--S5, and S7--S9 as outlined in Web Appendix B.
Then,
\bse
\|\wh{\tau}_{n} - \tau_0 \|_{\ch_1} + \|\wh{\lambda}_{n} - \lambda_0 \|_{\ch_0}  =  O_P(n^{-1/2}\gamma_1^{-p/(4m_1)} + n^{-1/2}\gamma_0^{-p/(4m_0)} + \gamma_1^{1/2} + \gamma_0^{1/2}).
\ese
\end{theorem}

\begin{remark}\label{remark1}
If $m_1 = m_0 = m$ and $\gamma_1 \asymp \gamma_0 \asymp n^{-2m/(2m+p)}$,
then $\wh{\tau}_{n}$ and $\wh{\lambda}_n$ achieve the same optimal convergence rate of $O_P(n^{-m/(2m + p )})$, which is also the optimal rate in the most commonly used nonparametric methods.
\end{remark}

Suppose that the covariate $\bs X$ takes values in an open connect set $\Omega$ with $C^{\infty}$ boundary. The pointwise asymptotic normality, which is 
crucial for constructing confidence intervals, is investigated in Theorem \ref{th3} under additional conditions.

\begin{theorem}\label{th3}
Assuming that the assumptions in Theorem \ref{th2} are met, as well as Conditions S10--S14 as outlined in Web Appendix B.
Then, given $\bs x_0 \in \Omega$, we have
\bse
\left(n^{1/2}\gamma_1^{p/(4m_1)}\{\wh{\tau}_n(\bs x_0) - \tau^*_0(\bs x_0) \} ,
    n^{1/2}\gamma_0^{p/(4m_0)}\{\wh{\lambda}_n(\bs x_0) - \lambda^*_0(\bs x_0) \} \right)^{\rm T} \rightsquigarrow N(0, \bs \Sigma),
\ese
where $\rightsquigarrow$ denotes convergence in distribution, 
$\tau^*_0$ and $\lambda^*_0$ are the biased ``true values", and $\bs \Sigma$ is the covariance matrix. The definitions of $\tau^*_0$, $\lambda^*_0$, and $\bs \Sigma$ are provided in Web Appendix A.
\end{theorem}

Theorem \ref{th3} further indicates that the penalized components bias the proposed integrative estimators (see the definitions of $\tau^*_0$ and $\lambda^*_0$ in Web Appendix A).
Analyzing the biases is challenging, however, the biases can be disregarded under specific undersmoothing conditions, as shown in Corollary \ref{coro1lary1}.

\begin{corollary}\label{coro1lary1}
Suppose that the assumptions in Theorem \ref{th3} hold; furthermore, $n(\gamma_1 + \gamma_0) = O(1)$.
Then, given $\bs x_0 \in \Omega$,
\bse
\left(n^{1/2}\gamma_1^{p/(4m_1)}\{\wh{\tau}_n(\bs x_0) - \tau_0(\bs x_0) \} ,
    n^{1/2}\gamma_0^{p/(4m_0)}\{\wh{\lambda}_n(\bs x_0) - \lambda_0(\bs x_0) \} \right)^{\rm T} \rightsquigarrow N(0, \bs \Sigma).
\ese
\end{corollary}

Let $\wh{\tau}_{\rm rct}(\cdot)$ denote the estimator of $\tau_0(\cdot)$ derived solely from the trial data.
The last aim is to theoretically show that the proposed integrative method is more efficient than using only the trial data, specifically ${\rm Var}\{\wh{\tau}_n(\bs x_0) \} \le {\rm Var}\{ \wh{\tau}_{\rm rct}(\bs x_0) \}$.
However, the estimators involve estimating nuisance functions, which makes calculating the variances very complex. For simplification, under some stronger conditions and treating nuisance functions as known, we show the proposed method has the advantage of gaining efficiency in estimating $\tau(\cdot)$, which is summarized in Theorem \ref{th4}.

\begin{theorem}\label{th4}
Assuming that Assumptions \ref{A1}--\ref{A2} are met, as well as Conditions S1--S3, and S15--S17 as outlined in Web Appendix B.
Then, given $\bs x_0 \in \Omega$,
$
{\rm Var} \{\wh{\tau}_n(\bs x_0)\} \le {\rm Var} \{\wh{\tau}_{\rm rct}(\bs x_0)\}
$.
\end{theorem}

\section{Numerical studies}\label{sec5}
\subsection{Simulation study}
In this section, we conduct some simulation studies to provide technical support for applying the proposed method. We consider various types of bias in the real-world study, including selection bias, censoring, outcome heterogeneity, and unmeasured confounding.
We first consider the case where $ p = 2$, that is $\bs X = (X_1, X_2)^{\rm T}$.
Tensor B-splines are employed to get the sieve basis.
In the simulation,
$\wh{G}_T$ and $\wh{G}_C$ is obtained by fitting a Cox proportional hazards model (Cox, 1972), furthermore, $\wh{\mu}_a(\bs X) = \int_0^L \wh{G}_T(t \mid \bs X, A = a, S = 1){\rm d}t $, where $a = 0 , 1$. $\wh{\sigma}^2_s(\bs x)$  is obtained by a kernel estimation method and the corresponding bandwidth is twice the optimal bandwidth by the generalized
cross-validation (GCV), where $s = 0 ,1$.
The penalized parameters $\gamma_1$ and $\gamma_0$ are selected by GCV.
For comparison,  we also consider the estimators based solely on the RCT data and those based solely on the RWD.

Throughout the simulation, $ e(\bs X) = P(A = 1 \mid \bs X, S = 1) = 0.5 $ is known and
$P(A = 1 \mid \bs X, S = 0 )$ is estimated via a generalized linear model.
We consider three cases, each representing different scenarios. Case 1 corresponds to the situation where the failure time model is correctly specified, and $ E(T_L \mid \bs X, A, S = 1) = E(T_L \mid \bs X, A, S = 0) $. Case 2 corresponds to the situation where the failure time model is correctly specified, but $ E(T_L \mid \bs X, A, S = 1) \neq E(T_L \mid \bs X, A, S = 0) $. Case 3 addresses the situation where the failure time model is incorrectly specified, and $ E(T_L \mid \bs X, A, S = 1) = E(T_L \mid \bs X, A, S = 0) $.
The specific settings are summarized in Table \ref{table1}.

The simulation results are summarized in Figures \ref{fig-1}--\ref{fig-2} and S1--S4 in Web Appendix F.
When presenting results graphically, we use different scales for the legends to ensure visibility and interpretability.
The RWD-based only method has a distinct scale due to its significantly larger bias and standard deviation values.
Using the same scale for all methods would make the RCT-based only and proposed methods indistinguishable.
Thus, we assign a separate scale to the RWD-based only method and use a shared scale for the RCT-based only and proposed integrative methods to highlight the superiority of the proposed method.

Our expectations are met. 
The estimator based solely on the RWD performs poorly. Even though its performance improves with increasing sample size, the results remain poor, characterized by large biases, standard deviations, and unfavorable coverage probabilities.
Both the proposed estimator and the estimator based solely on the RCT data perform well in estimation accuracy and coverage probability,
closely aligning with the nominal values.
The empirical standard
deviations of both the proposed estimator and the estimator based solely on the RCT data  decrease when the
sample size $(n_1, n_0)$ is increased from $(500, 1000)$ to $(1000, 2000)$.
Notably, in all configurations, the empirical standard deviations of
the proposed estimators are smaller than those obtained solely from the RCT data.
Considering estimation accuracy, coverage probability, and standard deviation together, it is evident that while both methods perform well in accuracy and coverage probability, the proposed method stands out due to its lower standard deviation.
The reduced variability of the proposed estimator enhances its reliability, making it a superior choice for estimating the HTE in real-world scenarios.
Overall, these findings confirm that the proposed integrative method performs well in finite-sample settings and significantly improves the HTE estimation compared to using the RCT data solely.

We also consider the case where $p = 4$, which includes two continuous covariates and two categorical covariates. The details are presented in Web Appendix F. 

\subsection{Real data application}
Due to advances in radiologic technology, the detection rate of early-stage non-small-cell lung cancer is rising.
The research community on the treatment of early-stage lung cancer with $\le$ 2cm tumor had a great interest in evaluating the effect of limited resection relative to lobectomy.
Lobectomy is a popular surgical resection in which the entire lobe of the lung where the tumor resides is removed.
Limited resection, including wedge and segmental resection, only removes a smaller section of the complicated lobe.
Limited resection is known for shorter hospital stays, fewer postoperative complications,  and better preservation of pulmonary function. CALGB 140503 is a multicenter non-inferiority randomized phase 3 trial in which 697 patients with
non-small-cell lung cancer clinically staged as stage 1A with $\leq$ 2cm tumors were randomly assigned to
undergo limited resection or lobectomy (Altorki et al., 2023). The results of this trial firmly established that for stage 1A non-small-cell lung cancer patients with a tumor size of 2cm or less, limited resection  was not inferior to lobectomy concerning overall survival with a hazard ratio = 0.95 (90\% confidence interval 0.72--1.26) and disease-free survival with a hazard ratio = 1.01 (90\% confidence interval 0.83--1.24).
Further subgroup analysis reveals that patients older than 70 years tended to have more prolonged disease-free survival and overall survival when receiving limited resection. At the same time,  patients with larger tumor size (1.5--2cm) tended to benefit more from lobectomy. This leads to a strong interest in exploring treatment effect heterogeneity over age and tumor size.
In  CALGB 140503, we removed two patients with overall survival times equal to 0, as survival time should logically be greater than 0, and these data points were considered invalid.
In the remaining 695 patients, we detected an outlier in tumor size by calculating the $3\sigma$ interval for tumor size, which was found to be $(0.413, 2.551)$. Since the tumor size of 3 exceeded this interval, we concluded it was an outlier.
Therefore, we excluded this patient's data to ensure the robustness of our analysis. After these exclusions, we used data from 694 patients in CALGB 140503.

The National Cancer Database (NCDB) is a clinical oncology database maintained by the American College of Surgeons, and it captured $72\%$ of all newly diagnosed lung cancers.  
From the NCDB database, we selected a cohort of 17,995 stage 1A NSCLC patients with $\leq$ 2cm tumor and who met all eligibility criteria of CALGB 140503.
The NCDB-only analysis based on multivariable Cox proportional hazards model and propensity score-based methods reveals a significant overall benefit of lobectomy over limited resection, which contradicts the findings of CALGB 140503.
Unobserved hidden confounders in the NCDB-only analysis could explain the beneﬁt of lobectomy over limited resection.
It has been well documented that surgeons and patients tend to choose limited resection over lobectomy if the patient has a bad health status and poor functional respiratory reserve and/or high comorbidity burden (Zhang et al, 2019; Lee and Altorki, 2023).
Unfortunately these confounders were not captured in the NCDB database, and these hidden confounders may have inevitably led to biased estimates of the treatment effects.
It is of a great interest to illustrate our proposed method to estimate the HTE. In particular, we  want to examine the precision of the proposed HTE estimator for the difference of treatment-specific CRMSTs, conditional on age and tumor size, that can be improved by synthesizing information from CALGB 140503 and NCDB cohorts with the latter subject to possible hidden confounders.

Our analysis considers the time to death as the survival time and takes age and tumor size as the covariates of interest in a two-dimensional context. The restricted time horizon is set at 3 years. Table \ref{table2} displays the descriptive statistics of age and tumor size  for the CALGB 140503 and NCDB cohorts.
Figure \ref{fig-3} summarizes the results from the trial data-based only and proposed integrative methods. Specifically, the first two panels show the point estimates of the HTE and highlight in blue the regions where the treatment effects are statistically significant, as indicated by $95\%$ confidence intervals that do not include zero.

In general, there is a trend where patients with large tumor size  benefit more from lobectomy compared to limited resection, while patients with small tumor size and older age benefit more from limited resection over lobectomy.
The blue regions are more extensive in the second panel, which represents the proposed integrative approach.
This indicates that the proposed method identifies a broader range of patient groups where the treatment effects are statistically significant.
Specifically, the blue regions highlight that patients with larger tumor sizes ($>$ 1.5cm) and younger ages ($<$ 65 years) show statistically significant benefits from lobectomy.
The third panel illustrates the percentage reduction in standard error achieved by the proposed integrative approach compared to the trial data-based only approach.
This reduction demonstrates that the proposed integrative method provides HTE estimates with lower standard errors and narrower confidence intervals.
These findings indicate that the proposed integrative method enhances estimation efficiency in real-world applications.

\section{Discussion}\label{sec6}

We proposed an integrative estimator of the HTE by combining evidence from the RCT and the RWD in the presence of right censoring. We avoided the assumption of no biases for the RWD and instead defined a  bias function to account for various biases in the RWD. The proposed method considered the HTE in a fully nonparametric form, making it flexible, model-free, and data-driven, and hence more practical for use in various applications. Our research aimed to increase the efficiency of the HTE estimation by leveraging supplemental information provided by biased RWD, as verified by our numerical studies.

In the introduction, we emphasized the crucial role that the HTE plays in precision medicine. Individualized treatment and precision medicine are frequently used interchangeably. Since the HTE provides guidance regarding which treatment strategy should be adopted, our proposed framework is closely related to the individualized treatment regime, which involves a decision rule that assigns treatments based on patients' characteristics. Recent studies, such as Chu, Lu, and Yang (2022) and Zhao, Josse, and Yang (2025), have explored using data from various sources for the individualized treatment regime. Our research sheds light on the potential application of treatment effects for survival outcomes in precision medicine or individualized treatment, making it a valuable contribution to the field.

Our work also has several limitations that warrant discussion. Firstly, the proposed nonparametric estimators may suffer from boundary effects, which is a common issue in nonparametric statistics. Specifically, point estimates of the HTE in boundary regions could be less accurate than those at interior points due to slower convergence rates of nonparametric estimators around the boundary. Secondly, when using the sieve method to approximate the HTE and the  bias function, the dimensionality of the covariate $\bs X$ should not be high. 
Also, we have thus far focused on non–time-varying treatments. Moving forward, marginal structural models (e.g., Yang, Tsiatis, and Blazing, 2018; Yang, Pieper, and Cools, 2020) offer considerable potential to enhance interpretability in these contexts and constitute a promising direction for future research on data integration.

\newpage
\begin{figure}[H]
\centering
\includegraphics[width = \textwidth]{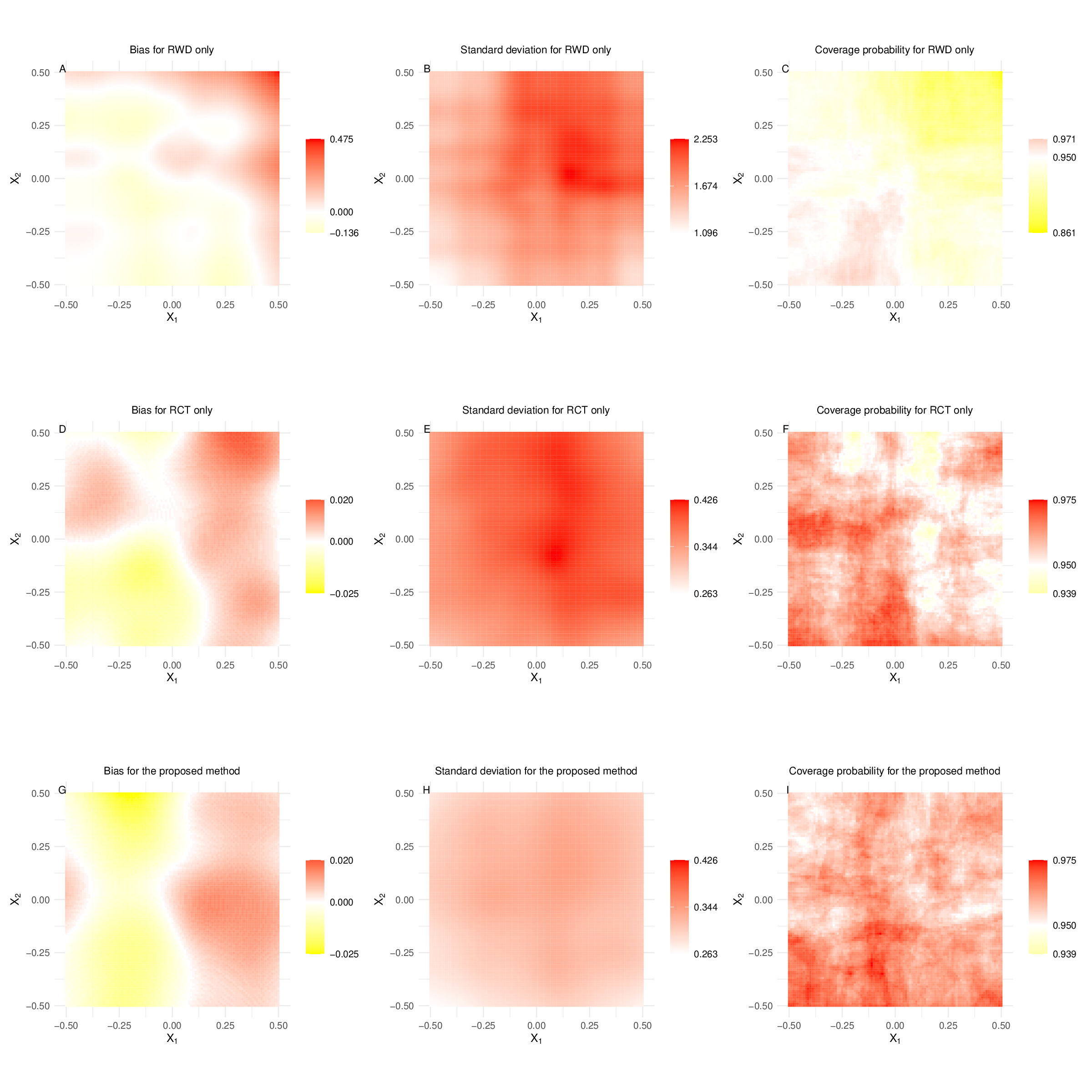}
\caption{The simulation results of Case 1 with $(n_1, n_0) = (500, 1000)$.}
\label{fig-1}
\end{figure}

\newpage
\begin{figure}[H]
\centering
\includegraphics[width = \textwidth]{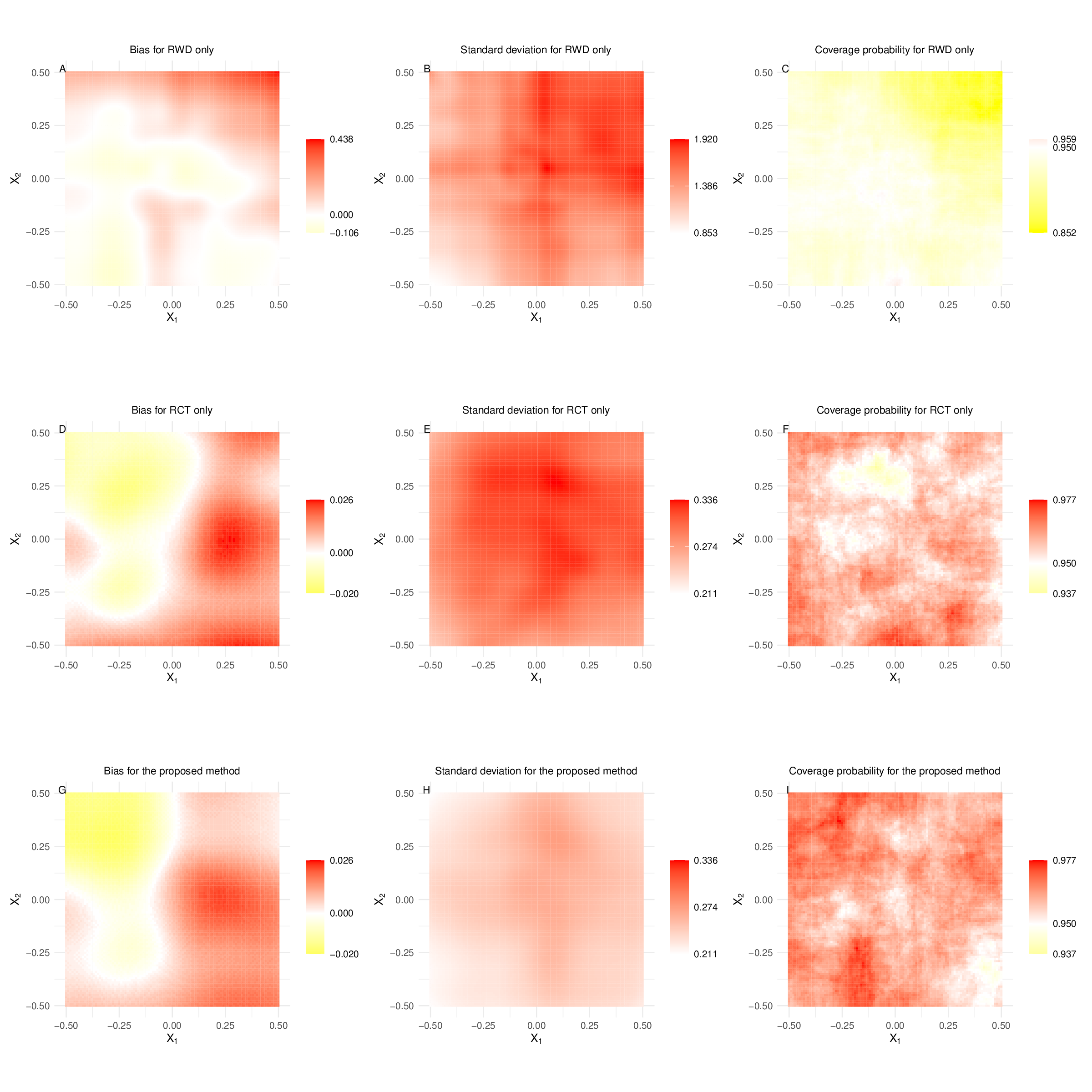}
\caption{The simulation results of Case 1 with $(n_1, n_0) = (1000, 2000)$.}
\label{fig-2}
\end{figure}

\newpage
\begin{figure}[H]
\centering
\includegraphics[width = \textwidth]{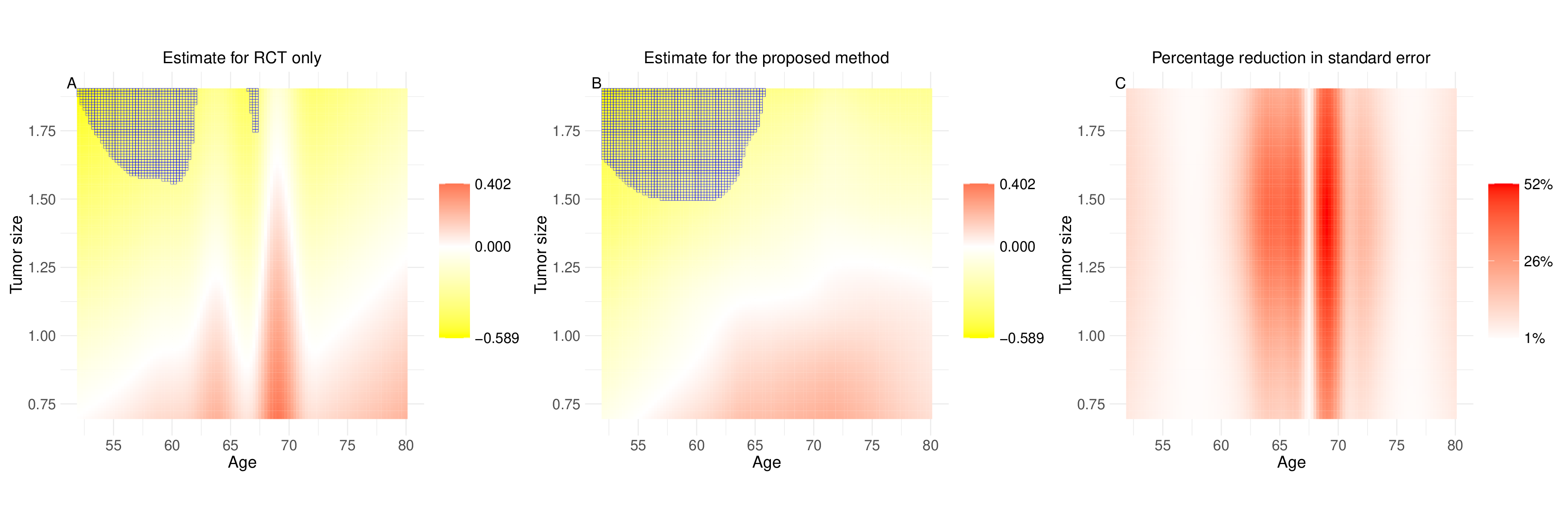}
\caption{The analysis results of real data.}
\label{fig-3}
\end{figure}

\newpage
\begin{table}[H]
\begin{center}
\renewcommand\arraystretch{1.3}
    \caption{Simulation Settings}
    \label{table1}
   \resizebox{\textwidth}{!}{%
    \begin{tabular}{lccc}
        \toprule
    \multicolumn{4}{l}{ Restricted time point $L = 3$; Sample size $(n_1, n_0) = (500, 1000)$ and $(1000, 2000)$; Replication times = 1000} \\

 \textbf{Case}  &     \textbf{Parameter} & \textbf{RCT} & \textbf{RWD}
\\
        \midrule
           Case 1  &  $X_1, X_2$ & $\mathcal{N}(0, 1)$ & $\mathcal{N}(0, 0.5)$
\\
   &    $X_u$ &  ${\rm Exp}(5)$ & ${\rm Exp}(5)$
\\
     &   $A$ & Bernoulli(0.5) & Bernoulli(${\rm expit}\{0.5(X_1 + X_2 - X_u + 1 )\}$)
\\
     & $T$ & \multicolumn{2}{c}{$ G_{T}(t \mid X_1, X_2, A , X_u )
 =  (1 + 0.02t)\exp\{ -0.1X_u t -0.2t \exp( -0.2X_1 -0.5 X_2 + 0.4AX_1 + 1.3AX_2 )\}
$}
\\
     & $C$ & \multicolumn{2}{c}{$ h_C(t \mid X_1, X_2) = h_{0C}(t) \exp( 0.5 X_1 + 0.5 X_2 )$}
\\
& & $ h_{0C}(t) = 1.47 \times 10^{-2} $ & $ h_{0C}(t) = 0.441 $
\\
& $\wt{L}$ & 4.9 & 4.5
\\
& CR & 40\% & 70\%
\\
        \midrule
           Case 2  &  $X_1, X_2$ & \multicolumn{2}{c}{Same as in Case 1}
\\
   &    $X_u$ &  $-$ & $\mathcal{N}(0, 1)$
\\
     &   $A$ & \multicolumn{2}{c}{Same as in Case 1}
\\
     & $T$ & \multicolumn{2}{c}{$h_{T}(t \mid X_1, X_2, A ) = $}
\\

 & &  $0.2 \exp( -0.2X_1 -0.5 X_2 + 0.4AX_1 + 1.3AX_2 )$ & $ 0.2 \exp( -0.2X_1 -0.5 X_2 + 0.4AX_1 + 1.3AX_2 + X_u )$
\\
     & $C$ & \multicolumn{2}{c}{$ h_C(t \mid X_1, X_2) = h_{0C}(t) \exp( 0.5 X_1 + 0.5 X_2 )$}
\\
& & $ h_{0C}(t) = 1.84 \times 10^{-2} $ & $ h_{0C}(t) = 0.552 $
\\
& $\wt{L}$ & 5 & 4.5
\\
& CR & \multicolumn{2}{c}{Same as in Case 1}
\\
        \midrule
           Case 3  &  $X_1, X_2$ & \multicolumn{2}{c}{Same as in Case 1}
\\
   &    $X_u$ &  $\mathcal{N}(0, 1)$ & $\mathcal{N}(0, 1)$
\\
     &   $A$ & \multicolumn{2}{c}{Same as in Case 1}
\\
     & $T$ & \multicolumn{2}{c}{$h_{T}(t \mid X_1, X_2, A ) = 0.2 \exp( -0.2X_1 -0.5 X_2 + 0.4AX_1 + 1.3AX_2 + X_u )$}
\\
     & $C$ & \multicolumn{2}{c}{$ h_C(t \mid X_1, X_2) = h_{0C}(t) \exp( 0.5 X_1 + 0.5 X_2 )$}
\\
& & $ h_{0C}(t) = 1.47 \times 10^{-2} $ & $ h_{0C}(t) = 0.552 $
\\
& $\wt{L}$ & 5.2 & 4.5
\\
& CR & \multicolumn{2}{c}{Same as in Case 1}
\\
        \bottomrule
    \end{tabular}%
}
\end{center}
\vspace{0.02in}
\footnotesize{
$X_1, X_2$: Baseline covariates; $X_u$: Unmeasured confounder in the RWD; $A$: Treatment assignment; $T$: Failure time;
$C$: Censoring time;
$\wt{L}$: Study duration; CR: Censoring rate; $\mathcal{N}$: Normal distribution; Exp: Exponential distribution;
expit: ${\rm expit}(x) = \exp(x) / \exp(1 + x)$;
$G_T$: Survival function of $T$; $h_T$: Hazard function of $T$; $h_C$: Hazard function of $C$;
$h_{0C}$: Baseline hazard function of $C$.
}
\end{table}

\newpage
\begin{table}[H]
  \centering
  \caption{Distribution of age and tumor size for CALGB 140503 and NCDB Cohorts}
    \label{table2}
    \vspace{0.05in}
\renewcommand\arraystretch{1.5}
{\setlength{\tabcolsep}{0.2mm}
 \begin{tabular}{lrrrr}
    \toprule
          & \multicolumn{2}{c}{\textbf{NCDB}} & \multicolumn{2}{c}{\textbf{CALGB 140503}} \\
\cmidrule{2-5}          & \textbf{Lobectomy} & \textbf{Limited Resection} & \textbf{Lobectomy} & \textbf{Limited Resection} \\
          & (N=14505) & (N=3490) & (N=355) & (N=339) \\
    \midrule
    \textbf{age (years)} &       &       &       &  \\
    Mean (SD) & 65.3 (9.64) & 67.6 (9.81) & 67.1 (8.67) & 67.2 (8.73) \\
    Median [Min, Max] & 66.0 [38.0, 89.0] & 68.0 [38.0,89.0] & 67.5 [43.2,88.9] &  68.3 [37.8,89.7] \\
    \textbf{tumor size (cm)} &       &       &       &  \\
    Mean (SD) & 1.52 (0.369) & 1.40 (0.389) & 1.48 (0.355) & 1.48 (0.350) \\
    Median [Min, Max] &  1.50 [0.400, 2.00] &  1.50 [0.400, 2.00] &  1.50 [0.600, 2.50] &  1.50 [0.400, 2.30] \\
    \bottomrule
    \end{tabular}
}
\end{table}

\newpage
\section*{Appendices}
\global\long\def\theequation{S\arabic{equation}}
\setcounter{equation}{0}
\global\long\def\thefigure{S\arabic{figure}}
\setcounter{figure}{0}
\global\long\def\thetheorem{S\arabic{theorem}}
\setcounter{theorem}{0}
\global\long\def\thecondition{S\arabic{condition}}
\setcounter{condition}{0}
\global\long\def\theremark{S\arabic{remark}}
\setcounter{remark}{0}
\global\long\def\theproposition{S\arabic{proposition}}
\setcounter{proposition}{0}
\global\long\def\thelemma{S\arabic{lemma}}
\setcounter{lemma}{0}
\global\long\def\theassumption{S\arabic{assumption}}
\setcounter{assumption}{0}

The Appendices are organized as the following. 
Appendix A provides the technical details and definitions of the notations mentioned in the main paper.
Appendix B presents some additional assumptions for studying the theoretical results of the proposed method.
Appendix C provides some propositions mentioned in the main paper. 
Appendix D contains some lemmas for proving the  theoretical properties. Web Appendix E provides proof of the  theoretical properties of the paper.
In Appendix F, we present the simulation results of Case 2 and Case 3 in Section 5.1 of the main paper, as well as an additional simulation study.

Let $Z_1, \ldots, Z_{n'}$ be the i.i.d.\hspace{-0.05cm} random variables with probability distribution  $\mb Q$, and $\mb Q_{n'}$ be the empirical measure of these random variables.
For a function $f$, we agree on $
\mb Q_{n'}f = (n')^{-1}\sum_{i = 1}^{n'}f(Z_i)$ and $\mb Q f = \int f(z) {\rm d}\mb Q(z)
$.
Let $\| f \|_{2} = \{ \int f^2(z){\rm d}z \}^{1/2}$ and
$\| f \|_{\infty} = \sup_{z}|f(z)|$ denote the $L_2$ norm and supremum norm of $f$, respectively.
Furthermore, for a random function $\wh{f}(z) = \wh{f}(z; Z_1, \ldots, Z_{n'})$, which is a measurable function concerning $z$ given observations $Z_1, \ldots, Z_{n'}$, we agree on  $E_r\{ \wh{f}(Z) \} = \int \wh{f}(z){\rm d}\mb Q(z) = E\{ f(Z) \}|_{f = \wh{f}}$ and ${\rm Var}_r \{ \wh{f}(Z) \} = \int \wh{f}^2(z) {\rm d}\mb Q(z) - \{ \int \wh{f}(z) {\rm d}\mb Q(z) \}^2 = {\rm Var}\{f(Z) \} |_{f = \wh{f}}$.
Let $\mb P_1$ and $\mb P_0$ denote the probability distribution of $(Y, \Delta, \bs X, A, S = 1)$ and $(Y, \Delta, \bs X, A, S = 0)$, respectively.
In the following, we use $c_j$ to denote a generic positive constant with an appropriate
subscript $j$ for $j = 0 ,1,\ldots$, and $c_j$ may represent different values in different contexts or lines.

\subsection*{Appendix A: Technical details and notations}
\subsubsection*{A.1 Sobolev space and penalties}
For a multi-index vector $k = (k_1, \ldots, k_p)$ of non-negative integers and a differentiable function $f(\bs x)$ with $\bs x = (x_1, \ldots, x_p)^{\rm T}$, we define $|k| = k_1 + \ldots + k_p$ and the corresponding $k$-th order derivative of $f(\bs x)$ as
$
f^{(k)}( \bs x) = {\partial^k f(\bs x)}/(\partial x_1^{k_1}\ldots \partial x_p^{k_p})
$.
For $m'_1 > p/2$, we define the $m'_1$-th order Sobolev space ${\cal H}_1$ as
\bse
{\cal H}_1 &=& \bigg\{ f: \Omega \rightarrow \mb R\, \big|\, f^{(j)} \textrm{ is absolutely continuous for } |j| = 0, 1, \ldots, (m'_1 - 1),
\\
&& \textrm{ and } f^{(k)} \in L_2(\Omega) \textrm{ for } |k| = m'_1\bigg\},
\ese
where $L_2(\Omega)$ is the collection of all square-integrable functions defined on $\Omega$.
Similarly, we define the $m'_0$-th order Sobolev space ${\cal H}_0$ for $m'_0 > p/2$.
For $p / 2 < m_1 \le m_1'$, $p / 2 < m_0 \le m_0'$, $k = (k_1,\ldots, k_p)$, $k' = (k'_1,\ldots, k'_p)$, $f, \wt{f} \in {\cal H}_1$, and $g, \wt{g} \in {\cal H}_0$, define
\bse
J_1(f, \wt{f}) & = & \sum_{|k| = m_1 }\frac{m_1!}{\prod_{j = 1}^p k_j!}\int_\Omega f^{(k)}(\bs x)\wt{f}^{(k)}(\bs x){\rm d}\bs x,
\\
J_0(g, \wt{g}) & = & \sum_{|k'| = m_0 }\frac{m_0!}{\prod_{j = 1}^p k'_j!} \int_\Omega g^{(k')}(\bs x)\wt{g}^{(k')}(\bs x){\rm d}\bs x.
\ese
\subsubsection*{A.2 Reproducing kernel Hilbert space}
Suppose that $n_1 / n \rightarrow \varrho$ as $n\rightarrow \infty$, where $0 < \varrho < 1$.
Let $\sigma_{0s}^2(\bs x)$ denote the probability limit of $\wh{\sigma}^2_s(\bs x)$, and $q_1(\bs x )$ and $q_0(\bs x )$ denote the density functions of $\bs X$ for the RCT and RWD, respectively.
For $f, \wt{f} \in {\cal H}_1$ and $g, \wt{g} \in {\cal H}_0$, define
\bse
V_1(f, \wt{f}) & = & \int_\Omega \left[ \varrho \left\{\sigma^{2}_{01}(\bs x)\right\}^{-1}  q_1(\bs x )
+ (1 - \varrho) \left\{\sigma^{2}_{00}(\bs x) \right\}^{-1}  q_0(\bs x) \right]f(\bs x)\wt{f}(\bs x) {\rm d}\bs x,
\\
V_0(g, \wt{g}) & = & (1 - \varrho ) \int_\Omega \left\{\sigma^{2}_{00}(\bs x)\right\}^{-1} g(\bs x)\wt{g}(\bs x)  q_0(\bs x  ){\rm d} \bs x.
\ese
Then $\ch_1$ is a reproducing kernel Hilbert space (RKHS)
endowed with the inner product $
\langle f, \wt{f}\rangle_{\ch_1}=V_1(f,\wt{f})+ \gamma_1 J_1(f,\wt{f})
$
and the norm
$
\|f\|^2_{\ch_1}=\langle f,f \rangle_{\ch_1}
$, and
$\ch_0$ is a RKHS
endowed with the inner product $
\langle g, \wt{g}\rangle_{\ch_0}=V_0(g,\wt{g})+ \gamma_0 J_0(g, \wt{g})
$
and the norm
$
\|g\|^2_{\ch_0}=\langle g,g \rangle_{\ch_0}
$.
Define the product space $\ch = \ch_1 \times \ch_0$  endowed with the inner product
$\langle \bs h, \wt{\bs h} \rangle_{\ch} = \langle f, \wt{f}\rangle_{\ch_1} + \langle g, \wt{g}\rangle_{\ch_0}$ and the norm $\| \bs h\|_{\ch}^2 = \langle \bs h, {\bs h} \rangle_{\ch} = \|f \|_{\ch_1}^2 + \| g \|_{\ch_0}^2$ for $\bs h = (f, g) \in \ch$
and $\wt{\bs h} = (\wt{f}, \wt{g}) \in \ch$.

Let $K_1(\cdot, \cdot)$ and $K_0(\cdot, \cdot)$ denote the reproducing kernels for $\ch_1$ and $\ch_0$, respectively, then $K_1(\cdot, \cdot)$ and $K_0(\cdot, \cdot)$ are known to have the properties $K_{1 \bs x}(\cdot)\equiv K_1(\bs x,\cdot) = K_1(\cdot , \bs x)\in\ch_1 $ and $\langle K_{1\bs x},f \rangle_{\ch_1} = f(\bs x)$,
and $K_{0\bs x}(\cdot)\equiv K_0(\bs x,\cdot) = K_0(\cdot , \bs x)\in\ch_0$
and $\langle K_{0\bs x},g \rangle_{\ch_0} = g(\bs x)$.
We further define  two nonnegative
definite and self-adjoint linear operators $W_{\gamma_1}$  and
$W_{\gamma_0}$ respectively defined on $\ch_1$ and $\ch_0$ such that
$\langle W_{\gamma_1} f,\wt{f}\rangle_{\ch_1} = \gamma_1 J_1(f, \wt{f})$ for $f, \wt{f} \in \ch_1$,
and $\langle W_{\gamma_0} g, \wt{g}\rangle_{\ch_0} = \gamma_0 J_0(g, \wt{g})$
for $g, \wt{g}\in \ch_0$.
Then $\langle f , \wt{f} \rangle_{\ch_1} = V_1(f, \wt{f}) + \langle W_{\gamma_1} f,\wt{f}\rangle_{\ch_1}$ and $\langle g , \wt{g} \rangle_{\ch_0} = V_0(g, \wt{g}) + \langle W_{\gamma_0} g,\wt{g}\rangle_{\ch_0}$.

\subsubsection{A.3 Notations}
Let
$\dot{\wh{\bs U}}_n(\tau, \lambda)$ denote the second order Fr\'{e}chet derivative operator of $\wh{\ell}_{n}(\tau, \lambda)$
with respective to $(\tau, \lambda)$, $\dot{\bs U}(\tau, \lambda)$ denote the probability limit of $\dot{\wh{\bs U}}_n(\tau, \lambda)$, and $\dot{\bs W}_{\bs \gamma}$ denote the linear operator map $\ch$ to $\ch$ such that $\dot{\bs W}_{\bs \gamma}[\bs h] = (W_{\gamma_1}f, W_{\gamma_0}g )$ for $\bs h = (f, g) \in \ch$.
Let $\dot{\bs U}_{\bs \gamma}(\tau, \lambda) = \dot{\bs U}(\tau, \lambda) + \dot{\bs W}_{\bs \gamma}$
and $\dot{\bs U}_{\bs \gamma}^{-1}(\tau, \lambda)$ denote the inverse operator of $\dot{\bs U}_{\bs \gamma}(\tau, \lambda)$. Furthermore, let
$
\tau_0^* = \tau_0 - b_\tau, \lambda_0^* = \lambda_0 - b_\lambda
$ with
$(b_\tau, b_\lambda) = \dot{\bs U}^{-1}_{\bs \gamma}(\tau_0, \lambda_0) [ (W_{\gamma_1}\tau_0, W_{\gamma_0}\lambda_0)]$, and $\bs \Sigma$ be the covariance matrix such that, for every $\bs \omega = (\omega_1, \omega_0 )^{\rm T} \in \mb R^2$,
\bse
\bs \omega^{\rm T}\bs \Sigma \bs \omega =  \lim_{n\rightarrow \infty}
\left\{ V_1(K_1^*, K_1^*) + V_0( K_0^*, K_0^* ) + 2(1 - \varrho)
\int_{\Omega}\left\{\sigma^{2}_{00}(\bs x) \right\}^{-1} K_1^*(\bs x )K_0^*(\bs x )q_0(\bs x){\rm d}\bs x  \right\}
\ese
with $(K_1^*, K_0^*) = \dot{\bs U}^{-1}_{\bs \gamma}(\tau_0, \lambda_0)[(\omega_1 \gamma_1^{p/(4m_1)} K_{1\bs x_0}, \omega_0 \gamma_0^{p/(4m_0)} K_{0\bs x_0})]$.

\subsection*{Appendix B: Conditions}
\begin{condition}\label{con1}
$\phi_1(\cdot), \ldots, \phi_{r_1}(\cdot)$ belong to $\ch_1$, and $\psi_1(\cdot), \ldots, \psi_{r_0}(\cdot)$ belong to $\ch_0$.
\end{condition}

\begin{condition}\label{con2}
For $f\in\ch_1$, $g\in\ch_0$,
there exist some $f_n(\cdot) \in \Phi_n$, $g_n(\cdot) \in \Psi_n$, $\kappa_1 > \kappa'_1 \ge 0 $, and $\kappa_0 > \kappa'_0 \ge 0 $ such that
\bse
\left\|f_n - f \right\|_{\infty} & = & O(n^{-\kappa_1}), \quad \sup_{|k| = m_1}\left\|f^{(k)}_n - f^{(k)} \right\|_{\infty} = O(n^{-\kappa'_1}),
\\
\left\|g_n - g \right\|_{\infty} & = & O(n^{-\kappa_0}) , \quad \sup_{|k'| = m_0}\left\|g^{(k')}_n - g^{(k')} \right\|_{\infty} = O(n^{-\kappa'_0}).
\ese
\end{condition}

\begin{condition}\label{con3}
$\wh{D}$ is uniformly bounded and belongs to a class of functions $\ms D$, and
$\{\wh{\sigma}^2_s(\bs x) \}^{-1}$  is uniformly bounded away from zero and
belongs to a class of  functions $\ms A$, such that
$\max\{ \log N(\epsilon, {\ms D}, \|\cdot \|_{\infty}), \log N(\epsilon, {\ms A}, \|\cdot \|_{\infty})\} < c_0\epsilon^{-\upsilon}$ for some constants $c_0 > 0 $ and $0 < \upsilon < 2$, where
$N(\epsilon, {\ms D}, \| \cdot \|_\infty )$ and  $N(\epsilon, {\ms A}, \| \cdot \|_\infty )$ are covering numbers of ${\ms D}$ and ${\ms A}$ under supremum norm, and $s= 0 , 1$.
\end{condition}

\begin{condition}\label{con4}
There exist some constants $0 < c_1 < 1$ and $c_2 > 0 $ such that
$c_1 \le \wh{e}(X) \le 1- c_1$,
$\wh{G}_T(L \mid \bs X , A, S = 1) \ge c_2$,
and
$\wh{G}_C(L \mid \bs X , A, S = 1) \ge c_2$.
\end{condition}

\begin{condition}\label{con5}
$c_3^{-1} \le q_s(\bs x) \le c_3$, where $s = 0 ,1$.
\end{condition}

\begin{condition}\label{con6}
$\sum_{a = 0}^1 \| \Theta_a \|_{2} + \|\wh{e} - e \|_2 \times \sum_{a = 0 }^1 \| \wh{\mu}_a - \mu_a\|_2 = o_P(1)$,
where
\bse
\Theta_A( \bs X ) & = & \int_0^L \Bigg[G_T(t \mid \bs X, A , S = 1)
\int_0^t \frac{\wh{G}_T( u \mid \bs X, A , S = 1) - {G}_T( u \mid \bs X, A , S = 1)}{\wh{G}_T( u \mid \bs X, A , S = 1)}
\\
&& \times
{\rm d} \left\{\frac{{G}_C( u \mid \bs X, A , S = 1) - \wh{G}_C( u \mid \bs X, A , S = 1)}{\wh{G}_C( u \mid \bs X, A , S = 1)} \right\}\Bigg] {\rm d}t.
\ese
\end{condition}

\begin{condition}\label{con7}
$\sum_{a = 0}^1 \| \Theta_a \|_{2} + \|\wh{e} - e \|_2 \times \sum_{a = 0 }^1 \| \wh{\mu}_a - \mu_a\|_2
= O_P(\delta_n)$, where $\delta_n = n^{-1/2}\gamma_1^{-p/(4m_1)} + n^{-1/2}\gamma_0^{-p/(4m_0)} + \gamma_1^{1/2} + \gamma_0^{1/2}$.
\end{condition}

\begin{condition}\label{con8}
$n^{-1/2}\{\gamma_s^{(-2m_s\upsilon - 4p + p\upsilon)/(8m_s)} + \gamma_s^{-(6m_s - p)p/(8m_s^2)} \} = o(1)$, where $s = 0, 1$.
\end{condition}

\begin{condition}\label{con9}
$(n^{-\kappa_1} + n^{-\kappa_0} ) = O(\delta_n)$.
\end{condition}

\begin{condition}\label{con10}
$\sum_{a = 0}^1 \| \Theta_a \|_{2} + \|\wh{e} - e \|_2 \times \sum_{a = 0 }^1 \| \wh{\mu}_a - \mu_a\|_2
= o_P(n^{-1/2})$.
\end{condition}

\begin{condition}\label{con11}
$\delta_n(\gamma_1^{-p/(4m_1)} + \gamma_0^{-p / (4m_0)} ) = o(1)$.
\end{condition}

\begin{condition}\label{con12}
$n^{1/2} (n^{-\kappa_1}\gamma_1^{-1/2} + n^{-\kappa_0}\gamma_0^{-1/2} )\delta_n = o(1)$
and $ n^{1/2}(n^{-\kappa'_1}+ n^{-\kappa'_0}) \delta_n  = O(1)$.
\end{condition}

\begin{condition}\label{con13}
$\| \wh{\sigma}^2_s - \sigma^2_{0s} \|_{\infty} = o_P(1)$, where $s = 0 , 1$.
\end{condition}

\begin{condition}\label{con14}
$\| \wh{D} - D_0 \|_{\infty} = o_P(1)$, where $D_0$ is the probability limit of $\wh{D}$.
\end{condition}

\begin{condition}\label{con15}
$\wh{\sigma}^2(\bs X, S) - {\rm Var}_r(\wh{D} \mid \bs X, S) = o_P(1)$.
\end{condition}

\begin{condition}\label{con16}
There exist some positive constants $\kappa^*$ and $\kappa^\dag$ such that
$ \rho( \bs A^{\rm T}  \bs A / n ) \asymp n^{-\kappa^*}$ and $\rho_{\max}( \bs P) = O( n^{\kappa^\dag} ) $ almost surely, where $\bs A$ and $\bs P$ are defined in the first paragraph of Subsection \ref{subsec4}, $\rho( \bs A^{\rm T}  \bs A / n )$ is the eigenvalues of $\bs A^{\rm T}  \bs A / n$
and $\rho_{\max}( \bs P)$ is the maximum eigenvalue of $\bs P$.
\end{condition}

\begin{condition}\label{con17}
$(n^{-\kappa_1} + n^{-\kappa_0}) = o(n^{-1/2})$ and $n^{\kappa^* + \kappa^\dag}( \gamma_1 + \gamma_0 ) = o(n^{-1/2})$.
\end{condition}

Condition \ref{con1} ensures that the estimators $\wh{\tau}_n$ and $\wh{\lambda}_n$ are within the subspaces of $\ch_1$ and $\ch_0$, respectively.
Since the theoretical properties are established within the framework of reproducing kernel Hilbert space, this condition is indispensable.
However, this condition is mild and easily satisfied.
For example,
If B-splines or power series with appropriate orders, or reproducing kernels with appropriate smoothness, are selected as the sieve basis, Condition \ref{con1} would be met.
Condition \ref{con2} requires uniform approximation rates to the function
and its derivatives.  If one chooses the reproducing kernel approximation with some appropriate kernel such as Gaussian kernel, and chooses all the observation data as the knots, then Condition \ref{con2} will satisfied with $\kappa_1 = m_1' / p$, $\kappa_1' = (m_1' - m_1) / p$, $\kappa_0 = m_0' / p$ and $\kappa_0' = (m_0' - m_0) / p$.
Additionally,
If $f$ and $g$ have continuous $m_1'$-th and $m_0'$-th order derivatives, respectively, and further
choose B-splines or power series as the sieve basis, and in each dimension, the maximum length between the spline knots is dominated by the inverse of the number of spline knots, then
$\|f_n - f\|_{\infty} = O(r_1^{-m_1' / p})$, $\sup_{|k| = m_1}\|f_n^{(k)} - f^{(k)}\|_{\infty} = O(r_1^{-( m_1' - m_1 ) / p})$, $\|g_n - g\|_{\infty} = O(r_0^{-m_0' / p})$,
and $\sup_{|k'| = m_0}\|g_n^{(k')} - g^{(k')}\|_{\infty} = O(r_0^{-( m_0' - m_0 ) / p})$. Moreover, if $r_1 \asymp n^{\nu_1}$ and $r_0 \asymp n^{\nu_0}$, then this condition is satisfied with $\kappa_1 = m_1'\nu_1 / p $, $\kappa_0 = m_0'\nu_0 / p $,
$\kappa'_1 = ( m_1' - m_1 )\nu_1 / p $, and $\kappa'_0 = (m_0' - m_0 )\nu_0 / p $.
Such a condition appears in Newey (1997).
Conditions \ref{con3}--\ref{con5} are technical.
Condition \ref{con3} is a regularization condition and controls the complexities of $\wh{D}$ and $\wh{\sigma}^2_s(\bs x)$.
Especially, Condition \ref{con3} implies that $\wh{D}$ and $\wh{\sigma}^2_s(\bs x)$ belong to Donsker classes.
Condition \ref{con4} ensures $\wh{D}$ to be well defined. It can be easily satisfied  for commonly used estimation methods such as some survival models for ${G}_T(t \mid \bs X , A, S = 1)$ and ${G}_C(t \mid \bs X , A, S = 1)$.
Condition \ref{con5} constitutes a vital prerequisite for Lemma \ref{lemma4}.
Similar or equivalent conditions are also imposed in Cox (1984), Cox (1988), O'Sullivan (1993), and Gu (2013).
Under Condition \ref{con5}, the square of $L_2$ norm
is equivalent to $V_1$ and $V_0$, that is, $c_4^{-1}\| f \|_{2}^2 \le V_1(f,f) \le c_4 \| f\|_{2}^2$ and
$c_5^{-1}\| g \|_{2}^2 \le V_0(g,g) \le c_5 \| g\|_{2}^2$ for $f\in \ch_1$ and $g \in \ch_0$.

The remaining conditions are the specific technical conditions required for the theoretical properties in the main paper.
Condition \ref{con6} is used to derive the robustness of the proposed estimators, as shown in Theorem 1. Conditions \ref{con7}--\ref{con9} are used for studying the convergence rate of the proposed estimator.
The conditions in Remark 1 which presented in the main paper imply that
$
 n^{-1/2} \{ \gamma_s^{(-2m_s\upsilon - 4p + p\upsilon)/(8m_s)} + \gamma_s^{-(6m_s - p)p/(8m_s^2)} \}
 \asymp  \{n^{-{(2m - p )(2 - v)}/ (8m + 4p)} + n^{-{(2m - p)^2}/{(8m^2 + 4mp )}} \}
 =  o(1)
$
by $m > p / 2$ and $0 < \upsilon < 2$. Furthermore, if $r_1 \asymp r_0 \asymp n^{\nu}$, then Condition \ref{con9} can be satisfied by taking $\nu = p / (2m + p) < 1 / 2 $. Thus, Conditions \ref{con8}--\ref{con9} remain valid under the conditions described in Remark 1 which presented in the main paper.
Conditions \ref{con10}--\ref{con14} are imposed for deriving the asymptotic normality.
Condition \ref{con10} is imposed to ignore the error caused by the nuisance functions.
By Theorem 1, various combinations of the convergence rates of $\wh{G}_T$ and $\wh{G}_C$ fulfill this condition.
For instance, if parametric models are selected to fit $G_T$ and $G_C$, and they are correctly specified, then $\sum_{a = 0}^1 \| \Theta_a \|_{2} = o_P(n^{-1})$.
Alternatively, if fully nonparametric estimation is employed for $G_T$ and $G_C$, Condition \ref{con10} would be satisfied under certain smoothness conditions for $G_T$ and $G_C$.
For the conditions outlined in Remark 1 which is presented in the main paper,
Condition \ref{con11} is automatically satisfied since $m > p / 2$.
Furthermore, if $m_1' = m_0' = m'$, $r_1 \asymp r_0 \asymp n^\nu$, then
Condition \ref{con12} is fulfilled when $\max\{ p / (2 m'), p^2 / \{(m' - m )(4m + 2p)\} \} < \nu < 1$.
Conditions \ref{con13}--\ref{con14} enhance pointwise convergence to uniform convergence, but they are rather mild conditions and often easily satisfied by common estimation methods. Uniform convergence is a technical enhancement. It does not significantly limit the applicability of the estimators, rather, it is a reasonable and common condition in both theoretical and practical contexts, especially in the field of survival analysis. Interestingly, here, $\wh{\sigma}^2_s(\bs X)$ and $\wh{D}$, as estimators of ${\rm Var}({D} \mid \bs X, S = s)$ and $D$, do not need to be consistent.
Conditions \ref{con15}--\ref{con17} are used for studying the efficiency of the proposed estimator.
Choosing B-splines as the sieve basis, with the maximum length between spline knots in each dimension dominated by the inverse of the number of knots, ensures that Condition \ref{con16} holds with $n^{-\kappa^*} \asymp (r_1 + r_0)^{-1}$ and $ n^{\kappa^\dag} \asymp ( r_1^3 + r_0^3 ) $.
Furthermore, if $r_1 \asymp r_0 \asymp n^\nu$, $\gamma_1 \asymp \gamma_0 \asymp n^{-\varsigma}$, and $m_1' = m_0' = m'$, then Condition \ref{con17} holds under the condition $p/(2m') < \nu < (2\varsigma - 1) / 8 $. This condition implies that the true functions $\tau_0$ and $\lambda_0$ should exhibit high smoothness, the number of spline knots should be moderate, and the penalized parameters should be sufficiently small.

\subsection*{Appendix C: Propositions}

\begin{proposition}\label{proposition1}
Suppose that Assumptions 1--2 in the main paper hold, then
\bse
E( R \mid \bs X , S = 1 ) = \mu_1(\bs X) - \mu_0(\bs X) = \tau(\bs X).
\ese
\end{proposition}
\begin{proof}
It follows that
\bse
E(R_1 \mid A = 1, \bs X ,S = 1) &= & E\{T(1)\wedge L \mid A = 1, \bs X, S = 1 \}\frac{1}{e(\bs X)} - \frac{1 - e(\bs X)}{e(\bs X)}\mu_1(\bs X)
\\
&=&
\mu_1(\bs X),
\\
E({R}_1 \mid A = 0, \bs X ,S = 1) &= & E\{\mu_1(\bs X) \mid A = 0, \bs X, S = 1 \} =
\mu_1(\bs X),
\ese
which implies that
$
E({R}_1 \mid \bs X ,S = 1 ) = \mu_1(\bs X )
$.
Similarly, we have
$
E({R}_0 \mid \bs X , S = 1) = \mu_0(\bs X)
$.
Combining Assumption 1 in the main paper, we conclude that
\bse
E(R \mid \bs X , S = 1 ) &=& \mu_1(\bs X) - \mu_0(\bs X)
\\
&=& E\{T(1) \wedge L \mid A = 1, \bs X , S = 1\} - E\{T(0) \wedge L \mid A = 0, \bs X , S = 1\}
\\
& = & E\{T(1) \wedge L \mid \bs X , S = 1\} - E\{T(0) \wedge L \mid \bs X , S = 1\}
\\
&  = & E\{ T(1) \wedge L - T(0) \wedge L \mid \bs X , S = 1 \}
\\
&= & \tau(\bs X).
\ese
Hence, we prove Proposition \ref{proposition1}.
\end{proof}

\begin{proposition}\label{proposition2}
Suppose that Assumption 2 in the main paper holds, then
\bse
E\left( \wt{T}_L \mid \bs X , A , S = 1 \right) =
E\left( {T}_L \mid \bs X , A , S = 1 \right).
\ese
\end{proposition}

\begin{proof}
By Assumption 2 in the main paper, 
\bsee\label{p2_eq0}
&&E\left\{ \left. {Y_L \wt{\Delta}}{G^{-1}_C( Y_L \mid \bs X , A, S = 1)} \,\right|\,  \bs X , A , S = 1 \right\}
\nonumber
\\
&= &
E\left[E\left\{ { Y_L \wt{\Delta}}{G^{-1}_C\left( Y_L \mid \bs X , A, S = 1 \right) } \,\bigg|\, T_L, \bs X , A , S = 1 \right\} \,\bigg|\, \bs X , A, S =1 \right]
\nonumber
\\
&=&
E\left[E\left\{ {T_L\wt{\Delta}}{G^{-1}_C\left(T_L \mid \bs X , A, S = 1 \right) } \,\bigg|\, T_L, \bs X , A , S = 1 \right\} \bigg| \bs X , A, S = 1 \right]
\nonumber
\\
& = & E\left\{ {T_L }{G^{-1}_C\left(T_L \mid \bs X , A, S = 1 \right)}P\left(\wt{\Delta} = 1 \mid T_L, \bs X , A , S = 1 \right) \,\bigg|\, \bs X , A, S = 1 \right\}
\nonumber
\\
& = & E\left\{ {T_L }{G^{-1}_C\left(T_L \mid \bs X , A , S = 1\right)}P(C \ge T_L \mid T_L, \bs X , A , S = 1 ) \,\bigg|\, \bs X , A, S = 1 \right\}
\nonumber
\\
& = & E(T_L \mid \bs X , A , S = 1).
\esee
Thus, it is enough to show that
\bsee\label{p2_eq1}
E \left\{ \left.\int_0^L  B(t)G^{-1}_C(t \mid \bs X , A, S = 1){\rm d}M_C(t)
 \,\right|\, \bs X, A, S = 1  \right\} =0.
\esee
Following Assumption 2 in the main paper, we have
\bse
E\{ N_C(t) \mid \bs X, A, S = 1 \} & = &E \big[ E\{ N_C(t) \mid C , \bs X, A, S = 1 \} \mid \bs X , A , S = 1 \big]
\\
& = & E_C \{ I(C \le t ) P(C \le T_L  \mid C , \bs X, A, S = 1 ) \mid \bs X, A, S = 1\}
\\
& = & -\int_0^t G_T(u \mid \bs X , A, S = 1){\rm d}G_C(u \mid \bs X , A, S = 1).
\ese
Thus
\bse
&& E \left\{ \left.\int_0^L B(t)G^{-1}_C(t \mid \bs X , A, S = 1) {\rm d} N_C(t) \,\right|\, \bs X, A, S = 1  \right\}
\\
& = & -\int_0^L B(t)G_T(t \mid \bs X , A, S = 1){G^{-1}_C(t \mid \bs X , A, S = 1)}{ {\rm d}G_C(t \mid \bs X , A, S = 1) } .
\ese
On the other hand,
\bse
&&E \left\{ \left.\int_0^L B(t) {G^{-1}_C(t \mid \bs X , A, S = 1)}{\rm d}Q_C(t) \,\right|\, \bs X, A, S = 1  \right\}
\\
& = &E \left\{ \left.\int_0^L B(t) {G^{-2}_C(t \mid \bs X , A, S = 1)} I(Y_L \ge t )  { {\rm d} G_C(t \mid \bs X , A, S = 1)} \,\right|\, \bs X, A, S = 1  \right\}
\\
& = & \int_0^L B(t) {G^{-2}_C(t \mid \bs X , A, S = 1)} P(Y \ge t \mid \bs X , A , S = 1) { {\rm d} G_C(t \mid \bs X , A, S = 1)}
\\
& = & \int_0^L B(t) G_T( t \mid \bs X , A , S = 1) {G^{-1}_C(t \mid \bs X , A, S = 1)} { {\rm d} G_C(t \mid \bs X , A, S = 1)},
\ese
which implies that equation (\ref{p2_eq1}) holds. Thus we prove Proposition \ref{proposition2}.
\end{proof}

\subsection*{Appendix D: Lemmas}

\begin{lemma}\label{lemma1}
Suppose that Assumptions 1--2 in the main paper hold, then
\bse
E_r \left(\wh{T}_L - \wt{T}_L \mid \bs X, A , S = 1 \right) = \Theta_A(\bs X ),
\ese
where
\bse
\wh{T}_L = {Y_L \wt{\Delta}}{\wh{G}^{-1}_C( Y_L \mid \bs X , A, S = 1)} -
\int_0^L \wh{B}(t)\wh{G}^{-1}_C(t \mid \bs X , A, S = 1){\rm d}\wh{M}_C(t),
\ese
is the corresponding estimator of $\wt{T}_L$,
$\wh{B}(t) = t + { \int_t^L \wh{G}_T(u \mid \bs X, A , S = 1 ){\rm d}u } / {\wh{G}_T(t \mid \bs X, A , S = 1 )}$ and
$\wh{M}_C(t) = N_C(t) + \wh{Q}_C(t)$ with $\wh{Q}_C(t) = \int_0^t I(Y_L \ge u ) \wh{G}^{-1}_C(u \mid \bs X, A, S = 1){\rm d}\wh{G}_C(u \mid \bs X, A, S = 1)$,
and
\bse
\Theta_A(\bs X ) & = & \int_0^L \Bigg[G_T(t \mid \bs X, A , S = 1)
\int_0^t \frac{\wh{G}_T( u \mid \bs X, A , S = 1) - {G}_T( u \mid \bs X, A , S = 1)}{\wh{G}_T( u \mid \bs X, A , S = 1)}
\\
&& \times
{\rm d} \left\{\frac{{G}_C( u \mid \bs X, A , S = 1) - \wh{G}_C( u \mid \bs X, A , S = 1)}{\wh{G}_C( u \mid \bs X, A , S = 1)} \right\}\Bigg] {\rm d}t.
\ese
\end{lemma}
\begin{proof}
By the definition of $E_r$ and Proposition \ref{proposition2},
\bsee\label{lemma1_eq0}
E_r \left(\wt{T}_L \mid \bs X, A , S = 1 \right) =
E\left(\wt{T}_L \mid \bs X, A , S = 1 \right) & = & E\left(T_L \mid \bs X, A , S = 1 \right)
\nonumber
\\
& = & \int_0^L G_T(t \mid \bs X ,A , S = 1){\rm d}t.
\esee
Mimicking the proof of equation (\ref{p2_eq0}), we have
\bsee\label{lemma1_eq1}
&& E_r \left\{ \left. {Y_L \wt{\Delta}}{\wh{G}^{-1}_C( Y_L \mid \bs X , A, S = 1)} \,\right|\,  \bs X , A , S = 1 \right\}
\nonumber
\\
& = &E_r \left\{ \left. T_L{G_C( T_L \mid \bs X , A, S = 1)}{\wh{G}^{-1}_C( T_L \mid \bs X , A, S = 1)} \,\right|\,  \bs X , A , S = 1 \right\}
\nonumber
\\
&= &\int_0^L G_T(t \mid \bs X ,A , S = 1)G_C(t \mid \bs X ,A , S = 1)\wh{G}^{-1}_C(t \mid \bs X ,A , S = 1){\rm d}t
\nonumber
\\
&& + \int_0^L t G_T(t \mid \bs X ,A , S = 1)
{\rm d}\left\{ G_C(t \mid \bs X ,A , S = 1)\wh{G}^{-1}_C(t \mid \bs X ,A , S = 1) \right\}.
\esee
Furthermore,
\bsee\label{lemma1_eq2}
&& E_r \left\{ \int_0^L \wh{B}(t)\wh{G}^{-1}_C(t \mid \bs X , A, S = 1){\rm d}\wh{M}_C(t) \, \big | \, \bs X , A, S = 1 \right\}
\nonumber
\\
&= & - E_r \left\{ \int_0^L \wh{B}(t)\wh{G}^{-1}_C(t \mid \bs X , A, S = 1){\rm d}{Q}_C(t) \, \big | \, \bs X , A, S = 1 \right\}
\nonumber
\\
&& + E_r \left\{ \int_0^L \wh{B}(t)\wh{G}^{-1}_C(t \mid \bs X , A, S = 1){\rm d}\wh{Q}_C(t) \, \big | \, \bs X , A, S = 1 \right\}
\nonumber
\\
&= & -E_r \bigg\{ \int_0^L \wh{B}(t)\wh{G}^{-1}_C(t \mid \bs X , A, S = 1){G}^{-1}_C(t \mid \bs X , A, S = 1)
\nonumber
\\
&& \times I(Y_L \ge t){\rm d}{G}_C(t \mid \bs X , A, S = 1) \, \big | \, \bs X , A, S = 1 \bigg\}
\nonumber
\\
&& + E_r \bigg\{ \int_0^L \wh{B}(t)\wh{G}^{-1}_C(t \mid \bs X , A, S = 1)\wh{G}^{-1}_C(t \mid \bs X , A, S = 1)
\nonumber
\\
&& \times I(Y_L \ge t){\rm d}\wh{G}_C(t \mid \bs X , A, S = 1) \, \big | \, \bs X , A, S = 1 \bigg\}
\nonumber
\\
& = & \int_0^L \wh{B}(t)G_T(t \mid \bs X , A, S = 1){\rm d}\left\{{G}_C(t \mid \bs X , A, S = 1) \wh{G}^{-1}_C(t \mid \bs X , A, S = 1)\right\}
\nonumber
\\
& = & \int_0^L t G_T(t \mid \bs X ,A , S = 1)
{\rm d}\left\{ G_C(t \mid \bs X ,A , S = 1)\wh{G}^{-1}_C(t \mid \bs X ,A , S = 1) \right\}
\nonumber
\\
&& + \int_0^L G_T(t \mid \bs X, A , S = 1) \bigg[ \int_0^t G_T(u \mid \bs X, A , S = 1)
\nonumber
\\
&& \times \wh{G}^{-1}_T(u \mid \bs X, A , S = 1) {\rm d}\left\{ G_C(u \mid \bs X, A , S = 1)\wh{G}^{-1}_C(u \mid \bs X, A , S = 1) \right\}  \bigg] {\rm d} t,
\esee
by using  the Fubini Theorem.

Combining equations (\ref{lemma1_eq0})--(\ref{lemma1_eq2}), we have
\bse
&&E_r \left(\wh{T}_L - \wt{T}_L \mid \bs X, A , S = 1 \right)
\\
& = & \int_0^L G_T(t \mid \bs X, A , S = 1) \Bigg[ \frac{{G}_C( t \mid \bs X, A , S = 1) - \wh{G}_C( t \mid \bs X, A , S = 1)}{\wh{G}_C( t \mid \bs X, A , S = 1)}
\\
&& - \int_0^t \frac{G_T( u \mid \bs X, A , S = 1)}{\wh{G}_T( u \mid \bs X, A , S = 1)}
{\rm d}\left\{ \frac{G_C( u \mid \bs X, A , S = 1)}{\wh{G}_C( u \mid \bs X, A , S = 1)} \right\}
\Bigg]{\rm d}t
\\
&=  & \int_0^L G_T(t \mid \bs X, A , S = 1) \Bigg[
\int_0^t \frac{\wh{G}_T( u \mid \bs X, A , S = 1) - {G}_T( u \mid \bs X, A , S = 1)}{\wh{G}_T( u \mid \bs X, A , S = 1)}
\\
&& \times
{\rm d} \left\{\frac{{G}_C( u \mid \bs X, A , S = 1) - \wh{G}_C( u \mid \bs X, A , S = 1)}{\wh{G}_C( u \mid \bs X, A , S = 1)} \right\}\Bigg] {\rm d}t
\\
& = & \Theta_A(\bs X).
\ese
Thus, we complete the proof of Lemma \ref{lemma1}.
\end{proof}

\begin{lemma}\label{lemma2}
Suppose that Assumptions 1--2 in the main paper hold, then
\bse
E_r \left(\wh{D} - D \mid \bs X , S = 1 \right) = \Gamma_1(\bs X) + \Gamma_2(\bs X),
\ese
where
\bse
\Gamma_1(\bs X) & = & \sum_{a = 0}^1 \frac{\left\{ \wh{e}(\bs X ) - e(\bs X ) \right\}\left\{ \wh{\mu}_a(\bs X ) - \mu_a(\bs X) \right\}}{a\wh{e}(\bs X ) + (1 - a)\left\{ 1 - \wh{e}(\bs X )\right\}},
\\
\Gamma_2(\bs X) & = &
\sum_{a = 0}^1 \frac{a{e}(\bs X ) + (a - 1)\{ 1 - {e}(\bs X )\}}{a\wh{e}(\bs X ) + (1 - a)\left\{ 1 - \wh{e}(\bs X )\right\}}\Theta_a(\bs X).
\ese
\end{lemma}
\begin{proof}
For $S = 1$, we have $\wh{D} - D = \wh{R} - \wt{R}$.
Rewrite $\wh{R}$ as
\bse
\wh{R} = \frac{A - \wh{e}(\bs X)}{\wh{e}(\bs X )\{1 - \wh{e}(\bs X)\}}
\left\{\wh{T}_L - \wh{\mu}_A(\bs X) \right\} + \wh{\mu}_1(\bs X) - \wh{\mu}_0(\bs X),
\ese
where
$\wh{\mu}_A(\bs X) = A\,\wh{\mu}_1(\bs X) + (1 - A)\,\wh{\mu}_0(\bs X)$.
Using Proposition \ref{proposition2} and Lemma \ref{lemma1}, we have
\bse
&& E_r \left(\wh{R}  \mid \bs X , S = 1 \right)
\\
&=& E_r \left[\frac{A - \wh{e}(\bs X)}{\wh{e}(\bs X)\{1 - \wh{e}(\bs X)\}}
\left\{\wh{T}_L - \wh{\mu}_A(\bs X) \right\} + \wh{\mu}_1(\bs X ) - \wh{\mu}_0(\bs X) \,\bigg|\,\bs X , S = 1 \right]
\\
&& + E_r \left[\frac{A - \wh{e}(\bs X )}{\wh{e}(\bs X )\{1 - \wh{e}(\bs X )\}}
\left(\wh{T}_L - \wt{T}_L \right) \,\bigg|\, \bs X , S = 1
\right]
\\
& = & E_r \left[\frac{A - \wh{e}(\bs X )}{\wh{e}(\bs X)\{1 - \wh{e}(\bs X)\}}
\left\{\wt{T}_L - \wh{\mu}_A(\bs X) \right\} \,\bigg|\,\bs X , S = 1 \right] + \wh{\mu}_1(\bs X) - \wh{\mu}_0(\bs X)
\\
&& + E_r \left( E_r \left[\frac{A - \wh{e}(\bs X)}{\wh{e}(\bs X)\{1 - \wh{e}(\bs X)\}}
\left(\wh{T}_L - \wt{T}_L \right)\, \bigg|\, \bs X , A , S = 1
\right] \,\bigg|\, \bs X , S = 1  \right) ,
\\
& = & \frac{e(\bs X )}{\wh{e}(\bs X )}\left\{\mu_1(\bs X) - \wh{\mu}_1(\bs X)\right\} - \frac{1 - e(\bs X)}{1 - \wh{e}(\bs X )}\left\{\mu_0(\bs X) - \wh{\mu}_0(\bs X)\right\}
\\
&& + \wh{\mu}_1(\bs X) - \wh{\mu}_0(\bs X) + \Gamma_2(\bs X),
\ese
combined with Proposition \ref{proposition1},
\bse
&& E_r \left(\wh{R} - \wt{R} \mid \bs X, S = 1 \right)
\\
&=& \frac{\wh{e}(\bs X) - e(\bs X)}{\wh{e}(\bs X)}\left\{\wh{\mu}_1(\bs X) - \mu_1(\bs X) \right\}
+ \frac{\wh{e}(\bs X) - e(\bs X)}{1 - \wh{e}(\bs X)}\left\{\wh{\mu}_0(\bs X)-\mu_0(\bs X) \right\} + \Gamma_2(\bs X)
\\
&= & \Gamma_1(\bs X ) +  \Gamma_2(\bs X),
\ese
which proves Lemma \ref{lemma2}.
\end{proof}

\begin{lemma}\label{lemma3}
Let $\pi_1( \bs x)$ and $\pi_0(\bs x)$ be some functions which are uniformly bounded away from zero, $\varpi$ be some positive constant such that $0 < \varpi < 1$, and
\bse
\ell^\dag(f, g) = \varpi \int_\Omega \pi_1(\bs x) f^2(\bs x) q_1(\bs x){\rm d} \bs x + (1 - \varpi) \int_\Omega \pi_0(\bs x)\{ f(\bs x) + g(\bs x) \}^2 q_0(\bs x){\rm d}\bs x,
\ese
where $(f, g) \in \ch$. Suppose that Condition \ref{con5} holds,
then there exists some positive constant $c_0$ such that
\bse
\ell^\dag(f, g) \ge c_0\{V_1(f ,f) + V_0(g, g) \},
\ese
for all $(f ,g) \in \ch$.
\end{lemma}

\begin{proof}
It follows from Condition \ref{con5} and H\"{o}lder inequality that
\bse
\ell^\dag(f, g) & = & \varpi \int_\Omega \pi_1(\bs x) f^2(\bs x) q_1(\bs x){\rm d} \bs x + (1 - \varpi) \int_\Omega \pi_0(\bs x)\{ f(\bs x) + g(\bs x) \}^2 q_0(\bs x) {\rm d}\bs x
\\
& \ge & c_1 \left\{ \int_\Omega f^2(\bs x) {\rm d} \bs x + \int_\Omega \{ f(\bs x) + g(\bs x) \}^2 {\rm d}\bs x \right\}
\\
& \ge & c_1 \left[ 2\int_\Omega f^2(\bs x) {\rm d} \bs x + \int_\Omega g^2(\bs x) {\rm d} \bs x  -2 \left\{\int_\Omega f^2(\bs x) {\rm d} \bs x\right\}^{1/2}  \left\{\int_\Omega g^2(\bs x) {\rm d} \bs x\right\}^{1/2} \right],
\ese
and the equality in the last inequality holds if and only if there exists a non-negative function $\varsigma(\bs x)$ which is constant almost everywhere in Lebesgue measure such that $g(\bs x) = - \varsigma(\bs x) f(\bs x) = -\varsigma f(\bs x) $.

For $\bs h = (f ,g) \in \ch$ and $\wt{\bs h} = (\wt{f}, \wt{g}) \in \ch$,
define the $L_2$ norm on $\ch$ as $\| \bs h \|_{2} = \left( \| f \|_{2}^2 + \| g \|_{2}^2 \right)^{1/2}$ and a bilinear functional $B^\dag: \ch \times \ch \rightarrow \mb R$ as
\bse
B^\dag[\bs h][\wt{\bs h}] & = & \varpi \int_\Omega \pi_1(\bs x) f(\bs x)\wt{f}(\bs x) q_1(\bs x) {\rm d} \bs x
\\
&& + (1 - \varpi) \int_\Omega \pi_0(\bs x)\{ f(\bs x) + g(\bs x) \}\{  \wt{f}(\bs x) + \wt{g}(\bs x) \} q_0(\bs x){\rm d}\bs x.
\ese
Obviously, $\ell^\dag(f, g) = B^\dag[\bs h][\bs h]$. Then
\bse
\inf_{\| \bs h\|_{2} = 1} B^\dag[\bs h][{\bs h}]
& = & \inf_{\| f\|_{2}^2 + \| g \|_{2}^2 = 1} \ell^\dag(f ,g)
\\
& \ge &  \inf_{\| f\|_{2}^2 + \| g \|_{2}^2 = 1, g = -\varsigma f} c_1  \left\{ \int_\Omega f^2(\bs x) {\rm d} \bs x + \int_\Omega \{ f(\bs x) + g(\bs x) \}^2 {\rm d}\bs x \right\}
\\
& = & c_1 \inf_{\varsigma > 0 } \left( 1 + \frac{1 - 2 \varsigma}{1 + \varsigma^2} \right)
\\
& = & \frac{\sqrt{5} - 1}{\sqrt{5} + 1}c_1,
\ese
which implies that
\bse
\ell^{\dag}(f, g) = B^\dag[\bs h][\bs h]
& \ge & \left(\inf_{\| \bs h\|_{2} =   1}B^\dag[\bs h][\bs h] \right) \cdot \| \bs h\|_{2}^2
\\
& \ge & \frac{\sqrt{5} - 1}{\sqrt{5} + 1}c_1 \| \bs h\|_{2}^2
\\
& = & \frac{\sqrt{5} - 1}{\sqrt{5} + 1}c_1\left( \| f \|_{2}^2 + \| g \|_{2}^2  \right)
\\
& \ge & \frac{\sqrt{5} - 1}{\sqrt{5} + 1}c_1\left\{ c_2 V_1(f ,f) + c_3 V_0(g ,g) \right\}
\\
& \ge & c_4 \left\{V_1(f ,f) + V_0(g ,g) \right\},
\ese
by using Condition \ref{con5}. Thus, we proves Lemma \ref{lemma3}.
\end{proof}

\begin{lemma}\label{lemma4}
Suppose that Condition \ref{con5} holds, then
\begin{enumerate}
\item[(1)]
There exist a sequence of eigenfunctions $\varphi_{1\mu}\in \ch_1$ ($\mu\in\mb N$)
satisfying $\sup_{\mu\in\mb N}\|{\varphi_{1\mu}}\|_{\infty} \leq c_1$ and the corresponding non-decreasing sequence of eigenvalues
$\rho_{1\mu}$ ($\mu\in\mb N$) such that
$V_1(\varphi_{1\mu},\varphi_{1\mu'})= \delta_{\mu\mu'}$ and
$J_1(\varphi_{1\mu},\varphi_{1\mu'})=\rho_{1\mu}\delta_{\mu\mu'}$,
where $\delta_{\mu\mu'}$ is the Kronecker's delta and
$\mb N = \{0,1,2,\ldots\}$. Furthermore,
there exist a sequence of eigenfunctions $\varphi_{0\mu}\in \ch_0$ ($\mu\in\mb N$)
satisfying $\sup_{\mu\in\mb N}\|{\varphi_{0\mu}}\|_{\infty} \leq c_0$ and the corresponding non-decreasing sequence of eigenvalues
$\rho_{0\mu}$ ($\mu\in\mb N$) such that
$V_0(\varphi_{0\mu},\varphi_{0\mu'})= \delta_{\mu\mu'}$ and
$J_0(\varphi_{0\mu},\varphi_{0\mu'})=\rho_{0\mu}\delta_{\mu\mu'}$.

\item[(2)]
$\rho_{1\mu} \asymp \mu^{2m_1 / p} $ and $\rho_{0\mu} \asymp \mu^{2m_0 / p} $ for  sufficiently large $\mu$.

\item[(3)]
$f = \sum_{\mu}V_1(f, \varphi_{1\mu})\varphi_{1\mu}$ and $g = \sum_{\mu}V_0(g, \varphi_{0\mu})\varphi_{0\mu}$ for $f \in \ch_1$ and $g \in \ch_0$, where $\sum_\mu$ denotes the sum over $\mb N = \{0,1,2,\ldots\}$.
\end{enumerate}
\end{lemma}
The conclusion of Lemma \ref{lemma4} appears in many literature, see Cox (1984) on page 799, Cox (1988) on page 699, O'Sullivan (1993) on page 132, Gu (2013) on page 322, and thus we omit the proof.

\begin{lemma}\label{lemma5}
Suppose that Condition \ref{con5} holds, then for $f \in \ch_1, g \in \ch_0$,
and $\bs x \in \Omega$, we have
\bse
K_{1\bs x}(\cdot) & = & \sum_\mu\frac{\varphi_{1\mu}(\bs x)}{1+\gamma_1 \rho_{1\mu} }\varphi_{1\mu}(\cdot), \quad
W_{\gamma_1}\varphi_{1\mu}(\cdot) = \frac{\gamma_1\rho_{1\mu}}{1 + \gamma_1\rho_{1\mu}}\varphi_{1\mu}(\cdot),
\\
\|f\|_{\ch_1}^2 & = & \sum_\mu |V_1( f ,\varphi_{1\mu})|^2(1+\gamma_1\rho_{1\mu})
= \sum_\mu \frac{|\langle f ,\varphi_{1\mu}\rangle_{\ch_1}|^2}{1+\gamma_1\rho_{1\mu}},
\ese
and
\bse
K_{0\bs x}(\cdot) & = & \sum_\mu\frac{\varphi_{0\mu}(\bs x)}{1+\gamma_0 \rho_{0\mu} }\varphi_{0\mu}(\cdot), \quad
W_{\gamma_0}\varphi_{0\mu}(\cdot) = \frac{\gamma_0\rho_{0\mu}}{1 + \gamma_0\rho_{0\mu}}\varphi_{0\mu}(\cdot),
\\
\|g\|_{\ch_0}^2 & = & \sum_\mu |V_0( g ,\varphi_{0\mu})|^2(1+\gamma_0\rho_{0\mu})
= \sum_\mu \frac{|\langle g ,\varphi_{0\mu}\rangle_{\ch_0}|^2}{1+\gamma_0\rho_{0\mu}}.
\ese
\end{lemma}

\begin{lemma}\label{lemma6}
Suppose that Condition \ref{con5} holds, then for every $\bs x \in \Omega$, we have
\bse
\| K_{1\bs x} \|_{\ch_1} \le \wt{c}_1 \gamma_1^{-p/(4m_1)}, \quad \| K_{0\bs x} \|_{\ch_0} \le \wt{c}_0 \gamma_0^{-p/(4m_0)},
\ese
where $\wt{c}_1$ and $\wt{c}_0$ do not depend on the choice of $\bs x$.
Furthermore,
\bse
\|f\|_{\infty} \leq \wt{c}_1 \gamma_1^{-p/(4m_1)}\|f\|_{\ch_1},\quad
\|g\|_{\infty} \leq \wt{c}_0 \gamma_0^{-p/(4m_0)}\|g\|_{\ch_0},
\ese
for every  $f \in \ch_1$ and $ g \in \ch_0$.
\end{lemma}

Lemmas \ref{lemma5} and \ref{lemma6} are similar to Proposition 2.1 and Lemma 3.1 in Shang and Cheng (2013), and Lemma 0.1 and Lemma 0.2 in the supplementary material of Liu, Mao, and Zhao (2020).
The proofs essentially proceed along the lines of these literature and are trivial and omitted  for the sake of brevity.

\begin{lemma}\label{lemma7}
Suppose that Condition \ref{con5} holds, then for $f \in \ch_1$ and $g \in \ch_0$ such that
$J_1(f, f) \le c_1$ and $J_0(g, g) \le c_0$, we have
\bse
\| W_{\gamma_1}f \|_{\ch_1} = o(\gamma_1^{1/2}), \quad \| W_{\gamma_0}g \|_{\ch_0} = o(\gamma_0^{1/2}).
\ese
\end{lemma}
\begin{proof}
It follows from Lemmas \ref{lemma4} and \ref{lemma5} that
\bse
\| W_{\gamma_1} f \|_{\ch_1}^2
& = & \sum_\mu |V_1(W_{\gamma_1}f, \varphi_{1\mu})|^2(1 + \gamma_1\rho_{1\mu})
\\
& = & \sum_\mu \left|V_1\left( \sum_{\mu'}V_1(f, \varphi_{1\mu'})W_{\gamma_1}\varphi_{1\mu'}, \varphi_{1\mu}   \right) \right|^2 (1 + \gamma_1\rho_{1\mu})
\\
& = & \sum_{\mu} | V_1(f, \varphi_{1\mu}) |^2 \left(\frac{\gamma_1\rho_{1\mu}}{1 + \gamma_1\rho_{1\mu}} \right)^2 (1 + \gamma_1\rho_{1\mu})
\\
& = & \gamma_1 \sum_{\mu} | V_1(f, \varphi_{1\mu}) |^2 \frac{\gamma_1\rho_{1\mu}^2}{1 + \gamma_1\rho_{1\mu}}.
\ese
Thus, to prove $\| W_{\gamma_1} f \|_{\ch_1} = o(\gamma_1^{1/2})$, it suffices to show
\bsee\label{lemma6_eq1}
\sum_{\mu} | V_1(f, \varphi_{1\mu}) |^2 \frac{\gamma_1\rho_{1\mu}^2}{1 + \gamma_1\rho_{1\mu}} = o(1).
\esee
Let
$
f^\dag_{\gamma_1}(\mu) = | V_1(f, \varphi_{1\mu}) |^2 \frac{\gamma_1\rho_{1\mu}^2}{1 + \gamma_1\rho_{1\mu}}
$, $f^\dag(\mu) = | V_1(f, \varphi_{1\mu}) |^2 \rho_{1\mu}$ and ${\cal N}(\cdot)$ denotes the discrete measure over $\mb N$. Then we have
\bse
| f^\dag_{\gamma_1}(\mu) | \le f^\dag(\mu).
\ese
Furthermore, by Lemmas  \ref{lemma4} and \ref{lemma5},
\bse
J_1(f ,f) = J_1\left(\sum_\mu V_1(f, \varphi_{1\mu})\varphi_{1\mu}, \sum_\mu V_1(f, \varphi_{1\mu})\varphi_{1\mu} \right)
= \sum_\mu | V_1(f, \varphi_{1\mu}) |^2 \rho_{1\mu} \le c_1,
\ese
which implies that
\bse
\int_{\mb N} f^\dag(\mu) {\rm d}{\cal N(\mu)} = \sum_\mu | V_1(f, \varphi_{1\mu}) |^2 \rho_{1\mu} \le c_1.
\ese
Therefore, by dominated convergence theorem,
\bse
\lim_{\gamma_1 \rightarrow 0 } \sum_{\mu} | V_1(f, \varphi_{1\mu}) |^2 \frac{\gamma_1\rho_{1\mu}^2}{1 + \gamma_1\rho_{1\mu}} = \lim_{\gamma_1 \rightarrow 0 } \int_{\mb N} f^\dag_{\gamma_1}(\mu) {\rm d}{\cal N(\mu)}
=  \int_{\mb N}\lim_{\gamma_1 \rightarrow 0 } f^\dag_{\gamma_1}(\mu) {\rm d}{\cal N(\mu)}
= 0,
\ese
which proves equation (\ref{lemma6_eq1}). The other conclusion can be proved in a similar way. Thus, we complete the proof of Lemma \ref{lemma7}.
\end{proof}

Define the class of functions
${\ms G}_1 = \{ f \in \ch_1 :  \| f \|_{\ch_1} \le 1 \}$,
${\ms G}_0 = \{ g \in \ch_0 :  \| g \|_{\ch_0} \le 1 \}$. We state the following lemma.
\begin{lemma}\label{lemma8}
Let $\bs X_i$, $i = 1, \ldots, n'$, be the i.i.d.\hspace{-0.12cm} samples of $\bs X$, and $\ms H$ be a Donsker class of functions which are uniformly bounded such that $\log N(\epsilon, {\ms H}, \|\cdot \|_{\infty}) \le c_0 \epsilon^{-\upsilon'}$ for some constants $c_0 > 0 $ and $0 < \upsilon' < 2 $.
Suppose that Condition \ref{con5} holds, then
\bse
&&\sup_{h \in {\ms H}f, \wt{f} \in {\ms G}_1}\left|\frac{1}{n'}\sum_{i = 1}^{n'}h(\bs X_i)f(\bs X_i)\wt{f}(\bs X_i) - E \left\{\frac{1}{n'}\sum_{i = 1}^{n'}h(\bs X_i) f(\bs X_i)\wt{f}(\bs X_i) \right\} \right|
\\
& = & O_P\left((n')^{-1/2} \left\{\gamma_1^{(-2m_1\upsilon' - 4p + p\upsilon')/(8m_1)} +  \gamma_1^{-(6m_1 - p)p/(8m_1^2)} \right\}
\right),
\\
&&\sup_{h \in {\ms H}, g, \wt{g} \in {\ms G}_0}\left|\frac{1}{n'}\sum_{i = 1}^{n'}h(\bs X_i)g(\bs X_i)\wt{g}(\bs X_i) - E \left\{\frac{1}{n'}\sum_{i = 1}^{n'}h(\bs X_i)g(\bs X_i)\wt{g}(\bs X_i) \right\} \right|
\\
& = & O_P\left((n')^{-1/2} \left\{\gamma_0^{(-2m_0\upsilon' - 4p + p\upsilon')/(8m_0)} +  \gamma_0^{-(6m_0 - p)p/(8m_0^2)} \right\}
\right).
\ese
\end{lemma}

\begin{proof}
We first prove the first result. Define the class of functions
\bse
{\ms G}^*_1 = \left\{f^* = \gamma_1^{1/2}f: f\in \ch_1, \| f \|_{\infty} \le \wt{c}_1 \gamma_1^{-p/(4m_1)}, J_1(f, f) \le \gamma_1^{-1} \right\},
\ese
then we have
\bse
\log N( \epsilon, {\ms G}^*_1, \|\cdot\|_{\infty} ) \leq c_{1}\epsilon^{-p/m_1}.
\ese
We further define the class of functions
\bse
{\ms F} = \{F = h f^* : h \in {\ms H}, f^* \in{\ms G}^*_1 \},
\ese
then
\bsee\label{lemmA7_eq0}
\log N( \epsilon, {\ms F}, \|\cdot\|_{\infty} ) \leq c_{2}\left(\epsilon^{-\upsilon'} + \epsilon^{-p/m_1}\right).
\esee
For $F \in \ms F$,
define the empirical process ${\cal Z}_{n'}(F)$ as
\bse
{\cal Z}_{n'}(F) = \frac{1}{\sqrt{n'}}\sum_{i = 1}^{n'}\left[ \wt{c}_1^{-1} \gamma_1^{p/(4m_1)}F(\bs X_i)K_{1\bs X_i} - E\big\{ \wt{c}_1^{-1}\gamma_1^{p/(4m_1)}F(\bs X_i)K_{1\bs X_i} \big\} \right].
\ese
For any $F, \wt{F} \in {\ms F}$, by Lemma \ref{lemma6},
\bse
&&\left\|\wt{c}_1^{-1}\ \gamma_1^{p/(4m_1)}F(\bs X)K_{1\bs X} - \wt{c}_1^{-1}\gamma_1^{p/(4m_1)} \wt{F}(\bs X)K_{1\bs X} \right\|_{\ch_1}
\\
&\le & \wt{c}_1^{-1}\gamma_1^{p/(4m_1)} \| F - \wt{F} \|_{\infty} \| K_{1\bs X}\|_{\ch_1}
\\
&\le& \wt{c}_1^{-1}\gamma_1^{p/(4m_1)} \| F - \wt{F} \|_{\infty} \wt{c}_1 \gamma_1^{-p/(4m_1)}
\\
&= & \| F - \wt{F} \|_{\infty},
\ese
coupled with Theorem 2 of Hoeffding (1963), entail that
\bse
P\left(\left\| {\cal Z}_{n'}(F) - {\cal Z}_{n'}(\wt{F})\right\|_{\ch_1} \ge t  \right) \le 2\exp\left\{-\frac{t^2}{8\| F - \wt{F} \|_{\infty}^2} \right\}.
\ese
Together with Lemma 2.2.1 of van der Vaart and Wellner (1996),
\bse
\Bigg\| \left\| {\cal Z}_{n'}(F) - {\cal Z}_{n'}(\wt{F})\right\|_{\ch_1} \Bigg\|_{\eta_2}\le 8\left\| F - \wt{F} \right\|_{\infty},
\ese
where $\eta_2$ is the Orlicz norm associated with $\eta_2(x) = \exp(x^2) - 1$.
Using Theorem 2.2.4 of van der Vaart and Wellner (1996) and equation (\ref{lemmA7_eq0}), for any $\delta > 0 $, we have
\bse
&&\Bigg\| \sup_{F, \wt{F} \in {\ms F}, \| F - \wt{F} \|_{\infty} \le \delta } \left\| {\cal Z}_{n'}(F) - {\cal Z}_{n'}(\wt{F})\right\|_{\ch_1} \Bigg\|_{\eta_2}
\\
&\le& c_3 \left[\int_0^\delta \sqrt{ \log\left\{ 1 + N(\epsilon, {\ms F}, \| \cdot \|_\infty ) \right\} }{\rm d}\epsilon  + \delta \sqrt{ \log\left\{ 1 + N(\delta, {\ms F}, \| \cdot \|_\infty )^2 \right\} }\right]
\\
&\le& c_4 \left\{ \delta^{1 - \upsilon'/2 } + \delta^{1 - p/(2m_1)} \right\}.
\ese
Then by Markov's inequality,
\bse
P\left(\sup_{F \in {\ms F}, \| F \|_{\infty} \le \delta } \| {\cal Z}_{n'}(F)\|_{\ch_1} \ge t \right) \le c_5 \exp\left[-c_6 \left\{ \delta^{1 - \upsilon'/2 } + \delta^{1 - p/(2m_1)} \right\}^{-2} t^2 \right].
\ese
Set $\delta = \wt{c}_1 \gamma_1^{1/2 - p / (4m_1)}$, we get
\bse
P\left(\sup_{F \in {\ms F}} \| {\cal Z}_{n'}(F)\|_{\ch_1} \ge t \right) \le c_5 \exp\left[-c_7 \left\{\gamma_1^{ (2m_1 - p )(2- \upsilon')/(8m_1)} + \gamma_1^{ (2m_1 - p)^2 / (8m_1^2)} \right\}^{-2}t^2 \right],
\ese
which implies that
\bse
\sup_{F \in {\ms F}} \| {\cal Z}_{n'}(F)\|_{\ch_1} = O_P\left( \gamma_1^{ (2m_1 - p )(2- \upsilon')/(8m_1)} + \gamma_1^{ (2m_1 - p)^2 / (8m_1^2)} \right).
\ese

Noticing that $\gamma_1^{1/2}{\ms G}_1 = \{\gamma_1^{1/2} f : f \in{\ms G}_1 \} \subset {\ms G}_1^*$,  we obtain
\bse
&&\sup_{h \in {\ms H}, f ,\wt{f}\in {\ms G}_1}\left|\frac{1}{n'}\sum_{i = 1}^{n'}h(\bs X_i)f(\bs X_i)\wt{f}(\bs X_i) - E \left\{\frac{1}{n'}\sum_{i = 1}^{n'}h(\bs X_i)f(\bs X_i)\wt{f}(\bs X_i) \right\} \right|
\\
& = & \sup_{h \in{\ms H}, f ,\wt{f}\in {\ms G}_1}\left|\left \nl\frac{1}{n'}\sum_{i = 1}^{n'}h(\bs X_i)f(\bs X_i)K_{1\bs X_i} - E \left\{\frac{1}{n'}\sum_{i = 1}^{n'}h(\bs X_i)f(\bs X_i)K_{1\bs X_i} \right\} , \wt{f} \right\nr_{\ch_1} \right|
\\
&\le & \sup_{h \in{\ms H}, f ,\wt{f}\in {\ms G}_1}\left\|\frac{1}{n'}\sum_{i = 1}^{n'}h(\bs X_i)f(\bs X_i)K_{1\bs X_i} - E \left\{\frac{1}{n'}\sum_{i = 1}^{n'}h(\bs X_i) f(\bs X_i)K_{1\bs X_i} \right\} \right\|_{\ch_1} \cdot \| \wt{f} \|_{\ch_1}
\\
&\le & (n')^{-1/2} \wt{c}_1 \gamma_1^{-1/2 -p/(4m_1)} \sup_{F \in {\ms F}} \| {\cal Z}_{n'}(F)\|_{\ch_1}
\\
&= & O_P\left((n')^{-1/2} \left\{\gamma_1^{(-2m_1\upsilon' - 4p + p\upsilon')/(8m_1)} +  \gamma_1^{-(6m_1 - p)p/(8m_1^2)} \right\}
\right).
\ese
The remainder can be proved through a similar argument. Thus, we prove Lemma \ref{lemma8}.
\end{proof}

According to Condition \ref{con2}, there exist some $\tau_n \in \Phi_n$ and $\lambda_n \in \Psi_n$ such that
\bsee\label{th1_eq0}
\|\tau_n - \tau_0 \|_{\infty} & = &  O(n^{-\kappa_1}),\quad
\sup_{|k| = m_1}\left\|\tau_n^{(k)} - \tau_0^{(k)} \right\|_{\infty} =  O(n^{-\kappa'_1}),
\nonumber
\\
\|\lambda_n - \lambda_0 \|_{\infty} & = & O(n^{-\kappa_0}),\quad
\sup_{|k'| = m_0}\left\|\lambda_n^{(k')} - \lambda_0^{(k')} \right\|_{\infty} = O(n^{-\kappa'_0}).
\esee
Then we assert the following lemma.
\begin{lemma}\label{lemma9}
Suppose that the assumptions in Theorem 1 of the main paper hold,
then $(\wh{\tau}_{n}, \wh{\lambda}_{n})$, the minimum loss estimator of $\wh{\ell}_{n ,\gamma_1, \gamma_0}(\tau, \lambda)$ over $\Phi_n \times \Psi_n$, satisfies
\bse
&&V_1(\wh{\tau}_{n} - \tau_n,  \wh{\tau}_{n} - \tau_n ) = o_P(1),\quad  V_0(\wh{\lambda}_{n} - \lambda_n, \wh{\lambda}_{n} - \lambda_n) = o_P(1),
\\
&&J_1(\wh{\tau}_{n} - \tau_n,  \wh{\tau}_{n} - \tau_n ) = o_P(1),\quad  J_0(\wh{\lambda}_{n} - \lambda_n, \wh{\lambda}_{n} - \lambda_n) = o_P(1).
\ese
\end{lemma}

\begin{proof}
Choose $\xi_{1n} \in \Phi_n$ and $\xi_{0n} \in \Psi_n$ such that $(\xi_{1n}, \xi_{0n}) \in {\cal G}_1 \times {\cal G}_0$ with
\bse
{\cal G}_1 \times {\cal G}_0 = \bigg\{ (f , g)\in \ch:
c_1^{-1} \le V_1(f, f) + V_0(g, g) \le c_1, \
J_1(f , f) + J_0(g , g) \le c_2 \bigg\}.
\ese
Moreover, by Sobolev embedding theorem, we have
\bsee\label{th1_eq01}
\| \xi_{1n}\|_{\infty} + \| \xi_{0n} \|_\infty \le c_3.
\esee
For every $t \in \mb R$, let
\bse
H_n(t) & = & \wh{\ell}_{n, \gamma_1, \gamma_0}(\tau_n + t \xi_{1n}, \lambda_n + t \xi_{0n})
\\
& = & \frac{1}{2n} \sum_{i = 1}^n \left\{\wh{\sigma}^2(\bs X_i, S_i)\right\}^{-1}\left[ \wh{D}_i - \{\tau_n(\bs X_i) + t\xi_{1n}(\bs X_i)\} - (1 - S_i)\{ \lambda_n(\bs X_i) + t \xi_{0n}(\bs X_i) \} \right]^2
\\
&&+ \frac{\gamma_1}{2}J_1(\tau_n + t \xi_{1n}, \tau_n + t \xi_{1n}) + \frac{\gamma_0}{2} J_0(\lambda_n + t \xi_{0n}, \lambda_n + t \xi_{0n}),
\ese
then the derivative of $H_n(t)$ with respect to $t$ is
\bsee\label{th1_eq1}
\dot{H}_n(t) &=& -\frac{1}{n}\sum_{i = 1}^n \left\{\wh{\sigma}^2(\bs X_i, S_i)\right\}^{-1} \left[ \wh{D}_i - \{\tau_n(\bs X_i) + t \xi_n(\bs X_i)\} - (1 - S_i)\{ \lambda_n(\bs X_i) + t \eta_n(\bs X_i) \} \right]
\nonumber
\\
&&\times \{ \xi_{1n}(\bs X_i) + ( 1 - S_i )\xi_{0n}(\bs X_i) \}
\nonumber
\\
&&+ t\gamma_1J_1(\xi_{1n}, \xi_{1n})
+ \gamma_1J_1(\tau_n , \xi_{1n}) + t \gamma_0J_0(\xi_{0n}, \xi_{0n}) + \gamma_0J_0(\lambda_n , \xi_{0n})
\nonumber
\\
& = & -\frac{1}{n}\sum_{i = 1}^{n_1}\left\{\wh{\sigma}^2(\bs X_i, 1)\right\}^{-1}( \wh{D}_i - {D}_i )\xi_{1n}(\bs X_i)
-\frac{1}{n}\sum_{i = 1}^{n_1}\left\{\wh{\sigma}^2(\bs X_i, 1)\right\}^{-1}\left\{{D}_i - \tau_n(\bs X_i)\right\}\xi_{1n}(\bs X_i)
\nonumber
\\
&&+ \frac{t}{n}\sum_{i = 1}^{n_1}\left\{\wh{\sigma}^2(\bs X_i, 1)\right\}^{-1}\xi_{1n}^2(\bs X_i)
\nonumber
\\
&&
- \frac{1}{n}\sum_{i = n_1 + 1}^n\left\{\wh{\sigma}^2(\bs X_i, 0)\right\}^{-1}\left\{{D}_i - \tau_n(\bs X_i) - \lambda_n(\bs X_i) \right\}\left\{\xi_{1n}(\bs X_i) + \xi_{0n}(\bs X_i)\right\}
\nonumber
\\
&& + \frac{t}{n}\sum_{i = n_1 + 1}^n\left\{\wh{\sigma}^2(\bs X_i, 0)\right\}^{-1} \left\{\xi_{1n}(\bs X_i) + \xi_{0n}(\bs X_i) \right\}^2
+t\gamma_1J_1(\xi_{1n}, \xi_{1n})
+ \gamma_1J_1(\tau_n , \xi_{1n})
\nonumber
\\
&&+ t \gamma_0J_0(\xi_{0n}, \xi_{0n}) + \gamma_0J_0(\lambda_n , \xi_{0n})
\nonumber
\\
& = & -I_{n1} - I_{n2} + t I_{n3} - I_{n4} + t I_{n5} + t I_{n6} + I_{n7} + t I_{n8} + I_{n9},
\esee
where $I_{n1}, \ldots, I_{n9}$ are clear from the above equation.

We first consider $I_{n2}$. It follows that
\bse
|I_{n2}| & = & \frac{n_1}{n}\left| \mb P_{n_1}\left\{(\wh{\sigma}^2_1)^{-1}(D - \tau_n)\xi_{1n} \right\} \right|
\\
& \le & \frac{n_1}{n}\left| \mb P_{n_1}\left\{(\wh{\sigma}^2_1)^{-1}(D - \tau_0)\xi_{1n} \right\} \right| +
\frac{n_1}{n}\left| \mb P_{n_1}\left\{(\wh{\sigma}^2_1)^{-1}(\tau_n - \tau_0)\xi_{1n} \right\} \right|
\\
& \le & \left| \mb P_{n_1}\left\{(\wh{\sigma}^2_1)^{-1}(D - \tau_0)\xi_{1n} \right\} \right| + \left| \mb P_{n_1}\left\{(\wh{\sigma}^2_1)^{-1}(\tau_n - \tau_0)\xi_{1n} \right\} \right|.
\ese
Define the class of functions
\bse
{\ms F}^\dag = \{ (\wh{\sigma}^2_1)^{-1}(D - \tau_0)\xi_{1n}: (\wh{\sigma}^2_1)^{-1} \in {\ms A}, \xi_{1n} \in {\cal G}_1 \}.
\ese
According to the definition of $\xi_{1n}$, it is easy to see that
$
\log N(\epsilon, {\cal G}_1, \| \cdot \|_\infty ) \le c_4\epsilon^{-p / m_1}
$.
By Assumptions 1--2 in the main paper, Condition \ref{con3}, and $\tau_0$ is uniformly bounded,
\bse
\log N(\epsilon, {\ms F}^\dag, \| \cdot \|_\infty ) \le c_5 (\epsilon^{-\upsilon} + \epsilon^{-p / m_1} ),
\ese
combined with the fact $m_1 > p/2$, i.e., $0 < p/m_1 < 2$, $0 < \upsilon < 2$ stated in Condition \ref{con3}, and $\mb P_{1}\{(\wh{\sigma}^2_1 )^{-1}(D - \tau_0)\xi_{1n} \} = 0  $, lead to
\bsee\label{th1_add_eq1}
&& \sup_{(\wh{\sigma}^2_1)^{-1}\in{\ms A}, \xi_{1n} \in {\cal G}_1}\left| \mb P_{n_1}\left\{(\wh{\sigma}^2_1)^{-1}(D - \tau_0)\xi_{1n} \right\} \right|
\nonumber
\\
& = & \sup_{(\wh{\sigma}^2_1)^{-1}\in{\ms A}, \xi_{1n} \in {\cal G}_1}\left| (\mb P_{n_1} - \mb P_1 )\left\{(\wh{\sigma}^2_1)^{-1}(D - \tau_0)\xi_{1n} \right\} \right|
\nonumber
\\
& = & O_P\left(n_1^{-1/2}\right).
\esee
By equation (\ref{th1_eq0}) and Condition \ref{con3},
\bse
\sup_{(\wh{\sigma}^2_1)^{-1}\in{\ms A}, \xi_{1n} \in {\cal G}_1} \left| \mb P_{n_1}\left\{(\wh{\sigma}^2_1)^{-1}(\tau_n - \tau_0)\xi_{1n} \right\} \right| \le O(1)\| \tau_n - \tau_0 \|_{\infty} \| \xi_{1n} \|_{\infty} = O\left(n^{-\kappa_1}\right).
\ese
Therefore
\bsee\label{th1_eq3}
|I_{n2}| &\le& \sup_{(\wh{\sigma}^2_1)^{-1}\in{\ms A}, \xi_{1n} \in {\cal G}_1}\left| \mb P_{n_1}\left\{(\wh{\sigma}^2_1)^{-1}(D - \tau_0)\xi_{1n} \right\} \right| + \sup_{(\wh{\sigma}^2_1)^{-1}\in{\ms A}, \xi_{1n} \in {\cal G}_1} \left| \mb P_{n_1}\left\{(\wh{\sigma}^2_1)^{-1}(\tau_n - \tau_0)\xi_{1n} \right\} \right|
\nonumber
\\
& = & O_P\left(n^{-1/2} \right)+ O\left( n^{-\kappa_1}\right)
\nonumber
\\
& = & o_P(1).
\esee
Mimicking the proof of (\ref{th1_eq3}), we can also get
\bsee\label{th1_eq4}
|I_{n4}| = o(1) +  O\left(n^{-\kappa_1} + n^{-\kappa_0} \right) = o(1).
\esee
Using an argument as we did in equation (\ref{th1_add_eq1}), we have
\bse
\sup_{(\wh{\sigma}^2_1)^{-1}\in{\ms A}, \xi_{1n} \in {\cal G}_1}\left|(\mb P_{n_1} - \mb P_1)(\wh{\sigma}^2_1)^{-1} \xi_{1n}^2 \right| = O_P\left( n_1^{-1/2} \right),
\ese
which implies that
\bsee\label{th1_eq5}
I_{n3} & = & \frac{n_1}{n} \left\{ (\mb P_{n_1} - \mb P_1)(\wh{\sigma}^2_1)^{-1} \xi_{1n}^2 + \mb P_1 (\wh{\sigma}^2_1)^{-1} \xi_{1n}^2 \right\}
\nonumber
\\
& = & \frac{n_1}{n} \int \left\{ \wh{\sigma}_1^2(\bs x) \right\}^{-1} \xi_{1n}^2(\bs x)q_1(\bs x ){\rm d}\bs x +o_P(1).
\esee
Similarly, we can also obtain
\bsee\label{th1_eq6}
I_{n5} = \frac{n_2}{n} \int\left\{ \wh{\sigma}^2_0(\bs x) \right\}^{-1}\{\xi_{1n}(\bs x) + \xi_{0n}(\bs x) \}^2q_0(\bs x){\rm d}\bs x + o_P(1).
\esee
It follows from equation (\ref{th1_eq0}) that
\bsee\label{th1_eq7}
&&I_{n6} = o(1), \quad I_{n7} \le \gamma_1 \{J_1(\tau_n , \tau_n)\}^{1/2}\{J_1(\xi_{1n} , \xi_{1n})\}^{1/2} = o(1),
\nonumber
\\
&& I_{n8} = o(1), \quad I_{n9} \le \gamma_0 \{J_0(\lambda_n , \lambda_n)\}^{1/2}\{J_0(\xi_{0n} , \xi_{0n})\}^{1/2} = o(1).
\esee

We now deal with $I_{n1}$. It follows that
\bse
|I_{n1}| & = & \left|\frac{n_1}{n}\mb P_{n_1}(\wh{\sigma}^2_1)^{-1}(\wh{D} - D )\xi_{1n} \right|
\\
&\le & \left|\mb P_{n_1}(\wh{\sigma}^2_1 )^{-1}(\wh{D} - D)\xi_{1n}\right|
\\
&\le & \left|(\mb P_{n1} - \mb P_1 )(\wh{\sigma}^2_1 )^{-1}(\wh{D} - D)\xi_{1n} \right| + \left|\mb P_1 (\wh{\sigma}^2_1 )^{-1}(\wh{D} - D)\xi_{1n} \right|.
\ese
Using arguments similar to that in equation (\ref{th1_add_eq1}), we get
\bse
\sup_{(\wh{\sigma}^2_1 )^{-1} \in {\ms A}, \wh{D}\in{\ms D}, \xi_{1n} \in{\cal G}_1 }\left|(\mb P_{n1} - \mb P_1 )(\wh{\sigma}^2_1 )^{-1} (\wh{D} - D )\xi_{1n} \right| = O_P\left(n_1^{-1/2}\right).
\ese
Under Conditions \ref{con3}--\ref{con6}, by Lemma \ref{lemma2}, H\"{o}lder inequality and Minkowski inequality, we have
\bse
&& \sup_{(\wh{\sigma}^2_1 )^{-1} \in {\ms A}, \wh{D}\in{\ms D}, \xi_{1n} \in{\cal G}_1 }\left|\mb P_1 \left(\wh{\sigma}^2_1 \right)^{-1}\left(\wh{D} - D \right)\xi_{1n} \right|
\\
&= &\sup_{(\wh{\sigma}^2_1 )^{-1} \in {\ms A}, \wh{D}\in{\ms D}, \xi_{1n} \in{\cal G}_1 } \left|E_r \left[ \left\{\wh{\sigma}^2_1(\bs X) \right\}^{-1}\xi_{1n}(\bs X) E_r \left\{\left(\wh{D} - D\right) \,\bigg|\, \bs X, S = 1 \right\} \right] \right|.
\\
& = &\sup_{(\wh{\sigma}^2_1 )^{-1} \in {\ms A}, \xi_{1n} \in{\cal G}_1 } \bigg|E_r \big[ \left\{\wh{\sigma}^2_1(\bs X) \right\}^{-1}\xi_{1n}(\bs X) \left\{\Gamma_1(\bs X) + \Gamma_2(\bs X) \right\} \big] \,\big|\, S = 1 \bigg|
\\
&\le & O(1) \left( \sum_{a = 0 }^1 \| \Theta_a \|_2
+ \| \wh{e} -e \|_2 \times \sum_{a = 0 }^1 \| \wh{\mu}_a - \mu_a \|_2 \right)
\\
&= & o_P(1).
\ese
As a consequence,
\bsee\label{th1_eq8}
|I_{n1}| = o_P(1).
\esee

Plugging (\ref{th1_eq3})--(\ref{th1_eq8}) into (\ref{th1_eq1}), we have
\bse
&& \dot{H}_n(t)
\\
&=& t \frac{n_1}{n} \int_\Omega \left\{\wh{\sigma}^2_1(\bs x) \right\}^{-1}\xi_{1n}^2(\bs x)q_1(\bs x){\rm d}\bs x
\\
&& + t \frac{n_2}{n} \int_\Omega \left\{\wh{\sigma}^2_0(\bs x) \right\}^{-1}\{\xi_{1n}(\bs x) + \xi_{0n}(\bs x) \}^2q_0(\bs x){\rm d}\bs x + o_P(1)
\\
& = & t \bigg\{\varrho \int_\Omega \left\{\wh{\sigma}^2_1(\bs x) \right\}^{-1}\xi_{1n}^2(\bs x)q_1(\bs x ){\rm d}\bs x
\\
&& + (1 - \varrho) \int_\Omega\left\{\wh{\sigma}^2_0(\bs x) \right\}^{-1}\{\xi_{1n}(\bs x) + \xi_{0n}(\bs x) \}^2q_0(\bs x){\rm d}\bs x \bigg\}+ o_P(1).
\ese
By Lemma \ref{lemma3}, Conditions \ref{con3} and \ref{con5},
\bse
&&\varrho \int \left\{ \wh{\sigma}^2_{1}(\bs x) \right\}^{-1}\xi_{1n}^2(\bs x)q_1(\bs x ){\rm d}\bs x + (1 - \varrho) \int \left\{ \wh{\sigma}^2_{0}(\bs x) \right\}^{-1}\{\xi_{1n}(\bs x) + \xi_{0n}(\bs x) \}^2q_0(\bs x){\rm d}\bs x
\\
&\ge & c_6 \left\{ V_1(\xi_{1n} , \xi_{1n}) + V_0(\xi_{0n} ,\xi_{0n} ) \right\}.
\ese
Immediately, with probability tending to one,
\bse
\dot{H}_n(t)  \ge  t c_6 \left\{ V_1(\xi_{1n} , \xi_{1n}) + V_0(\xi_{0n} ,\xi_{0n} ) \right\} \ge  c_6c_1^{-1}t >  0
\ese
for $t > 0 $. Similarly, we can also obtain that $\dot{H}_n(t) < 0$ for $t < 0 $. In conclusion,
$
t \dot{H}_n(t) > 0$ for $t \neq 0$ with probability tending to one.
By the arbitrariness of $t$, we conclude that
\bse
&&V_1(\wh{\tau}_{n} - \tau_n, \wh{\tau}_{n} - \tau_n ) = o_P(1), \quad V_0(\wh{\lambda}_{n} - \lambda_n, \wh{\lambda}_{n} - \lambda_n ) = o_P(1),
\\
&&J_1(\wh{\tau}_{n} - \tau_n, \wh{\tau}_{n} - \tau_n ) = o_P(1), \quad J_0(\wh{\lambda}_{n} - \lambda_n, \wh{\lambda}_{n} - \lambda_n ) = o_P(1),
\ese
which proves Lemma \ref{lemma9}.
\end{proof}

\subsection*{Appendix E: Proofs}
\subsubsection*{E.1 Proof of Theorem 1}
\begin{proof}
  It follows from Lemma \ref{lemma9} and equation (\ref{th1_eq0}) that
\bse
V_1(\wh{\tau}_{n} - \tau_0,  \wh{\tau}_{n} - \tau_0 )
&\le & 2\big\{V_1(\wh{\tau}_{n} - \tau_n, \wh{\tau}_{n} - \tau_n ) +
V_1({\tau}_{n} - \tau_0, {\tau}_{n} - \tau_0 ) \big\}
\\
& = & o_P(1) + O(n^{-2\kappa_1})
\\
& = & o_P(1),
\\
J_1(\wh{\tau}_{n} - \tau_0,  \wh{\tau}_{n} - \tau_0 )
&\le & 2\big\{J_1(\wh{\tau}_{n} - \tau_n, \wh{\tau}_{n} - \tau_n ) +
J_1({\tau}_{n} - \tau_0, {\tau}_{n} - \tau_0 ) \big\}
\\
& = & o_P(1) + O(n^{-2\kappa'_1}).
\ese
Similarly, we can also get
\bse
V_0(\wh{\lambda}_{n} - \lambda_0,  \wh{\lambda}_{n} - \lambda_0 ) = o_P(1), \quad
J_0(\wh{\lambda}_{n} - \lambda_0,  \wh{\lambda}_{n} - \lambda_0 ) = o_P(1) + O(n^{-2\kappa'_0}).
\ese
Immediately,
\bse
\| \wh{\tau}_n - \tau_0 \|_{\ch_1} &= &\bigg\{V_1(\wh{\tau}_{n} - \tau_0,  \wh{\tau}_{n} - \tau_0 ) + \gamma_1 J_1(\wh{\tau}_{n} - \tau_0,  \wh{\tau}_{n} - \tau_0 ) \bigg\}^{1/2} =  o_P(1),
\\
\left\| \wh{\lambda}_n - \lambda_0 \right\|_{\ch_0} & = & \bigg\{V_0(\wh{\lambda}_{n} - \lambda_0,  \wh{\lambda}_{n} - \lambda_0 ) + \gamma_0 J_0(\wh{\lambda}_{n} - \lambda_0,  \wh{\lambda}_{n} - \lambda_0 ) \bigg\}^{1/2} =  o_P(1).
\ese
Thus, we conclude Theorem 1.
\end{proof}

\subsubsection*{E.2 Proof of Theorem 2}
\begin{proof}
  Taylor's expansion for $\wh{\ell}_{n, \gamma_1, \gamma_0}(\wh{\tau}_{n}, \wh{\lambda}_{n})$ around $(\tau_n, \lambda_n)$ yields
\bse
&&\wh{\ell}_{n, \gamma_1, \gamma_0}(\wh{\tau}_{n}, \wh{\lambda}_{n}) - \wh{\ell}_{n,\gamma_1, \gamma_0}(\tau_n, {\lambda}_{n})
\\
&&=
\dot{\wh{\ell}}_{n, \gamma_1, \gamma_0}(\tau_n, {\lambda}_{n})[(\wh{\tau}_{n} - \tau_n, \wh{\lambda}_{n} - \lambda_n)]
\\
&&+ \frac{1}{2}\ddot{\wh{\ell}}_{n, \gamma_1, \gamma_0}(\tau_n, {\lambda}_n)[(\wh{\tau}_{n} - \tau_n, \wh{\lambda}_{n} - \lambda_n)][(\wh{\tau}_{n} - \tau_n, \wh{\lambda}_{n} - \lambda_n)]
\\
&& = I^*_{n1} + I^*_{n2},
\ese
where
\bse
I^*_{n1} & = & \dot{\wh{\ell}}_{n, \gamma_1, \gamma_0}(\tau_n, {\lambda}_{n})[(\wh{\tau}_{n} - \tau_n, \wh{\lambda}_{n} - \lambda_n)]
\\
& = &
-\frac{1}{n} \sum_{i = 1}^n \left\{\wh{\sigma}^2(\bs X_i, S_i)\right\}^{-1} \left\{ \wh{D}_i - \tau_n(\bs X_i) - (1 - S_i){\lambda}_{n}(\bs X_i) \right\}
\\
&&\times \left[ \{\wh{\tau}_{n}(\bs X_i) - \tau_n(\bs X_i)\} + (1 - S_i)
\left\{ \wh{\lambda}_{n}(\bs X_i) - \lambda_n(\bs X_i) \right\} \right]
\\
&&+ \gamma_1 J_1(\tau_n, \wh{\tau}_{n} - \tau_n)
+ \gamma_0 J_0(\lambda_n, \wh{\lambda}_{n} - \lambda_n),
\ese
and
\bse
I^*_{n2} & = & \frac{1}{2}\ddot{\wh{\ell}}_{n, \gamma_1, \gamma_0}(\tau_n, {\lambda}_n)[(\wh{\tau}_{n} - \tau_n, \wh{\lambda}_{n} - \lambda_n)][(\wh{\tau}_{n} - \tau_n, \wh{\lambda}_{n} - \lambda_n)]
\\
& = &
\frac{1}{2n} \sum_{i = 1}^n \left\{\wh{\sigma}^2(\bs X_i, S_i)\right\}^{-1}\left[ \{\wh{\tau}_{n}(\bs X_i) - \tau_n(\bs X_i)\} + (1 - S_i)
\left\{ \wh{\lambda}_{n}(\bs X_i) - \lambda_n(\bs X_i) \right\} \right]^2
\\
&& + \frac{\gamma_1}{2} J_1(\wh{\tau}_{n} - \tau_n, \wh{\tau}_{n} - \tau_n)
+ \frac{\gamma_0}{2} J_0(\wh{\lambda}_{n} - \lambda_n, \wh{\lambda}_{n} - \lambda_n).
\ese
We first consider $I^*_{n1}$. It follows that
\bsee\label{An1}
-I^*_{n1} &=& \frac{n_1}{n}\mb P_{n_1}(\wh{\sigma}_1^2)^{-1}(\wh{D} - D)(\wh{\tau}_n - \tau_n)
+
\frac{n_1}{n}\mb P_{n_1}(\wh{\sigma}_1^2)^{-1}({D} - \tau_0)(\wh{\tau}_n - \tau_n)
\nonumber
\\
&& + \frac{n_0}{n}\mb P_{n_0}(\wh{\sigma}_0^2)^{-1}({D} - \tau_0 - \lambda_0)(\wh{\tau}_n - \tau_n)
+ \frac{n_0}{n}\mb P_{n_0}(\wh{\sigma}_0^2)^{-1}({D} - \tau_0 - \lambda_0)(\wh{\lambda}_n - \lambda_n)
\nonumber
\\
&&+ \frac{n_1}{n}\mb P_{n_1}(\wh{\sigma}_1^2)^{-1}(\tau_0 - \tau_n)(\wh{\tau}_n - \tau_n)
+ \frac{n_0}{n}\mb P_{n_0}(\wh{\sigma}_0^2)^{-1}(\tau_0 - \tau_n)(\wh{\tau}_n - \tau_n)
\nonumber
\\
&& + \frac{n_0}{n}\mb P_{n_0}(\wh{\sigma}_0^2)^{-1}(\tau_0 - \tau_n)(\wh{\lambda}_n - \lambda_n) + \frac{n_0}{n}\mb P_{n_0}(\wh{\sigma}_0^2)^{-1}(\lambda_0 - \lambda_n)(\wh{\tau}_n - \tau_n)
\nonumber
\\
&&
+\frac{n_0}{n}\mb P_{n_0}(\wh{\sigma}_0^2)^{-1}(\lambda_0 - \lambda_n)(\wh{\lambda}_n - \lambda_n)
- \gamma_1J_1(\tau_n, \wh{\tau}_{n} - \tau_n) - \gamma_0J_0(\lambda_n, \wh{\lambda}_{n} - \lambda_n)
 \nonumber
\\
& = & I^*_{n1,1} + I^*_{n1,2} +  I^*_{n1,3} + I^*_{n1,4} + I^*_{n1,5} + I^*_{n1,6}
+ I^*_{n1,7} + I^*_{n1,8} + I^*_{n1,9}
\nonumber
 \\
&&- \gamma_1J_1(\tau_n, \wh{\tau}_{n} - \tau_n) - \gamma_0J_0(\lambda_n, \wh{\lambda}_{n} - \lambda_n),
\esee
where $I^*_{n1, 1}, \ldots, I^*_{n1, 9}$ are clear from the expression.
Following equation (\ref{th1_eq0}), it is clear that $
J_1(\tau_n, \tau_n) \le O(1)$, combined with Lemma \ref{lemma7}, deduce that
$
\|W_{\gamma_1}\tau_n\|_{\ch_1} = o(\gamma_1^{1/2})
$.
Consequently,
\bsee\label{lemma9_eq1}
|\gamma_1J_1(\tau_n, \wh{\tau}_{n} - \tau_n)| = \big|\nl W_{\gamma_1}\tau_n, \wh{\tau}_{n} - \tau_n \nr_{\ch_1}\big|
&\le&
\| W_{\gamma_1}\tau_n \|_{\ch_1}\|  \wh{\tau}_{n} - \tau_n\|_{\ch_1}
\nonumber
\\
&= & o(\gamma_1^{1/2})\| \wh{\tau}_{n} - \tau_n\|_{\ch_1}.
\esee
Similarly,
\bsee\label{lemma9_eq1a}
\left|\gamma_0J_0(\lambda_n, \wh{\lambda}_{n} - \lambda_n) \right| = o(\gamma_0^{1/2})\left\|\wh{\lambda}_{n} - \lambda_n\right\|_{\ch_0}.
\esee
For $I^*_{n1, 1}$,
\bse
|I^*_{n1, 1}|
&\le& \left| \mb P_{n_1}(\wh{\sigma}_1^2)^{-1}(\wh{D} - D)(\wh{\tau}_n - \tau_n) \right|
\\
& \le &  \left|(\mb P_{n_1} - \mb P_1)(\wh{\sigma}_1^2)^{-1}(\wh{D} - D)(\wh{\tau}_n - \tau_n)\right| + \left|\mb P_{1}(\wh{\sigma}_1^2)^{-1}(\wh{D} - D)(\wh{\tau}_n - \tau_n)\right|.
\ese
Using a similar argument as we did in the proof of Lemma \ref{lemma8}, we have
\bse
\left|(\mb P_{n_1} - \mb P_1)(\wh{\sigma}_1^2)^{-1}(\wh{D} - D)(\wh{\tau}_n - \tau_n)\right|
& = & \left| \left\nl (\mb P_{n_1} - \mb P_1)(\wh{\sigma}_1^2)^{-1}(\wh{D} - D)K_1, \wh{\tau}_{n} - \tau_n \right\nr_{\ch_1} \right|
\\
& \le & \left\| (\mb P_{n_1} - \mb P_1)(\wh{\sigma}_1^2)^{-1}(\wh{D} - D)K_1 \right\|_{\ch_1} \cdot \| \wh{\tau}_{n} - \tau_n \|_{\ch_1}
\\
& = & O_P\left(n_1^{-1/2}\gamma_1^{-p/(4m_1)}\right)
\|\wh{\tau}_{n} - \tau_n \|_{\ch_1}
\\
& = & O_P\left(n^{-1/2}\gamma_1^{-p/(4m_1)}\right)
\|\wh{\tau}_{n} - \tau_n \|_{\ch_1}.
\ese
In the following, we calculate $\mb P_{1}(\wh{\sigma}_1^2)^{-1}(\wh{D} - D)(\wh{\tau}_n - \tau_n)$.
Under Assumptions 1--2 in the main paper, Conditions \ref{con3}--\ref{con5} ans \ref{con7},
by Lemma \ref{lemma2}, H\"{o}lder inequality and Minkowski inequality, we have
\bse
&&\left|\mb P_1 (\wh{\sigma}_1^2)^{-1} (\wh{D} - D)(\wh{\tau}_n - \tau_n) \right|
\\
&= &\left|E_r \left[ \{\wh{\sigma}_1^2(\bs X )\}^{-1} \{\wh{\tau}_n(\bs X) - \tau_n(\bs X) \} E_r \left\{(\wh{D} - D) \,\big|\, \bs X, S = 1 \right\} \right] \right|.
\\
& = &\bigg|E_r \big[ \{\wh{\sigma}_1^2(\bs X)\}^{-1} \{\wh{\tau}_n(\bs X) - \tau_n(\bs X) \} \left\{\Gamma_1(\bs X ) + \Gamma_2(\bs X) \right\} \mid S = 1 \big] \bigg|
\\
&\le & O(1) \big( E_r \big[ \{\wh{\tau}_n(\bs X) - \tau_n(\bs X)\}^2 \mid S = 1 \big]\big)^{1/2} \left(
\big[ E_r \{\Gamma_1(\bs X)\}^2 \big]^{1/2} + \big[ E_r \{\Gamma_2(\bs X)\}^2 \big]^{1/2}
\right)
\\
&\le & O(1) \| \wh{\tau}_n - \tau_n\|_{\ch_1} \left( \sum_{a = 0 }^1 \| \Theta_a \|_2
+ \| \wh{e} -e \|_2 \times \sum_{a = 0 }^1 \| \wh{\mu}_a - \mu_a \|_2 \right)
\\
&= & O_P\left(\delta_n \right)\| \wh{\tau}_n - \tau_n\|_{\ch_1},
\ese
where $\delta_n = n^{-1/2}\gamma_1^{-p/(4m_1)} + n^{-1/2}\gamma_0^{-p/(4m_0)} + \gamma_1^{1/2} + \gamma_0^{1/2}$.
Therefore,
\bsee\label{An10}
I^*_{n1,1} & = & \left\{O_P\left(n^{-1/2}\gamma_1^{-p/(4m_1)}\right) + O_P\left(\delta_n  \right)\right\} \|\wh{\tau}_{n} - \tau_n \|_{\ch_1}
\nonumber
\\
& = &O_P\left(\delta_n \right)\|\wh{\tau}_{n} - \tau_n \|_{\ch_1}.
\esee
A similar argument as in  (\ref{An10}) is used for $I^*_{n1, 2}, I^*_{n1, 3}$, and $I^*_{n1,4}$, we can get the similar results as follows.
\bsee\label{An11}
|I^*_{n1, 2}|  & = &    O_P\left(\delta_n \right)\| \wh{\tau}_{n} - \tau_n \|_{\ch_1},\quad
|I^*_{n1, 3}|  =  O_P\left(\delta_n \right)\| \wh{\tau}_{n} - \tau_n \|_{\ch_1},
\nonumber
\\
|I^*_{n1, 4}|  &= & O_P\left(\delta_n \right)\| \wh{\lambda}_{n} - \lambda_n \|_{\ch_0}.
\esee

We now consider $I^*_{n1,5}$. Mimicking the proof of Lemma \ref{lemma8}, we have
\bse
\frac{n}{n_1}\left|I^*_{n1, 5}- E_r( I^*_{n1,5} ) \right| & = & \left|(\mb P_{n_1} - \mb P_1)(\wh{\sigma}_1^2 )^{-1}(\tau_n - \tau_0)(\wh{\tau}_n - \tau_n)\right|
\\
& = & | \left\nl (\mb P_{n_1} - \mb P_1)(\wh{\sigma}_1^2)^{-1}(\tau_n - \tau_0)K_1, \wh{\tau}_{n} - \tau_n \right\nr_{\ch_1} |
\\
& \le & \| (\mb P_{n_1} - \mb P_1)(\wh{\sigma}_1^2)^{-1}(\tau_n - \tau_0)K_1 \|_{\ch_1} \cdot \| \wh{\tau}_{n} - \tau_n \|_{\ch_1}
\\
& \le & O_P\left(n_1^{-1/2}\gamma_1^{-p/(4m_1)} \| \tau_n - \tau_0 \|_{\infty}^{1 - p/(2m_1)} \right)
\|\wh{\tau}_{n} - \tau_n \|_{\ch_1}
\\
&\le& O_P\left(n^{-(1/2 +  \kappa_1(2m_1 - p )/(2m_1))}\gamma_1^{-p/(4m_1)} \right)
\|\wh{\tau}_{n} - \tau_n \|_{\ch_1}.
\ese
By equation (\ref{th1_eq0}) and Condition \ref{con3},
\bse
\frac{n}{n_1}|E_r (I^*_{n1, 5})| &=& \left|\int_\Omega \{\wh{\sigma}_1^2(\bs x)\}^{-1}\{\tau_n(\bs x) - \tau_0(\bs x) \}\left\{ \wh{\tau}_{n}(\bs x) - \tau_n(\bs x)\right\}q_1(\bs x) {\rm d}\bs x \right|
\\
&\le & O(1) \| \tau_n - \tau_0 \|_{\infty} \| \wh{\tau}_{n} - \tau_n \|_{\ch_1}
\\
&\le &O\left(n^{-\kappa_1} \right)\| \wh{\tau}_{n} - \tau_n \|_{\ch_1}.
\ese
Therefore,
\bsee\label{An15}
| I^*_{n1, 5} | &\le& |I^*_{n1, 5} - E_r (I^*_{n1,5})| + |E_r (I^*_{n1,5})|
\nonumber
\\
& = &O_P\left(n^{-(1/2  + \kappa_1(2m_1 - p )/(2m_1))}\gamma_1^{-p/(4m_1)} + n^{-\kappa_1} \right)\| \wh{\tau}_{n} - \tau_n \|_{\ch_1}
\nonumber
\\
& = & \left\{ o_P\left( n^{-1/2}\gamma_1^{-p/(4m_1)} \right) + O_P\left( n^{-\kappa_1} \right) \right\}\| \wh{\tau}_{n} - \tau_n \|_{\ch_1}.
\esee
Using an argument  similar to that in (\ref{An15}) for $I^*_{n1,6},\ldots ,I^*_{n1,9}$, we get similar results as follows.
\bsee\label{An16}
I^*_{n1, 6} &= & \left\{ o_P\left( n^{-1/2}\gamma_0^{-p/(4m_0)} \right) + O_P\left( n^{-\kappa_1} \right) \right\}\| \wh{\tau}_{n} - \tau_n \|_{\ch_1},
\nonumber
\\
I^*_{n1, 7} &=& \left\{ o_P\left( n^{-1/2}\gamma_0^{-p/(4m_0)} \right) + O_P\left( n^{-\kappa_1} \right) \right\}\left\| \wh{\lambda}_{n} - \lambda_n \right\|_{\ch_0},
\nonumber
\\
I^*_{n1, 8} &= & \left\{ o_P\left( n^{-1/2}\gamma_0^{-p/(4m_0)} \right) + O_P\left( n^{-\kappa_0} \right) \right\}\| \wh{\tau}_{n} - \tau_n \|_{\ch_1},
\nonumber
\\
I^*_{n1, 9} &= & \left\{ o_P\left( n^{-1/2}\gamma_0^{-p/(4m_0)} \right) + O_P\left( n^{-\kappa_0} \right) \right\}\left\| \wh{\lambda}_{n} - \lambda_n \right\|_{\ch_0}.
\esee

Submitting (\ref{An10})--(\ref{An16}) into (\ref{An1}) and using Condition \ref{con9}, we have
\bsee\label{th2_eq1}
|I^*_{n1}| & \le &  O_P\left(\delta_n + n^{-\kappa_1} + n^{-\kappa_0}\right)
  \left(\| \wh{\tau}_{n} - \tau_n \|_{\ch_1} + \left\| \wh{\lambda}_{n} - \lambda_n \right\|_{\ch_0} \right)
  \nonumber
\\
&\le & O_P\left(\delta_n \right) \left(\| \wh{\tau}_{n} - \tau_n \|_{\ch_1} + \left\| \wh{\lambda}_{n} - \lambda_n \right\|_{\ch_0} \right).
\esee

We now deal with $I^*_{n2}$.
\bsee\label{th2_eq2}
&&I^*_{n2}
\nonumber
\\
& = & \frac{n_1}{n} \mb P_{n_1}(\wh{\sigma}^2_1)^{-1}(\wh{\tau}_{n} - \tau_n)^2
+ \frac{n_0}{n} \mb P_{n_0}(\wh{\sigma}^2_0)^{-1}(\wh{\tau}_{n} - \tau_n)^2
+ \frac{n_0}{n} \mb P_{n_0}(\wh{\sigma}^2_0)^{-1}(\wh{\lambda}_{n} - \lambda_n)^2
\nonumber
\\
&& + 2\frac{n_0}{n} \mb P_{n_0}(\wh{\sigma}^2_0)^{-1}(\wh{\tau}_{n} - \tau_n)(\wh{\lambda}_{n} - \lambda_n)
+
\gamma_1 J_1(\wh{\tau}_{n} - \tau_n, \wh{\tau}_{n} - \tau_n)
\nonumber
\\
&& + \gamma_0 J_0(\wh{\lambda}_{n} - \lambda_n, \wh{\lambda}_{n} - \lambda_n).
\esee
Following Lemma \ref{lemma8} and Condition \ref{con8}, we have
\bse
&&\left|(\mb P_{n_1} - \mb P_1)(\wh{\sigma}^2_1)^{-1}(\wh{\tau}_{n} - \tau_n)^2\right|
\\
& \le & \sup_{(\wh{\sigma}^2_1)^{-1} \in {\ms A}, f,\wt{f}\in{\ms G}_1} \left| (\mb P_{n_1} - \mb P_1)(\wh{\sigma}^2_1)^{-1} f\wt{f}\right|
\| \wh{\tau}_{n} - \tau_n \|_{\ch_1}^2.
\\
& = &  O_P\left(n^{-1/2}\left\{\gamma_1^{(-2m_1\upsilon - 4p + p\upsilon)/(8m_1)} + \gamma_1^{-(6m_1 - p)p/(8m_1^2)} \right\} \right)\| \wh{\tau}_{n} - \tau_n \|_{\ch_1}^2
\\
& = & o_P(1) \| \wh{\tau}_{n} - \tau_n \|_{\ch_1}^2,
\ese
which implies that
\bsee\label{th2_eq2a}
\mb P_{n_1}(\wh{\sigma}^2_1)^{-1}(\wh{\tau}_{n} - \tau_n)^2 & = &
\mb P_1(\wh{\sigma}^2_1)^{-1}(\wh{\tau}_{n} - \tau_n)^2 + o_P(1) \| \wh{\tau}_{n} - \tau_n \|_{\ch_1}^2.
\esee
Employing arguments analogous to equation (\ref{th2_eq2a}), we conclude
\bsee\label{th2_eq2b}
\mb P_{n_0}(\wh{\sigma}^2_{0})^{-1}(\wh{\tau}_{n} - \tau_n)^2 & = & \mb P_0(\wh{\sigma}^2_{0})^{-1}(\wh{\tau}_{n} - \tau_n)^2 + o_P(1) \| \wh{\tau}_{n} - \tau_n \|_{\ch_1}^2,
\nonumber
\\
\mb P_{n_0}(\wh{\sigma}^2_{0})^{-1}(\wh{\lambda}_{n} - \lambda_n)^2 & = &
\mb P_0(\wh{\sigma}^2_{0})^{-1}(\wh{\lambda}_{n} - \lambda_n)^2 + o_P(1)\| \wh{\lambda}_{n} - \lambda_n\|_{\ch_0}^2.
\esee

Now we consider $\mb P_{n_0}(\wh{\sigma}^2_{0})^{-1}(\wh{\tau}_{n} - \tau_n)(\wh{\lambda}_{n} - \lambda_n)$.
Define the set ${\cal A } = \{ V_0(\wh{\lambda}_{n} - \lambda_n, \wh{\lambda}_{n} - \lambda_n) + J_0(\wh{\lambda}_{n} - \lambda_n, \wh{\lambda}_{n} - \lambda_n) \le 1 \}$.
For every $\epsilon > 0$, according to Lemma \ref{lemma9}, there exists some $N(\epsilon)$ such that $P({\cal A}) > 1 - \epsilon / 2$ for all $n > N(\epsilon)$.
Mimicking the proof of (\ref{An15}), we further have
\bse
&& P \left( \sup_{(\wh{\sigma}^2_{0})^{-1} \in {\ms A}, V_0(\wh{\lambda}_{n} - \lambda_n, \wh{\lambda}_{n} - \lambda_n) + J_0(\wh{\lambda}_{n} - \lambda_n, \wh{\lambda}_{n} - \lambda_n) \le 1} \left\| (\mb P_{n_0} - \mb P_0 )(\wh{\sigma}^2_{0})^{-1}(\wh{\lambda}_{n} - \lambda_n)K_1 \right\|_{\ch_1} > t \right)
\\
& \le & c_1 \exp(-c_2 n_0\gamma_1^{p / (2m_1)} t^2 ),
\ese
which implies that
\bse
&& P \left( \sup_{(\wh{\sigma}^2_{0})^{-1} \in {\ms A}, \wh{\lambda}_n, \lambda_n} \left\| (\mb P_{n_0} - \mb P_0 )(\wh{\sigma}^2_{0})^{-1}(\wh{\lambda}_{n} - \lambda_n)K_1 \right\|_{\ch_1} > t \right)
\\
&\le& P \left( \sup_{(\wh{\sigma}^2_{0})^{-1} \in {\ms A}, \wh{\lambda}_n, \lambda_n } \left\| (\mb P_{n_0} - \mb P_0 )(\wh{\sigma}^2_{0})^{-1}(\wh{\lambda}_{n} - \lambda_n)K_1 \right\|_{\ch_1} > t \, \bigg  | \, {\cal A} \right) + P({\cal A}^C)
\\
& \le & P \left( \sup_{(\wh{\sigma}^2_{0})^{-1} \in {\ms A}, V_0(\wh{\lambda}_{n} - \lambda_n, \wh{\lambda}_{n} - \lambda_n) + J_0(\wh{\lambda}_{n} - \lambda_n, \wh{\lambda}_{n} - \lambda_n) \le 1} \left\| (\mb P_{n_0} - \mb P_0 )(\wh{\sigma}^2_{0})^{-1}(\wh{\lambda}_{n} - \lambda_n)K_1 \right\|_{\ch_1} > t \right)
\\
&& + P({\cal A}^C)
\\
&\le & c_1 \exp(-c_2 n_0\gamma_1^{p / (2m_1)} t^2 ) + \epsilon / 2
\\
&\le & \epsilon
\ese
by taking $t = \{ c_2^{-1} \log (2c_1 / \epsilon) n_0^{-1} \gamma_1^{-p/(2m_1)} \}^{1/2} $.
Therefore,
\bse
\sup_{(\wh{\sigma}^2_{0})^{-1} \in {\ms A}, \wh{\lambda}_n, \lambda_n} \left\| (\mb P_{n_0} - \mb P_0 )(\wh{\sigma}^2_{0})^{-1}(\wh{\lambda}_{n} - \lambda_n)K_1 \right\|_{\ch_1} = O_P\left(n_0^{-1/2}\gamma_1^{-p / (4m_1)} \right) = O_P\left(n^{-1/2}\gamma_1^{-p / (4m_1)} \right),
\ese
furthermore,
\bse
\left| (\mb P_{n_0} - \mb P_0 )(\wh{\sigma}^2_{0})^{-1}(\wh{\tau}_{n} - \tau_n)(\wh{\lambda}_{n} - \lambda_n) \right|
&\le &
\left\| (\mb P_{n_0} - \mb P_0 )(\wh{\sigma}^2_{0})^{-1}(\wh{\lambda}_{n} - \lambda_n)K_1 \right\|_{\ch_1} \| \wh{\tau}_{n} - \tau_n \|_{\ch_1}
\\
& = & O_P\left(n^{-1/2}\gamma_1^{-p / (4m_1)} \right)\| \wh{\tau}_{n} - \tau_n \|_{\ch_1}.
\ese
Consequently,
\bsee\label{th2_eq2c}
&& \mb P_{n_0}(\wh{\sigma}^2_{0})^{-1}(\wh{\tau}_{n} - \tau_n)(\wh{\lambda}_{n} - \lambda_n)
\nonumber
\\
& = & \mb P_0(\wh{\sigma}^2_{0})^{-1}(\wh{\tau}_{n} - \tau_n)(\wh{\lambda}_{n} - \lambda_n)
+ O_P\left(n^{-1/2}\gamma_1^{-p/(4m_1)}\right)\|\wh{\tau}_{n} - \tau_n \|_{\ch_1}.
\esee
Plugging equations (\ref{th2_eq2a})--(\ref{th2_eq2c}) into equation (\ref{th2_eq2}), we conclude that
\bse
I^*_{n2} & = & \varrho \mb P_1(\wh{\sigma}^2_{1})^{-1}(\wh{\tau}_{n} - \tau_n)^2 + (1 - \varrho)
\mb P_0(\wh{\sigma}^2_{0})^{-1}(\wh{\tau}_{n} - \tau_n)^2 + (1 - \varrho)\mb P_0(\wh{\sigma}^2_{0})^{-1}(\wh{\lambda}_{n} - \lambda_n)^2
\\
&&+ 2(1 - \varrho) \mb P_0(\wh{\sigma}^2_{0})^{-1}(\wh{\tau}_{n} - \tau_n)(\wh{\lambda}_{n} - \lambda_n)
+\gamma_1 J_1(\wh{\tau}_{n} - \tau_n, \wh{\tau}_{n} - \tau_n)
+ \gamma_0 J_0(\wh{\lambda}_{n} - \lambda_n, \wh{\lambda}_{n} - \lambda_n)
\\
&&+ o_P(1) \| \wh{\tau}_{n} - \tau_n \|_{\ch_1}^2 +
o_P(1)\| \wh{\lambda}_{n} - \lambda_n\|_{\ch_0}^2
+ O_P\left(n^{-1/2}\gamma_1^{-p/(4m_1)}\right)\|\wh{\tau}_{n} - \tau_n \|_{\ch_1}.
\ese
It follows from Lemma \ref{lemma3}, Conditions \ref{con3} and \ref{con5} that
\bse
&& \varrho \mb P_1(\wh{\sigma}^2_{1})^{-1}(\wh{\tau}_{n} - \tau_n)^2 + (1 - \varrho)\mb P_0(\wh{\sigma}^2_{0})^{-1}(\wh{\tau}_{n} - \tau_n)^2 + (1 - \varrho)\mb P_0(\wh{\sigma}^2_{0})^{-1}(\wh{\lambda}_{n} - \lambda_n)^2
\\
&&+ 2(1 - \varrho)\mb P_0(\wh{\sigma}^2_{0})^{-1}(\wh{\tau}_{n} - \tau_n)(\wh{\lambda}_{n} - \lambda_n)
+\gamma_1 J_1(\wh{\tau}_{n} - \tau_n, \wh{\tau}_{n} - \tau_n)
+ \gamma_0 J_0(\wh{\lambda}_{n} - \lambda_n, \wh{\lambda}_{n} - \lambda_n)
\\
& = & \varrho \mb P_1(\wh{\sigma}^2_{1})^{-1}(\wh{\tau}_{n} - \tau_n)^2
+ (1 - \varrho)  \mb P_0 (\wh{\sigma}^2_{0})^{-1}(\wh{\tau}_{n} - \tau_n + \wh{\lambda}_{n} - \lambda_n )^2
\\
&& + \gamma_1 J_1(\wh{\tau}_{n} - \tau_n, \wh{\tau}_{n} - \tau_n)
+ \gamma_0 J_0(\wh{\lambda}_{n} - \lambda_n, \wh{\lambda}_{n} - \lambda_n)
\\
&\ge & c_3\left\{V_1( \wh{\tau}_{n} - \tau_n , \wh{\tau}_{n} - \tau_n ) +  V_0( \wh{\lambda}_{n} - \lambda_n ,  \wh{\lambda}_{n} - \lambda_n ) \right\}
\\
&& + \gamma_1 J_1(\wh{\tau}_{n} - \tau_n, \wh{\tau}_{n} - \tau_n)
+ \gamma_0 J_0(\wh{\lambda}_{n} - \lambda_n, \wh{\lambda}_{n} - \lambda_n)
\\
& \ge & c_4 \bigg\{V_1( \wh{\tau}_{n} - \tau_n , \wh{\tau}_{n} - \tau_n ) +  V_0( \wh{\lambda}_{n} - \lambda_n ,  \wh{\lambda}_{n} - \lambda_n )
\\
&& + \gamma_1 J_1(\wh{\tau}_{n} - \tau_n, \wh{\tau}_{n} - \tau_n)
+ \gamma_0 J_0(\wh{\lambda}_{n} - \lambda_n, \wh{\lambda}_{n} - \lambda_n) \bigg\}
\\
& = & c_4 \left\{ \|\wh{\tau}_{n} - \tau_n \|_{\ch_1}^2 +  \| \wh{\lambda}_{n} - \lambda_n \|_{\ch_0}^2\right\}
\ese
for some $0 < c_4 < 1$.
Therefore,
\bse
I^*_{n2} \ge c_5 \left\{ \|\wh{\tau}_{n} - \tau_n \|_{\ch_1}^2 +  \| \wh{\lambda}_{n} - \lambda_n \|_{\ch_0}^2\right\} + O_P\left(n^{-1/2}\gamma_1^{-p/(4m_1)}\right)\|\wh{\tau}_{n} - \tau_n \|_{\ch_1}.
\ese
Recall that $(\wh{\tau}_{n}, \wh{\lambda}_{n})$ is the minimizer of $\wh{\ell}_{n,\gamma_1, \gamma_0}(\tau, \lambda)$ over $\Phi_n \times \Psi_n$, thus $\wh{\ell}_{n, \gamma_1, \gamma_0}(\wh{\tau}_{n}, \wh{\lambda}_{n}) - \wh{\ell}_{n,\gamma_1, \gamma_0}(\tau_n, {\lambda}_n) \le 0 $, which implies that $I^*_{n2} \le |I^*_{n1}|$, that is,
\bse
&&c_5 \left\{ \|\wh{\tau}_{n } - \tau_n \|_{\ch_1}^2 +  \| \wh{\lambda}_{n} - \lambda_n \|_{\ch_0}^2\right\} + O_P\left(n^{-1/2}\gamma_1^{-p/(4m_1)}\right)\|\wh{\tau}_{n} - \tau_n \|_{\ch_1}
\nonumber
\\
&\le & O_P\left(\delta_n \right)
\left(\| \wh{\tau}_{n} - \tau_n \|_{\ch_1} + \| \wh{\lambda}_{n} - \lambda_n \|_{\ch_0} \right),
\ese
which leads to
\bse
\|\wh{\tau}_{n} - \tau_n \|_{\ch_1}  =  O_P\left(\delta_n \right), \quad
\| \wh{\lambda}_{n} - \lambda_n \|_{\ch_0}  =  O_P\left(\delta_n\right).
\ese
Equation (\ref{th1_eq0}) and some routine calculations entail that
\bse
\| \tau_n - \tau_0 \|_{\ch_1} = O(n^{-\kappa_1} + \gamma_1^{1/2}n^{-\kappa'_1}), \quad
\| \lambda_n - \lambda_0 \|_{\ch_0} = O(n^{-\kappa_0} + \gamma_0^{1/2}n^{-\kappa'_0}).
\ese
Therefore, the conclusion of Theorem 2 is deduced by
\bse
\|\wh{\tau}_{n} - \tau_0 \|_{\ch_1} &\le& \|\wh{\tau}_{n} - \tau_n \|_{\ch_1} + \| \tau_n - \tau_0 \|_{\ch_1}, \\
\| \wh{\lambda}_{n} - \lambda_0 \|_{\ch_0} &\le&
\| \wh{\lambda}_{n} - \lambda_n \|_{\ch_0} + \| \lambda_n - \lambda_0 \|_{\ch_0},
\ese
and Condition \ref{con9}.
\end{proof}

\subsubsection*{E.3 Proof of Theorem 3}
\begin{proof}
  For $\bs h, \wt{\bs h} \in \ch$,
define the linear functional $\bs h : \ch \rightarrow \mb R$ as
$\bs h[\wt{\bs h}] = \langle \bs h , \wt{\bs h} \rangle_{\ch}$.
Set $\bs \Upsilon = (\tau, \lambda), \bs \Upsilon_0 = (\tau_0, \lambda_0),
\wh{\bs \Upsilon}_n = (\wh{\tau}_n, \wh{\lambda}_n)$.
Let
$\wh{\bs U}_n(\bs \Upsilon)$ and $\dot{\wh{\bs U}}_n(\bs \Upsilon)$ denote the first and second order Fr\'{e}chet derivative operates of $\wh{\ell}_{n}(\bs \Upsilon)$
with respective to $\bs \Upsilon$,
$\wh{\bs U}_{n, \gamma_1, \gamma_0}(\bs \Upsilon)$ and $\dot{\wh{\bs U}}_{n, \gamma_1, \gamma_0}(\bs \Upsilon)$ denote the first and second order Fr\'{e}chet derivative operates of $\wh{\ell}_{n, \gamma_1, \gamma_0}(\bs \Upsilon)$
with respective to $\bs \Upsilon$,
and $\bs U(\bs \Upsilon)$ and $\dot{\bs U}(\bs \Upsilon)$ denote the probability limit of $\wh{\bs U}_n(\bs \Upsilon)$ and $\dot{\wh{\bs U}}_n(\bs \Upsilon)$,
respectively.
Then, for $\bs h = (f, g) \in \ch, \wt{\bs h} = (\wt{f}, \wt{g}) \in \ch$, we have
\bse
\wh{\bs U}_n(\bs \Upsilon)[\bs h] &= &
-\frac{1}{n} \sum_{i = 1}^n \left\{\wh{\sigma}^2(\bs X_i, S_i)\right\}^{-1} \left\{ \wh{D}_i - \tau(\bs X_i) - (1 - S_i){\lambda}(\bs X_i) \right\}
\\
&& \times \{f(\bs X_i) + (1 - S_i)g(\bs X_i) \}
\\
\dot{\wh{\bs U}}_n(\bs \Upsilon)[\bs h][\wt{\bs h}] &=&
\frac{1}{n} \sum_{i = 1}^n \left\{\wh{\sigma}^2(\bs X_i, S_i)\right\}^{-1}
\\
&& \times
\{f(\bs X_i) + (1 - S_i)g(\bs X_i) \}\left\{\wt{f}(\bs X_i) + (1 - S_i)\wt{g}(\bs X_i) \right\},
\\
\wh{\bs U}_{n, \gamma_1, \gamma_0}(\bs \Upsilon)[\bs h] &=& \wh{\bs U}_n(\bs \Upsilon)[\bs h] + \bs W_{\bs \gamma}(\bs \Upsilon)[\bs h],
\\
\dot{\wh{\bs U}}_{n, \gamma_1, \gamma_0}(\bs \Upsilon)[\bs h][\wt{\bs h}] &=&
\dot{\wh{\bs U}}_n(\bs \Upsilon)[\bs h][\wt{\bs h}] + \dot{\bs W}_{\bs \gamma}(\bs \Upsilon)[\bs h][\wt{\bs h}],
\ese
where $\bs \gamma = (\gamma_1, \gamma_0)$, $\bs W_{\bs \gamma}(\bs \Upsilon) = (W_{\gamma_1} \tau, W_{\gamma_0}\lambda)$ and
$
\dot{\bs W}_{\bs \gamma}(\bs \Upsilon)[\bs h][\wt{\bs h}] = \gamma_1 J_1(f_\vartheta, \wt{f}_\vartheta) + \gamma_0 J_0(g_\vartheta, \wt{g}_\vartheta)
$.
Given $\bs x_0 \in \Omega$ and $\vartheta > 0 $, define the class of functions
\bse
\ms H_\vartheta  =  \ms H_{1\vartheta} \times \ms H_{0\vartheta}
=\left\{(f, g) \in \ch : \left\|(f, g) - (\gamma_1^{p / (4m_1)}K_{1\bs x_0}, \gamma_0^{p / (4m_0)}K_{0\bs x_0} ) \right\|_{\ch} \le \vartheta \right\}
\ese
By Lemma \ref{lemma6}, $\mb P_1 f_\vartheta^2 \le \| f_\vartheta \|^2_{\ch_1} / \varrho \le ( \vartheta + \wt{c}_1 )^2 / \varrho $ for all $f_\vartheta \in \ms H_{1\vartheta}$,
thus $\ms H_{1\vartheta}$ is $\mb P_1$-Donsker. Similarly, $\ms H_{1\vartheta} + \ms H_{0\vartheta}$ is $\mb P_0$-Donsker.

The proof mainly consists of six steps.
\begin{enumerate}
\item[Step 1.]
Show that
\bsee\label{th3_step1}
\sqrt{n} \wh{\bs U}_n(\bs \Upsilon_0) \rightsquigarrow \bs G \quad \textrm{in } \ell^\infty(\ms H_\vartheta),
\esee
where
$\bs G[\bs h_{\vartheta}]$ is some zero-mean gaussian process with the covariance process
\bse
{\rm Cov}(\bs G[\bs h_{\vartheta}], \bs G[\bs h^\dag_{\vartheta}] ) =
\varrho \mb P_1 \left(\sigma^2_{01}\right)^{-1}f_\vartheta \wt{f}_\vartheta +
(1 - \varrho) \mb P_0 \left(\sigma^2_{00}\right)^{-1}(f_\vartheta + g_\vartheta ) (\wt{f}_\vartheta + \wt{g}_\vartheta ),
\ese
and $\bs h_{\vartheta} = (f_\vartheta, g_\vartheta) \in \ms H_{\vartheta},
\wt{\bs h}_{\vartheta} = (\wt{f}_\vartheta, \wt{g}_\vartheta) \in \ms H_{\vartheta}$.

\item[Step 2.]
Show that
\bsee\label{th3_step2}
\sqrt{n} \dot{\wh{\bs U}}_n(\bs \Upsilon_0)[\wh{\bs \Upsilon}_n - \bs \Upsilon_0] =
\sqrt{n} \dot{\bs U}(\bs \Upsilon_0)[\wh{\bs \Upsilon}_n - \bs \Upsilon_0]
+o_P\left(1 + \sqrt{n}\left\| \wh{\bs \Upsilon}_n - \bs \Upsilon_0 \right\|_\ch \right).
\esee

\item[Step 3.]
Show the invertibility of $\dot{\bs U}_{\bs \gamma}(\bs \Upsilon_0)$,
where $\dot{\bs U}_{\bs \gamma}(\bs \Upsilon_0) = \dot{\bs U}(\bs \Upsilon_0)
+ \dot{\bs W}_{\bs \gamma}(\bs \Upsilon_0)$.

\item[Step 4.]
Show that
\bsee\label{th3_step4}
\wh{\bs U}_{n,\gamma_1, \gamma_0}(\wh{\bs \Upsilon}_n)[\bs h_\vartheta] = o_P(n^{-1/2}),
\esee
holds uniformly for $\bs h_\vartheta = (f_\vartheta, g_\vartheta) \in \ms H_\vartheta$.

\item[Step 5.]
Steps 2--4 yield
\bse
&&\sqrt{n}\dot{\bs U}_{\bs \gamma}({\bs \Upsilon}_0)\left[\wh{\bs \Upsilon}_n - {\bs \Upsilon}_0 + \{{\bs U}_{\bs \gamma}'({\bs \Upsilon}_0)\}^{-1} [ \bs W_{\bs \gamma}(\bs \Upsilon_0)]\right]
\\
& = & -\sqrt{n} \wh{\bs U}_n(\bs \Upsilon_0)
+ o_P\left(1 + \sqrt{n} \left\| \wh{\bs \Upsilon}_n - \bs \Upsilon_0 \right\|_{\ch} \right),
\ese
and further by Step 1,
\bsee\label{th3_eq11a}
\sqrt{n} \left(\wh{\bs \Upsilon}_n - {\bs \Upsilon}_0 + \dot{\bs U}^{-1}_{\bs \gamma}({\bs \Upsilon}_0) [ \bs W_{\bs \gamma}(\bs \Upsilon_0)]\right)
\rightsquigarrow  \dot{\bs U}^{-1}_{\bs \gamma}({\bs \Upsilon}_0) [\bs G]
\quad \textrm{ in } \ell^\infty(\ms H_\vartheta).
\esee

\item[Step 6.]
For $\omega_1 \in \mb R$ and $\omega_0 \in \mb R$,
let $\bs h_{\bs \omega} = (\omega_1 \gamma_1^{p/(4m_1)} K_{1\bs x_0}, \omega_0 \gamma_0^{p/(4m_0)} K_{0\bs x_0})$ and plug it into Step 5, we finally conclude Theorem 4.3.
\end{enumerate}

We first prove Step 1.
According to Condition \ref{con10}, we have
\bse
E_r\left(\wh{D} \mid \bs X , S = 1 \right) - E(D \mid \bs X , S = 1) = o_P(n^{-1/2}) = o_P(1).
\ese
Additionally, by Condition \ref{con14}, $E_r(\wh{D} \mid \bs X , S = 1) - E(D_0 \mid \bs X , S = 1) = o_P(1)$.
Thus, with probability one,
\bsee\label{th3_eq0a}
E({D}_0 \mid \bs X , S = 1) = E( D \mid \bs X , S = 1) = \tau_0(\bs X).
\esee
It follows that
\bsee\label{th3_eq1}
&&-\sqrt{n}\wh{\bs U}_n(\bs \Upsilon_0)[\bs h_{\vartheta}]
\nonumber
\\
& = &
\frac{n_1}{n}\sqrt{n} \mb P_{n_1}({\sigma}_{01}^2)^{-1}( D_0 - \tau_0 )f_\vartheta
+ \frac{n_0}{n}\sqrt{n} \mb P_{n_0}({\sigma}_{00}^2)^{-1}( D - \tau_0 -\lambda_0 )
( f_\vartheta + g_\vartheta )
\nonumber
\\
&& + \frac{n_1}{n}\sqrt{n} \mb P_{n_1}(\wh{\sigma}_1^2)^{-1}(\wh{D} - D_0 )f_\vartheta
+ \frac{n_1}{n} \sqrt{n} \mb P_{n_1}\left\{(\wh{\sigma}_1^2)^{-1} - ({\sigma}_{01}^2)^{-1} \right\}( D - \tau_0 )f_\vartheta
\nonumber
\\
&& + \frac{n_0}{n} \sqrt{n} \mb P_{n_0}\left\{(\wh{\sigma}_0^2)^{-1} - ({\sigma}_{00}^2)^{-1} \right\}( D - \tau_0 -\lambda_0 )
( f_\vartheta + g_\vartheta ).
\esee
By Conditions \ref{con3} and \ref{con13},
\bse
 \mb P_1\left[\left\{(\wh{\sigma}_1^2)^{-1} - ({\sigma}_{01}^2)^{-1} \right\}( D_0 - \tau_0 )f_\vartheta \right]^2
\le  O(1) \left\| (\wh{\sigma}_1^2)^{-1} - ({\sigma}_{01}^2)^{-1}  \right\|^2_{\infty} \mb P_1 f_\vartheta ^2
 =  o_P(1),
\ese
combined with lemma 19.24 in van der Vaart (1998),
\bse
\sqrt{n_1} (\mb P_{n_1} - \mb P_1 )\left\{(\wh{\sigma}_1^2)^{-1} - ({\sigma}_{01}^2)^{-1} \right\}( D_0 - \tau_0 )f_\vartheta = o_P(1).
\ese
Obviously,
$
\mb P_1\left\{(\wh{\sigma}_1^2)^{-1} - ({\sigma}_{01}^2)^{-1} \right\}( D_0 - \tau_0 )f_\vartheta = 0
$ by equation (\ref{th3_eq0a}), then
\bsee\label{th3_eq2}
&& \frac{n_1}{n} \sqrt{n} \mb P_{n_1}\left\{(\wh{\sigma}_1^2)^{-1} - ({\sigma}_{01}^2)^{-1} \right\}( D_0 - \tau_0 )f_\vartheta
\nonumber
\\
& = &
\sqrt{\frac{n_1}{n}} \sqrt{n_1} (\mb P_{n_1} - \mb P_1 )\left\{(\wh{\sigma}_1^2)^{-1} - ({\sigma}_{01}^2)^{-1} \right\}( D - \tau_0 )f_\vartheta
\nonumber
\\
& = & o_P(1).
\esee
Using an argument as we did in equation (\ref{th3_eq2}), we can also get
\bsee\label{th3_eq2a}
&&\frac{n_0}{n} \sqrt{n} \mb P_{n_0}\left\{(\wh{\sigma}_0^2)^{-1} - ({\sigma}_{00}^2)^{-1} \right\}( D - \tau_0 -\lambda_0 )
( f_\vartheta + g_\vartheta )
\nonumber
\\
& = & \sqrt{\frac{n_0}{n}} \sqrt{n_0} ( \mb P_{n_0} - \mb P_0 )\left\{(\wh{\sigma}_0^2)^{-1} - ({\sigma}_{00}^2)^{-1} \right\}( D - \tau_0 -\lambda_0 )
( f_\vartheta + g_\vartheta )
\nonumber
\\
& = & o_P(1).
\esee

Now, we deal with the term $\sqrt{n} \mb P_{n_1} (\wh{\sigma}_1^2)^{-1}(\wh{D} - D_0 )f_\vartheta$. Following Conditions \ref{con3} and \ref{con14}, we have
\bse
\mb P_{1} \left\{(\wh{\sigma}_1^2)^{-1}(\wh{D} - D_0 )f_\vartheta \right\}^2
\le
O(1) \| \wh{D} - D_0 \|_{\infty}^2   \mb P_{1}f_\vartheta^2
 \le  o_P(1) \mb P_1 f_\vartheta^2
 =  o_P(1),
\ese
coupled with lemma 19.24 in van der Vaart (1998), lead to
\bsee\label{th3_eq2b}
\sqrt{n_1} (\mb P_{n_1} - \mb P_1 ) (\wh{\sigma}_1^2)^{-1}(\wh{D} - D_0 )f_\vartheta
= o_P(1).
\esee
Furthermore, by equation (\ref{th3_eq0a}) and Lemma \ref{lemma2},
\bse
E_r \left(\wh{D} - D_0 \mid \bs X , S = 1 \right) = E_r \left(\wh{D} - D \mid \bs X , S = 1 \right) =
\Gamma_{1}(\bs X) + \Gamma_{2}(\bs X),
\ese
which implies that
\bse
\left| \mb P_1 (\wh{\sigma}_1^2)^{-1}(\wh{D} - D_0 )f_\vartheta \right|
& = & \left| \int_\Omega \left\{\wh{\sigma}_1^2(\bs x) \right\}^{-1} \{\Gamma_1(\bs x) + \Gamma_2(\bs x)\} f_\vartheta(\bs x)q_1(\bs x){\rm d}\bs x \right|
\\
& \le & O(1) \left[\int_\Omega \{\Gamma_1(\bs x) + \Gamma_2(\bs x)\}^2 q_1(\bs x) {\rm d}x \right]^{1/2} \left[\int_\Omega \{f_\vartheta(\bs x) \}^2 q_1(\bs x) {\rm d}x \right]^{1/2}
\\
&\le & O(1) \left( \sum_{a = 0}^1 \| \Theta_a \|_{2} + \|\wh{e} - e \|_2 \times \sum_{a = 0 }^1 \| \wh{\mu}_a - \mu_a\|_2 \right)
\\
& = & o_P(n^{-1/2}),
\ese
by using Conditions \ref{con3} and \ref{con10}. Therefore,
\bsee\label{th3_eq3}
\sqrt{n} \mb P_{n_1} (\wh{\sigma}_1^2)^{-1}(\wh{D} - D_0 )f_\vartheta
& = & \sqrt{n} \mb P_1 (\wh{\sigma}_1^2)^{-1}(\wh{D} - D_0 )f_\vartheta + o_P(1)
\nonumber
\\
& = &\sqrt{n} o_P(n^{-1/2}) + o_P(1)
\nonumber
\\
& = & o_P(1).
\esee
Plugging equations (\ref{th3_eq2})--(\ref{th3_eq3}) into equation (\ref{th3_eq1}), we obtain
\bse
&&-\sqrt{n} \wh{\bs U}_n(\bs \Upsilon_0)[\bs h_{\vartheta}]
\\
& = & \frac{n_1}{n}\sqrt{n} \mb P_{n_1}({\sigma}_{01}^2)^{-1}( D_0 - \tau_0 )f_\vartheta
+ \frac{n_0}{n}\sqrt{n} \mb P_{n_0}({\sigma}_{00}^2)^{-1}( D - \tau_0 -\lambda_0 )
( f_\vartheta + g_\vartheta ) + o_P(1)
\\
& = & \sqrt{ \frac{n_1}{n} } \sqrt{n_1}(\mb P_{n_1} - \mb P_1 )({\sigma}_{01}^2)^{-1}( D_0 - \tau_0 )f_\vartheta
\\
&& + \sqrt{\frac{n_0}{n}} \sqrt{n_0} (\mb P_{n_0} - \mb P_0)({\sigma}_{00}^2)^{-1}( D - \tau_0 -\lambda_0 )
( f_\vartheta + g_\vartheta ) + o_P(1).
\ese
According to Assumptions 1--2 in the main paper and Condition \ref{con3}, the class of functions
$\{({\sigma}_{01}^2)^{-1}( D_0 - \tau_0 )f_\vartheta :  f_\vartheta \in \ms H_{1\vartheta} \}$ is $\mb P_1$-Donsker,
and
$\{({\sigma}_{00}^2)^{-1}( D - \tau_0 - \lambda_0 ) ( f_\vartheta + g_\vartheta )  : \bs h_\vartheta = (f_\vartheta, g_\vartheta) \in \ms H_\vartheta \}$
is $\mb P_0$-Donsker.
Consequently,
$\sqrt{n} \wh{\bs U}_n(\bs \Upsilon_0)[\bs h_{\vartheta}]$ converges to
$\bs G[\bs h_{\vartheta}]$ in distribution. Therefore,
\bse
\sqrt{n}\wh{\bs U}_n(\bs \Upsilon_0) \rightsquigarrow \bs G \quad \textrm{ in } \ell^\infty(\ms H_\vartheta),
\ese
which proves Step 1.

We next prove Step 2.
For $\bs h_\vartheta = (f_\vartheta, g_\vartheta) \in \ms H_\vartheta$,
\bsee\label{th3_eq4}
\dot{\wh{\bs U}}_n(\bs \Upsilon_0)[\wh{\bs \Upsilon}_n - \bs \Upsilon_0][\bs h_\vartheta]
& = & \frac{n_1}{n}\mb P_{n_1} \left(\wh{\sigma}^2_1\right)^{-1}(\wh{\tau}_n - \tau_0)f_\vartheta
\nonumber
\\
&& + \frac{n_0}{n}\mb P_{n_0} \left(\wh{\sigma}^2_0\right)^{-1}(\wh{\tau}_n - \tau_0 + \wh{\lambda}_n - \lambda_0)(f_\vartheta + g_\vartheta),
\nonumber
\\
\dot{\bs U}(\bs \Upsilon_0)[\wh{\bs \Upsilon}_n - \bs \Upsilon_0][\bs h_\vartheta]
& = & \varrho \mb P_{1} \left({\sigma}^2_{01}\right)^{-1}(\wh{\tau}_n - \tau_0)f_\vartheta
\nonumber
\\
&& + (1 - \varrho)\mb P_{0} \left({\sigma}^2_{00}\right)^{-1}(\wh{\tau}_n - \tau_0 + \wh{\lambda}_n - \lambda_0)(f_\vartheta + g_\vartheta).
\esee
For every $\epsilon_1 > 0$, define the set ${\cal B}_1 = \{ \| \wh{\tau}_n - \tau_0 \|_{\infty} \le \epsilon_1^{1/2} \}$ and ${\cal B}_2 = \{ V_1( \tau_n - \tau_0 , \tau_n - \tau_0 ) + J_1(\tau_n - \tau_0 , \tau_n - \tau_0 ) \le 1 \}$.
According to Theorem 4.1, Lemma \ref{lemma6} and Condition \ref{con11}, we have
$\| \wh{\tau}_n - \tau_0 \|_{\infty} \le \wt{c}_1 \gamma_1^{-p / (4m_1)} \| \wh{\tau}_n - \tau_0 \|_{\ch_1} = O_P(\gamma_1^{-p/(4m_1)}\delta_n ) = o_P(1)$, which implies that, for every $\epsilon_2 > 0 $, there exists some $N_1(\epsilon_2)$ such that
$P( {\cal B}_1 ) > 1 - \epsilon_2 / 3$ for all $n > N_1(\epsilon_2)$.
According to Lemma \ref{lemma9}, equation (\ref{th1_eq0}) and Condition \ref{con12}, there exists some $N_2(\epsilon_1, \epsilon_2)$ such that
$P({\cal B}_2) > 1 - \epsilon_2 / 3$ for all $n > N_2(\epsilon_1, \epsilon_2)$. Thus, for every $\epsilon_1 > 0 $ and $\epsilon_2 >0 $, there exists some $N_3(\epsilon_1, \epsilon_2) = \max\{N_1(\epsilon_2) ,N_2(\epsilon_1, \epsilon_2)\}$ such that
\bse
P\left(({\cal B}_1{\cal B}_2 )^C\right) \le P({\cal B}_1^C) + P({\cal B}_2^C) \le 2\epsilon_2 / 3
\ese
for all $n > N_3(\epsilon_1, \epsilon_2)$.
Additionally, on ${\cal B}_1{\cal B}_2$, by Condition \ref{con3},
\bse
\mb P_1 \left\{ \left(\wh{\sigma}^2_1 \right)^{-1}(\wh{\tau}_n - \tau_0)f_\vartheta \right\}^2 \le c_1 \| \wh{\tau_n} - \tau_0 \|_{\infty}^2 \mb P_1 f_\vartheta^2
\le c_2 \| \wh{\tau_n} - \tau_0 \|_{\infty}^2 \le c_2 \epsilon_1,
\ese
combined with lemma 19.24 in van der Vaart (1998), entails that, there exists some $N_4(\epsilon_1, \epsilon_2)$ such that
\bse
P\left(\sqrt{n_1}\left| (\mb P_{n_1} - \mb P_1)\left(\wh{\sigma}^2_1 \right)^{-1}(\wh{\tau}_n - \tau_0)f_\vartheta \right| > \epsilon_1 \right) \le \epsilon_2 / 3.
\ese
for all $n > N_4(\epsilon_1, \epsilon_2)$.
Consequently, for every $\epsilon_1 > 0 $ and $\epsilon_2 > 0 $, there exists some $N(\epsilon_1, \epsilon_2) = \max\{N_3(\epsilon_1, \epsilon_2), N_4(\epsilon_1, \epsilon_2) \}$ such that
\bse
&& P\left(\sqrt{n_1}\left| (\mb P_{n_1} - \mb P_1)\left(\wh{\sigma}^2_1 \right)^{-1}(\wh{\tau}_n - \tau_0)f_\vartheta \right| > \epsilon_1 \right)
\\
& \le &
P\left(\sqrt{n_1}\left| (\mb P_{n_1} - \mb P_1)\left(\wh{\sigma}^2_1 \right)^{-1}(\wh{\tau}_n - \tau_0)f_\vartheta \right| > \epsilon_1 \big | {\cal B}_1{\cal B}_2\right)
 + P\left(({\cal B}_1{\cal B}_2 )^C\right)
\\
&\le &\epsilon_2.
\ese
Immediately,
\bsee\label{th3_eq4a}
\sqrt{n_1}\left| (\mb P_{n_1} - \mb P_1)\left(\wh{\sigma}^2_1 \right)^{-1}(\wh{\tau}_n - \tau_0)f_\vartheta \right| =  o_P(1).
\esee
Next, by Condition \ref{con13} and H\"{o}lder inequality,
\bsee\label{th3_eq7}
\left| \mb P_{1}\left\{\left(\wh{\sigma}^2_1 \right)^{-1} - \left({\sigma}^2_{01} \right)^{-1} \right\}(\wh{\tau}_n - \tau_0)f_\vartheta \right|
&\le & \left\| \left(\wh{\sigma}^2_1 \right)^{-1} - \left({\sigma}^2_{01} \right)^{-1} \right\|_{\infty} \mb P_{1} \left|(\wh{\tau}_n - \tau_0)f_\vartheta \right|
\nonumber
\\
& \le & o_P(1) \left\{\mb P_{1}(\wh{\tau}_n - \tau_0)^2 \right\}^{1/2} \left(\mb P_1 f^2_\vartheta \right)^{1/2}
\nonumber
\\
& \le & o_P(1) \| \wh{\tau}_n - \tau_0 \|_{\ch_1}
\nonumber
\\
& \le & o_P(1) \left\| \wh{\bs \Upsilon}_n - \bs \Upsilon_0 \right\|_{\ch},
\nonumber
\\
\left| \mb P_{1} \left({\sigma}^2_{01} \right)^{-1} (\wh{\tau}_n - \tau_0)f_\vartheta \right|
& \le & O(1) \left\| \wh{\bs \Upsilon}_n - \bs \Upsilon_0 \right\|_{\ch}.
\esee
Thus
\bsee\label{th3_eq5}
&&\sqrt{n}\mb P_{n_1}\left(\wh{\sigma}^2_1 \right)^{-1}(\wh{\tau}_n - \tau_0)f_\vartheta
\nonumber
\\
& = & \sqrt{n}\mb P_{1}\left(\wh{\sigma}^2_1 \right)^{-1}(\wh{\tau}_n - \tau_0)f_\vartheta + o_P(1)
\nonumber
\\
& = &\sqrt{n}\mb P_{1}\left({\sigma}^2_{01} \right)^{-1}(\wh{\tau}_n - \tau_0)f_\vartheta + \sqrt{n}\mb P_{1}\left\{\left(\wh{\sigma}^2_1 \right)^{-1} - \left({\sigma}^2_{01} \right)^{-1} \right\}(\wh{\tau}_n - \tau_0)f_\vartheta + o_P(1)
\nonumber
\\
& = & \sqrt{n}\mb P_{1}\left({\sigma}^2_{01} \right)^{-1}(\wh{\tau}_n - \tau_0)f_\vartheta + \sqrt{n} o_P(1) \left\| \wh{\bs \Upsilon}_n - \bs \Upsilon_0 \right\|_{\ch} + o_P(1)
\nonumber
\\
& = & \sqrt{n}\mb P_{1}\left({\sigma}^2_{01} \right)^{-1}(\wh{\tau}_n - \tau_0)f_\vartheta + o_P\left(1 + \sqrt{n}\left\| \wh{\bs \Upsilon}_n - \bs \Upsilon_0 \right\|_\ch \right).
\esee
By similar arguments that used in equation (\ref{th3_eq5}),
\bsee\label{th3_eq6}
&& \sqrt{n} \mb P_{n_0} \left(\wh{\sigma}^2_0\right)^{-1}(\wh{\tau}_n - \tau_0 + \wh{\lambda}_n - \lambda_0)(f_\vartheta + g_\vartheta)
\nonumber
\\
& = & \sqrt{n}\mb P_{0}\left({\sigma}^2_{00} \right)^{-1}(\wh{\tau}_n - \tau_0 + \wh{\lambda}_n - \lambda_0 )f_\vartheta + o_P\left(1 + \sqrt{n}\left\| \wh{\bs \Upsilon}_n - \bs \Upsilon_0 \right\|_\ch \right).
\esee
Combining equations (\ref{th3_eq4})--(\ref{th3_eq6}), we have
\bse
&& \left|\sqrt{n} \dot{\wh{\bs U}}_n(\bs \Upsilon_0)[\wh{\bs \Upsilon}_n - \bs \Upsilon_0][\bs h_\vartheta] -
\sqrt{n} \dot{\bs U}(\bs \Upsilon_0)[\wh{\bs \Upsilon}_n - \bs \Upsilon_0][\bs h_\vartheta] \right|
\\
& \le & \frac{n_1}{n} \left|\sqrt{n}\mb P_{n_1}\left(\wh{\sigma}^2_1 \right)^{-1}(\wh{\tau}_n - \tau_0)f_\vartheta -  \sqrt{n}\mb P_{1}\left({\sigma}^2_{01} \right)^{-1}(\wh{\tau}_n - \tau_0)f_\vartheta \right\|
\\
&& + \frac{n_0}{n} \left| \sqrt{n} \mb P_{n_0} \left(\wh{\sigma}^2_0\right)^{-1}(\wh{\tau}_n - \tau_0 + \wh{\lambda}_n - \lambda_0)(f_\vartheta + g_\vartheta) - \sqrt{n}\mb P_{0}\left({\sigma}^2_{00} \right)^{-1}(\wh{\tau}_n - \tau_0 + \wh{\lambda}_n - \lambda_0 )f_\vartheta \right|
\\
&& + \left|\frac{n_1}{n} - \rho \right|  \left| \sqrt{n} \mb P_{1}\left({\sigma}^2_{01} \right)^{-1}(\wh{\tau}_n - \tau_0)f_\vartheta \right|
\\
&& + \left|\frac{n_0}{n} - (1 - \rho) \right| \left| \sqrt{n} \mb P_{0}\left({\sigma}^2_{00} \right)^{-1}(\wh{\tau}_n - \tau_0 + \wh{\lambda} - \lambda_0)( f_\vartheta + g_\vartheta)\right|
\\
&\le & o_P\left(1 + \sqrt{n}\left\| \wh{\bs \Upsilon}_n - \bs \Upsilon_0 \right\|_\ch \right),
\ese
holds uniformly for $\bs h_\vartheta \in \ms H_\vartheta$. Immediately,
\bse
\sqrt{n} \dot{\wh{\bs U}}_n(\bs \Upsilon_0)[\wh{\bs \Upsilon}_n - \bs \Upsilon_0]
=
\sqrt{n} \dot{\bs U}(\bs \Upsilon_0)[\wh{\bs \Upsilon}_n - \bs \Upsilon_0]
+ o_P\left(1 + \sqrt{n}\left\| \wh{\bs \Upsilon}_n - \bs \Upsilon_0 \right\|_\ch \right),
\ese
which proves Step 2.

We now show Step 3.
It follows from Lemma \ref{lemma3} that
\bse
&&\dot{\bs U}_{\bs \gamma}(\bs \Upsilon_0)[\bs h_\vartheta][\bs h_\vartheta]
\\
& = & \varrho \mb P_{1} \left({\sigma}^2_{01}\right)^{-1}f^2_\vartheta
+ (1 - \varrho)\mb P_{0} \left({\sigma}^2_{00}\right)^{-1}(f_\vartheta + g_\vartheta)^2
+ \gamma_1 J_1(f_\vartheta, f_\vartheta) + \gamma_0 J_0(g_\vartheta, g_\vartheta)
\\
& \ge & c_0 \{ V_1(f_\vartheta, f_\vartheta) + V_0(g_\vartheta, g_\vartheta) \}
+ \gamma_1 J_1(f_\vartheta, f_\vartheta) + \gamma_0 J_0(g_\vartheta, g_\vartheta)
\\
&\ge & c_1 \left\{ \| f_\vartheta \|_{\ch_1}^2 +  \| g_\vartheta \|_{\ch_0}^2 \right\}
\\
& = &  c_1 \| \bs h_\vartheta \|_{\ch}^2,
\ese
for every $\bs h_\vartheta \in \ms H_\vartheta$ and $\bs h_\vartheta \neq (0 , 0)$,
which shows the invertibility of $\dot{\bs U}_{\bs \gamma}(\bs \Upsilon_0)$.

We further prove Step 4.
According to Lemma \ref{lemma6}, we have $J_1(\gamma_1^{1/2}f_\vartheta, \gamma_1^{1/2}f_\vartheta ) \le O(1)$, $\|\gamma_1^{1/2}f_\vartheta \|_{\infty} \le O(1) $, $J_0(\gamma_0^{1/2}g_\vartheta, \gamma_0^{1/2}g_\vartheta ) \le O(1)$, and $\|\gamma_0^{1/2}g_\vartheta \|_{\infty} \le O(1) $, combined with Condition \ref{con2}, there exists some
$\bs h_{n\vartheta} = (f_{n\vartheta}, g_{n\vartheta}) \in \Phi_n \times \Psi_n$ such that
\bsee\label{th3_eq8}
\left\|\gamma_1^{1/2}f_{n\vartheta} -\gamma_1^{1/2} f_\vartheta \right\|_{\infty} = O(n^{-\kappa_1}), \quad
\sup_{|k| = m_1}\left\|\gamma_1^{1/2}f^{(k)}_{n\vartheta} -\gamma_1^{1/2} f_\vartheta^{(k)} \right\|_{\infty} = O(n^{-\kappa'_1}),
\nonumber
\\
\left\|\gamma_0^{1/2}g_{n\vartheta} - \gamma_0^{1/2}g_\vartheta \right\|_{\infty} = O(n^{-\kappa_0}), \quad
\sup_{|k'| = m_0}\left\|\gamma_0^{1/2}g^{(k')}_{n\vartheta} - \gamma_0^{1/2}g_\vartheta^{(k')} \right\|_{\infty} = O(n^{-\kappa'_0}).
\esee
Note that  $\wh{\bs \Upsilon}_{n}$ is the minimizer of $\wh{\ell}_{n,\gamma_1,\gamma_0}(\bs \Upsilon)$ over $\Psi_n \times \Phi_n$. Thus, for any $\bs h_n = (f_n ,g_n ) \in \Phi_n \times \Psi_n$,
$\wh{\bs U}_{n,\gamma_1, \gamma_0}(\wh{\bs \Upsilon}_n)[\bs h_n] = 0$,
which implies that
\bse
\wh{\bs U}_{n,\gamma_1, \gamma_0}(\wh{\bs \Upsilon}_n)[\bs h_{n\vartheta}] = 0.
\ese
Thus
\bsee\label{th3_eq9}
&& \wh{\bs U}_{n,\gamma_1, \gamma_0}(\wh{\bs \Upsilon}_n)[\bs h_{\vartheta}]
\nonumber
\\
& = & \wh{\bs U}_{n,\gamma_1, \gamma_0}(\wh{\bs \Upsilon}_n)[\bs h_{\vartheta}]
- \wh{\bs U}_{n,\gamma_1, \gamma_0}(\wh{\bs \Upsilon}_n)[\bs h_{n\vartheta}]
\nonumber
\\
& = &
\frac{n_1}{n} \mb P_{n_1} \left( \wh{\sigma}^2_1 \right)^{-1}(\wh{D} - \wh{\tau}_n)(f_{n\vartheta} - f_\vartheta)
+ \frac{n_0}{n} \mb P_{n_0} \left( \wh{\sigma}^2_0 \right)^{-1}({D} - \wh{\tau}_n - \wh{\lambda}_n)(f_{n\vartheta} - f_\vartheta + g_{n\vartheta} - g_\vartheta)
\nonumber
\\
&& + \gamma_1 J_1(\wh{\tau}_n, f_{\vartheta} - f_{n\vartheta})
+ \gamma_0 J_0(\wh{\lambda}_n, g_{\vartheta} - g_{n\vartheta})
\nonumber
\\
& = &
\frac{n_1}{n} \mb P_{n_1} \left( \wh{\sigma}^2_1 \right)^{-1}(\wh{D} - D_0)(f_{n\vartheta} - f_\vartheta)
+ \frac{n_1}{n} \mb P_{n_1} \left( \wh{\sigma}^2_1 \right)^{-1}({D}_0 - \wh{\tau}_n)(f_{n\vartheta} - f_\vartheta)
\nonumber
\\
&&+ \frac{n_0}{n} \mb P_{n_0} \left( \wh{\sigma}^2_0 \right)^{-1}({D} - \wh{\tau}_n - \wh{\lambda}_n)(f_{n\vartheta} - f_\vartheta + g_{n\vartheta} - g_\vartheta)
\nonumber
\\
&& + \gamma_1 J_1(\wh{\tau}_n, f_{\vartheta} - f_{n\vartheta})
+ \gamma_0 J_0(\wh{\lambda}_n, g_{\vartheta} - g_{n\vartheta})
\nonumber
\\
& = &\frac{n_1}{n} I^\dag_{n1} + \frac{n_1}{n} I^\dag_{n2} + \frac{n_0}{n} I^\dag_{n3}
+ I^\dag_{n4} + I^\dag_{n5},
\esee
where $I^\dag_{n1}, \ldots, I^\dag_{n5}$ are self-explained from the above equation.
By Conditions \ref{con3} and \ref{con14}, and equation (\ref{th3_eq8}),
\bse
\mb P_1 \left\{\left( \wh{\sigma}^2_1 \right)^{-1}(\wh{D} - D_0)(f_{n\vartheta} - f_\vartheta) \right\}^2 \le O(1) \|\wh{D} - D_0 \|_{\infty}^2 \| f_{n\vartheta} - f_\vartheta \|^2_{\infty} = o(1),
\ese
combined with lemma 19.24 in van der Vaart (1998), yield
\bsee\label{th3_eq9a1}
\sqrt{n_1}(\mb P_{n_1} - \mb P_1 ) \left( \wh{\sigma}^2_1 \right)^{-1}(\wh{D} - D_0)(f_{n\vartheta} - f_\vartheta)  =  o_P(1).
\esee
Using analogous arguments that used in equations (\ref{th3_eq9a1}) and (\ref{th3_eq4a}), we have
\bse
\sqrt{n_1}(\mb P_{n_1} - \mb P_1 ) \left( \wh{\sigma}^2_1 \right)^{-1}({D}_0 - {\tau}_0)(f_{n\vartheta} - f_\vartheta)  & = &  o_P(1),
\\
\sqrt{n_1} (\mb P_{n_1} - \mb P_1 ) \left( \wh{\sigma}^2_1 \right)^{-1}(\wh{\tau}_n - \tau_0)(f_{n\vartheta} - f_\vartheta) & = & o_P(1).
\ese
Therefore,
\bsee\label{th3_eq9a2}
&& \sqrt{n_1}\left| (\mb P_{n_1} - \mb P_1 ) \left( \wh{\sigma}^2_1 \right)^{-1}({D}_0 - \wh{\tau}_n)(f_{n\vartheta} - f_\vartheta) \right|
\nonumber
\\
& \le & \sqrt{n_1}\left| (\mb P_{n_1} - \mb P_1 ) \left( \wh{\sigma}^2_1 \right)^{-1}({D}_0 - {\tau}_0)(f_{n\vartheta} - f_\vartheta) \right|
\nonumber
\\
&& + \sqrt{n_1} \left| (\mb P_{n_1} - \mb P_1 ) \left( \wh{\sigma}^2_1 \right)^{-1}(\wh{\tau}_n - \tau_0)(f_{n\vartheta} - f_\vartheta) \right|
\nonumber
\\
& = & o_P(1).
\esee
Similarly, we can also derive
\bsee\label{th3_eq9a3}
\sqrt{n_0}(\mb P_{n_0} - \mb P_0) \left( \wh{\sigma}^2_0 \right)^{-1}({D} - \wh{\tau}_n - \wh{\lambda}_n)(f_{n\vartheta} - f_\vartheta + g_{n\vartheta} - g_\vartheta) =  o_P(1).
\esee
Under Conditions \ref{con3} and \ref{con12}, combining equations (\ref{th3_eq9a1})-- (\ref{th3_eq9a3}) and (\ref{th3_eq8}), we have
\bsee\label{th3_eq9a}
| I^\dag_{n1} | &=& \left| (\mb P_{n_1} - \mb P_1 ) \left( \wh{\sigma}^2_1 \right)^{-1}(\wh{D} - D_0)(f_{n\vartheta} - f_\vartheta) + \mb P_1 \left( \wh{\sigma}^2_1 \right)^{-1}(\wh{D} - D_0)(f_{n\vartheta} - f_\vartheta) \right|
\nonumber
\\
& \le & \left|\mb P_1 \left( \wh{\sigma}^2_1 \right)^{-1}(\wh{D} - D_0)(f_{n\vartheta} - f_\vartheta) \right| + o_P(n^{-1/2})
\nonumber
\\
& \le & \| f_{n\vartheta} - f_\vartheta \|_{\infty}  \mb P_1 |\Gamma_{1} + \Gamma_{2}| + o_P(n^{-1/2})
\nonumber
\\
& = & o_P(n^{-1/2}) + O\left(n^{-\kappa_1}\gamma_1^{-1/2} \right) \left( \sum_{a = 0 }^1 \| \Theta_a \|_2
+ \| \wh{e} -e \|_2 \times \sum_{a = 0 }^1 \| \wh{\mu}_a - \mu_a \|_2 \right)
\nonumber
\\
& = & o_P(n^{-1/2}) + O_P( n^{-\kappa_1}\gamma_1^{-1/2} n^{-1/2} )
\nonumber
\\
& = & o_P(n^{-1/2})
\esee
by Condition \ref{con10}, and
\bsee\label{th3_eq9b}
| I^\dag_{n2} | & = &\left| (\mb P_{n_1} - \mb P_1 ) \left( \wh{\sigma}^2_1 \right)^{-1}({D}_0 - \wh{\tau}_n)(f_{n\vartheta} - f_\vartheta) + \mb P_1 \left( \wh{\sigma}^2_1 \right)^{-1}({D}_0 - \wh{\tau}_n)(f_{n\vartheta} - f_\vartheta) \right|
\nonumber
\\
& \le & \left|\mb P_1 \left( \wh{\sigma}^2_1 \right)^{-1}({D}_0 - \wh{\tau}_n)(f_{n\vartheta} - f_\vartheta) \right| + o_P(n^{-1/2})
\nonumber
\\
& = & \left|E_r \left[ \left\{ \wh{\sigma}^2_1(\bs X) \right\}^{-1} E_r \{({D}_0 - \wh{\tau}_n) \mid \bs X, S = 1 \}\{f_{n\vartheta}(\bs X) - f_\vartheta(\bs X)\} \right] \right| + o_P(n^{-1/2})
\nonumber
\\
& \le & O(1) \{\mb P_1(\tau_0 - \wh{\tau}_n)^2 \}^{1/2} \{\mb P_1(f_{n\vartheta} - f_\vartheta)^2 \}^{1/2} + o_P(n^{-1/2})
\nonumber
\\
& \le & O\left(n^{-\kappa_1} \gamma_1^{-1/2} \right) \| \wh{\tau}_n - \tau_0 \|_{\ch_1} + o_P(n^{-1/2})
\nonumber
\\
& = & O_P\left(n^{-\kappa_1} \gamma_1^{-1/2} \delta_n \right) + o_P(n^{-1/2})
\nonumber
\\
& = & o_P(n^{-1/2}),
\nonumber
\\
| I^\dag_{n3} | & = & \bigg| (\mb P_{n_0} - \mb P_0) \left( \wh{\sigma}^2_0 \right)^{-1}({D} - \wh{\tau}_n - \wh{\lambda}_n)(f_{n\vartheta} - f_\vartheta + g_{n\vartheta} - g_\vartheta)
\nonumber
\\
&& + \mb P_0 \left( \wh{\sigma}^2_0 \right)^{-1}({D} - \wh{\tau}_n - \wh{\lambda}_n)(f_{n\vartheta} - f_\vartheta + g_{n\vartheta} - g_\vartheta)
\bigg|
\nonumber
\\
& \le & \left|\mb P_0 \left( \wh{\sigma}^2_0 \right)^{-1}({D} - \wh{\tau}_n - \wh{\lambda}_n)(f_{n\vartheta} - f_\vartheta + g_{n\vartheta} - g_\vartheta)
 \right| + o_P(n^{-1/2})
 \nonumber
\\
& \le & O(1) \left\{\mb P_0(\tau_0 - \wh{\tau}_n)^2 + \mb P_0(\lambda_0 - \wh{\lambda}_n)^2 \right\}^{1/2} \{\mb P_0(f_{n\vartheta} - f_\vartheta)^2 +  \mb P_0(g_{n\vartheta} - g_\vartheta)^2\}^{1/2}
\nonumber
\\
&& + o_P(n^{-1/2})
\nonumber
\\
& = & O_P\left( ( n^{-\kappa_1}\gamma_1^{-1/2} + n^{-\kappa_0}\gamma_0^{-1/2} )\delta_n \right)+ o_P(n^{-1/2})
\nonumber
\\
& = & o_P(n^{-1/2}),
\esee
by Theorem 4.2.

Next, we deal with the terms $I^\dag_{n4}$ and $I^\dag_{n5}$.
By Lemma \ref{lemma7}, equation (\ref{th3_eq8}), Theorem 4.2 and Condition \ref{con12},
\bsee\label{th3_eq9c}
|I^\dag_{n4}| & = & \left|\langle W_{\gamma_1}\wh{\tau}_n, f_{\vartheta} - f_{n\vartheta} \rangle_{\ch_1} \right|
\nonumber
\\
& \le & \left|\langle W_{\gamma_1}(\wh{\tau}_n - \tau_0), f_{\vartheta} - f_{n\vartheta} \rangle_{\ch_1}\right|
 + \left|\langle W_{\gamma_1} \tau_0, f_{\vartheta} - f_{n\vartheta} \rangle_{\ch_1}\right|
\nonumber
\\
& \le &  \left|\langle W_{\gamma_1}(f_{\vartheta} - f_{n\vartheta}), \wh{\tau}_n - \tau_0 \rangle_{\ch_1}\right|
+ \| W_{\gamma_1} \tau_0 \|_{\ch_1} \| f_{\vartheta} - f_{n\vartheta} \|_{\ch_1}
\nonumber
\\
& \le & n^{-\kappa'_1}\gamma_1^{-1/2} \| W_{\gamma_1}n^{\kappa'_1}\gamma_1^{1/2}(f_{\vartheta} - f_{n\vartheta}) \|_{\ch_1} \|\wh{\tau}_n - \tau_0 \|_{\ch_1} + o\left(n^{-\kappa_1} + n^{-\kappa'_1}\gamma_1^{1/2} \right)
\nonumber
\\
&\le& n^{-\kappa'_1}\gamma_1^{-1/2} o(\gamma_1^{1/2}) \delta_n + o\left(n^{-\kappa_1} + n^{-\kappa'_1}\gamma_1^{1/2} \right)
\nonumber
\\
& = & o\left( n^{-\kappa'_1}\delta_n +  n^{-\kappa_1} \right)
\nonumber
\\
&=& o_P(n^{-1/2}).
\esee
In a manner similar to equation (\ref{th3_eq9c}), we can also conclude
\bsee\label{th3_eq9d}
|I^\dag_{n5}| = o_P(n^{-1/2}).
\esee
Plugging equations (\ref{th3_eq9a})--(\ref{th3_eq9d}) into equation (\ref{th3_eq9}), we conclude
\bse
\wh{\bs U}_{n,\gamma_1, \gamma_0}(\wh{\bs \Upsilon}_n)[\bs h_{\vartheta}] = o_P(n^{-1/2}),
\ese
which proves Step 4.

Next, we move on to Step 5.
It follows that
\bse
&& \wh{\bs U}_{n,\gamma_1, \gamma_0}(\wh{\bs \Upsilon}_n)[\bs h_{\vartheta}] - \wh{\bs U}_{n,\gamma_1, \gamma_0}({\bs \Upsilon}_0)[\bs h_{\vartheta}]
\\
& = & \dot{ \wh{\bs U} }_{n,\gamma_1, \gamma_0}({\bs \Upsilon}_0)[\wh{\bs \Upsilon}_n - {\bs \Upsilon}_0][\bs h_{\vartheta}]
\\
& = &  \dot{\wh{\bs U}}_{n}({\bs \Upsilon}_0)[\wh{\bs \Upsilon}_n - {\bs \Upsilon}_0][\bs h_{\vartheta}] + \gamma_1 J_1(\wh{\tau}_n - \tau_0, f_\vartheta) + \gamma_0 J_0(\wh{\lambda}_n - \lambda_0, g_\vartheta),
\ese
combined with equations (\ref{th3_step2}) and (\ref{th3_step4}), we have
\bse
\sqrt{n}\dot{\bs U}_{\bs \gamma }({\bs \Upsilon}_0)[\wh{\bs \Upsilon}_n - {\bs \Upsilon}_0]
&= &-\sqrt{n}\wh{\bs U}_n(\bs \Upsilon_0)  -
\sqrt{n} \bs W_{\bs \gamma}(\bs \Upsilon_0)
\\
&& + o_P\left(1 + \sqrt{n} \left\| \wh{\bs \Upsilon}_n - \bs \Upsilon_0 \right\|_{\ch} \right)
\ese
Noting that the operate $\dot{\bs U}_{\bs \gamma}({\bs \Upsilon}_0)$ is reversible as stated in Step 3, we obtain
\bse
&&\sqrt{n}\dot{\bs U}_{\bs \gamma}({\bs \Upsilon}_0)\left[\wh{\bs \Upsilon}_n - {\bs \Upsilon}_0 + \dot{\bs U}^{-1}_{\bs \gamma}({\bs \Upsilon}_0)[  \bs W_{\bs \gamma}(\bs \Upsilon_0)]\right]
\\
& = & -\sqrt{n}\wh{\bs U}_n(\bs \Upsilon_0)
+ o_P\left(1 + \sqrt{n} \left\| \wh{\bs \Upsilon}_n - \bs \Upsilon_0 \right\|_{\ch} \right).
\ese
By equation (\ref{th3_step1}) and theorem 3.1 in van der Vaart and Wellner (1996), we finally conclude that
\bse
\sqrt{n} \left(\wh{\bs \Upsilon}_n - {\bs \Upsilon}_0 + \dot{\bs U}^{-1}_{\bs \gamma}({\bs \Upsilon}_0) [ \bs W_{\bs \gamma}(\bs \Upsilon_0)]\right) \rightsquigarrow \dot{\bs U}^{-1}_{\bs \gamma}({\bs \Upsilon}_0)[\bs G] \quad \textrm{ in } \ell^\infty(\ms H_\vartheta).
\ese
Thus, we prove Step 5.

Finally, we proceed to the last step.
For $\omega_1 \in \mb R$ and $\omega_0 \in \mb R$, choose appropriate $\vartheta$ such that
$
\bs h_{\bs \omega} = (\omega_1 \gamma_1^{p/(4m_1)} K_{1\bs x_0}, \omega_0 \gamma_0^{p/(4m_0)} K_{0\bs x_0}) \in \ms H_\vartheta
$
and $\dot{\bs U}^{-1}_{\bs \gamma}({\bs \Upsilon}_0)[\bs h_{\bs \omega}] \in \ms H_\vartheta$.
Of note, for $\bs h \in \ch$, we have $\dot{\bs U}_{\bs \gamma}(\bs \Upsilon_0)[\bs h] \ge c_1 \| \bs h \|_{\ch}$ as shown in Step 3, thus $\dot{\bs U}^{-1}_{\bs \gamma}({\bs \Upsilon}_0)[\bs h] \le c_2\| \bs h \|_{\ch}$, and further such an $\vartheta$ satisfies $\dot{\bs U}^{-1}_{\bs \gamma}(\bs \Upsilon_0)[\bs h_{\bs \omega}] \in \ms H_\vartheta$ always exists.
Let
$
(b_{\tau}, b_{\lambda}) = \dot{\bs U}^{-1}_{\bs \gamma}({\bs \Upsilon}_0) [\bs W_{\bs \gamma}(\bs \Upsilon_0)] \in \ch
$,
$\tau_0^*(\bs x_0) = \tau_0(\bs x_0) - b_{\tau}(\bs x_0)$, and
$\lambda_0^*(\bs x_0) = \lambda_0(\bs x_0) - b_{\lambda}(\bs x_0)$.
Then, following equation (\ref{th3_eq11a}), we have
\bse
&& \omega_1 \left[ n^{1/2}\gamma_1^{p/(4m_1)}\{\wh{\tau}_n(\bs x_0) - \tau^*_0(\bs x_0) \} \right] + \omega_0 \left[ n^{1/2}\gamma_0^{p/(4m_0)}\{\wh{\lambda}_n(\bs x_0) - \lambda^*_0(\bs x_0) \} \right]
\\
& = & \sqrt{n} \left(\wh{\bs \Upsilon}_n - {\bs \Upsilon}_0 + \dot{\bs U}^{-1}_{\bs \gamma}({\bs \Upsilon}_0) [ \bs W_{\bs \gamma}(\bs \Upsilon_0)]\right)[\bs h_{\bs \omega}].
\\
&\rightsquigarrow& \dot{\bs U}^{-1}_{\bs \gamma}({\bs \Upsilon}_0)[\bs G][\bs h_\omega]
\\
& = & \bs G[\{ \dot{\bs U}^{-1}_{\bs \gamma}({\bs \Upsilon}_0) \}^* [\bs h_\omega]]
\\
& = &\bs G[\dot{\bs U}^{-1}_{\bs \gamma}({\bs \Upsilon}_0)[\bs h_\omega] ]
\ese
by noting that $\dot{\bs U}_{\bs \gamma}({\bs \Upsilon}_0)$ is the self-adjoint operate, where $\{\dot{\bs U}^{-1}_{\bs \gamma}({\bs \Upsilon}_0)\}^*$ denote the adjoint operate of $\dot{\bs U}^{-1}_{\bs \gamma}({\bs \Upsilon}_0)$.
That is, for every $\omega_1 \in \mb R$ and $\omega_0 \in \mb R$,
$\omega_1 [ n^{1/2}\gamma_1^{p/(4m_1)}\{\wh{\tau}_n(\bs x_0) - \tau^*_0(\bs x_0) \} ] + \omega_0 [ n^{1/2}\gamma_0^{p/(4m_0)}\{\wh{\lambda}_n(\bs x_0) - \lambda^*_0(\bs x_0) \} ]$ converges to a zero-mean Gaussian distribution
with the variance
\bse
&& {\rm Cov}\left(\bs G[\dot{\bs U}^{-1}_{\bs \gamma}({\bs \Upsilon}_0)[\bs h_\omega] ], \bs G[\dot{\bs U}^{-1}_{\bs \gamma}({\bs \Upsilon}_0)[\bs h_\omega] ]
\right)
\\
& = & \varrho \mb P_1 \left(\sigma^2_{01}\right)^{-1}( K_1^* )^2  +
(1 - \varrho) \mb P_0 \left(\sigma^2_{00}\right)^{-1}(K_1^* + K_0^* )^2
\\
& = & \varrho \mb P_1 \left(\sigma^2_{01}\right)^{-1}( K_1^* )^2 +
(1 - \varrho ) \mb P_0 \left(\sigma^2_{00}\right)^{-1}( K_1^*)^2 +
2 (1 - \varrho )\mb P_0 \left(\sigma^2_{00}\right)^{-1}K_1^*K_0^*
\\
&& + (1 - \varrho )\mb P_0 \left(\sigma^2_{00}\right)^{-1}(K_0^* )^2
\\
& = & V_1(K_1^*, K_1^*) + V_0( K_0^*, K_0^* ) + 2(1 - \varrho)
\mb P_0 \left(\sigma^2_{00}\right)^{-1}K_1^*K_0^*,
\ese
where $(K_1^*, K_0^*) = \dot{\bs U}^{-1}_{\bs \gamma}(\bs \Upsilon_0)[(\omega_1 \gamma_1^{p/(4m_1)} K_{1\bs x_0}, \omega_0 \gamma_0^{p/(4m_0)} K_{0\bs x_0})]$.
Therefore,
\bse
\left(
    n^{1/2}\gamma_1^{p/(4m_1)}\{\wh{\tau}_n(\bs x_0) - \tau^*_0(\bs x_0) \} ,
    n^{1/2}\gamma_0^{p/(4m_0)}\{\wh{\lambda}_n(\bs x_0) - \lambda^*_0(\bs x_0) \}
\right)^{\rm T}
\ese
converges to a zero-mean bivariate Gaussian distribution. Thus, we complete the proof of Theorem 3.
\end{proof}

\subsubsection*{E.4 Proof of Corollary 1}
\begin{proof}
Here, we follow the notations that used in the proof of Theorem 3.
Obviously, it is enough to show that
\bsee\label{coro_eq2}
\bs W_{\bs \gamma}(\bs \Upsilon_0)[\bs h_\vartheta] = o_P(n^{-1/2}).
\esee
By Lemma \ref{lemma7} and $n(\gamma_1 + \gamma_0) = O(1)$ presented in Corollary 4.1,
\bse
n^{1/2} \left|\bs W_{\bs \gamma}(\bs \Upsilon_0)[\bs h_\vartheta] \right|
& = & n^{1/2} \left|\langle W_{\gamma_1}\tau_0, h_\vartheta \rangle_{\ch_1} +
\langle W_{\gamma_0}\lambda_0, g_\vartheta \rangle_{\ch_0} \right|
\\
& \le & n^{1/2} \left(\|W_{\gamma_1}\tau_0\|_{\ch_1}\| f_\vartheta \|_{\ch_1} + \|W_{\gamma_0}\lambda_0\|_{\ch_0} \| g_\vartheta \|_{\ch_0} \right)
\\
& \le &  n^{1/2}\left(o(\gamma_1^{1/2}) O(1)
+ o( \gamma_0^{1/2} ) O(1) \right)
\\
& = & o\left( \{n(\gamma_1 + \gamma_0) \}^{1/2} \right)
\\
& = & o(1),
\ese
which implies equation (\ref{coro_eq2}). Thus, we complete the proof of Corollary 1.
\end{proof}

\subsubsection*{E.5 Proof of Theorem 4}\label{subsec4}
\begin{proof}
For two symmetric matrixes $\bs B_1$ and $\bs B_2$, we denote $\bs B_1 \ge \bs B_2$ or $\bs B_2 \le \bs B_1$ if $\bs B_1 - \bs B_2$ is positive semi-definite.
Let $\rho(\bs B_1)$, $\rho_{\min}(\bs B_1)$, and
$\rho_{\max}(\bs B_1)$ be the eigenvalues, minimum, and maximum eigenvalues of $\bs B_1$, respectively.
Let $\bs H$ be an $n \times n $ diagonal matrix with the main diagonal elements being
$( \{\wh{\sigma}^2(\bs X_1, S_1)\}^{-1}, \ldots, \{\wh{\sigma}^2(\bs X_{n}, S_{n})\}^{-1} )^{\rm T}$,
$\bs H_1$ be the first $n_1 \times n_1 $ matrix of $\bs H$,
and $\bs H_2$ be the last $n_0 \times n_0$ matrix of $\bs H$.
Let
$\bs X^* = ( \bs X_1, \ldots, \bs X_n)^{\rm T}$,
$\bs S = (S_1, \ldots, S_n)^{\rm T}$,
$\bs \phi(\bs x) = (\phi_1(\bs x), \ldots, \phi_{r_1}(\bs x))^{\rm T}$,
$\bs \psi(\bs x) = (\psi_1(\bs x), \ldots, \psi_{r_0}(\bs x))^{\rm T} $,
$\bs \Phi(\bs X^*) = (\bs \phi(\bs X_1), \ldots, \bs \phi(\bs X_n))^{\rm T}$,
$\bs \Psi(\bs X^*) = (\bs \psi(\bs X_1), \ldots, \bs \psi(\bs X_n))^{\rm T}$,
and
$ \bs A = (\bs \Phi(\bs X^*), (1 - \bs S)\bs \Psi(\bs X^*))$.
Decompose $\bs A$ into an $2 \times 2$ block matrix form with the $(1,1)$-th elements being $\bs A_1$, $(1,2)$-th element being $\bs 0$, $(2,1)$-th element being $\bs A_{21}$, and
$(2,2)$-th element being $\bs A_2$,
where $\bs A_1$ is an $n_1 \times r_1$ matrix, $\bs A_{21}$is an $n_0 \times r_1$ matrix, and $\bs A_2$ is an $n_0 \times r_0$ matrix.
Let $ \bs P_1$ be an $r_1 \times r_1$ matrix with $(i,j)$-th element being
$J_1(\phi_i, \phi_j)$, where $1 \le i, j \le r_1$,
and $\bs P_2$ be an $r_0 \times r_0$ matrix with $(i',j')$-th element being
$J_0(\psi_{i'}, \psi_{j'})$, where $1 \le i', j' \le r_0$, $\bs P$ be an $2 \times 2$ diagonal block matrix with $(1,1)$-th element being $\bs P_1$ and $(2,2)$-th element being $\bs P_2$, and $\bs P_{\bs \gamma}$ be an $2 \times 2$ diagonal block matrix with $(1,1)$-th element being $\gamma_1 \bs P_1$ and $(2,2)$-th element being $\gamma_0 \bs P_2$.
Let
$\wh{\bs D}_1 = (\wh{D}_1, \ldots, \wh{D}_{n_1})^{\rm T}$,
$\bs D_2 = (D_{n_1 + 1}, \ldots, D_n)^{\rm T}$,
$\wh{\bs D} = (\wh{\bs D}_1^{\rm T}, \bs D_2^{\rm T} )^{\rm T}$.

Note that $\tau(\bs X)$ and $\lambda(\bs X)$ can be  approximated  by the corresponding parameterized functions $\bs \phi(\bs X)^{\rm T}\bs \alpha$ and $\bs \psi(\bs X)^{\rm T}\bs \beta$, respectively, where $\bs \alpha = (\alpha_1, \ldots, \alpha_{r_1})^{\rm T}$ and
$ \bs \beta = (\beta_1, \ldots, \beta_{r_0})^{\rm T}$.
Then minimizing $\wh{\ell}_{n,\gamma_1, \gamma_0}(\tau , \lambda)$ over $\Phi_n \times \Psi_n$ is equivalent to minimizing
\bse
&& \wh{l}_{n,\gamma_1, \gamma_0}(\bs \alpha, \bs \beta )
\\
& = & \left\{\wh{\bs D} - \bs \Phi(\bs X^*)\bs \alpha  - (1 - \bs S)\bs \Psi(\bs X^*)\bs \beta\right\}^{\rm T} \bs H \left\{\wh{\bs D} - \bs \Phi(\bs X^*)\bs \alpha  - (1 - \bs S)\bs \Psi(\bs X^*)\bs \beta \right\}
\\
&& + n \gamma_1 \bs \alpha^{\rm T} \bs P_1 \bs \alpha
+ n \gamma_0 \bs \beta^{\rm T} \bs P_2 \bs \beta.
\ese
Let
$\bs \theta = (\bs \alpha^{\rm T}, \bs \beta^{\rm T})^{\rm T}$, then $\wh{l}_{n,\gamma_1, \gamma_0}(\bs \alpha, \bs \beta )$ can be rewritten as
\bse
\wh{l}_{n,\gamma_1, \gamma_0}(\bs \theta) = \left(\wh{\bs D} - \bs A \bs \theta \right)^{\rm T} \bs H \left(\wh{\bs D} - \bs A \bs \theta\right) + n \bs \theta^{\rm T} \bs P_{\bs \gamma} \bs \theta,
\ese
and the corresponding minimizer can be deduced as
\bse
\wh{\bs \theta}_n = \left(\wh{\bs \alpha}^{\rm T}_n, \wh{\bs \beta}^{\rm T}_n \right)^{\rm T}
 = (\bs A^{\rm T} \bs H \bs A + n \bs P_{\bs \gamma})^{-1}\bs A^{\rm T} \bs H \wh{\bs D}.
\ese
Therefore, for a given point $\bs x_0 \in \Omega$, the estimator of $\tau_0(\bs x_0)$ can be expressed by
\bse
\wh{\tau}_n(\bs x_0) = \bs \phi(\bs x_0)^{\rm T}\wh{\bs \alpha}_n.
\ese
Using an similar argument, $\wh{\tau}_{\rm rct}(\bs x_0)$ can be expressed by
\bse
\wh{\tau}_{\rm rct}(\bs x_0) = \bs \phi(\bs x_0)^{\rm T}\wh{\bs \alpha}_{\rm rct},
\ese
where $\wh{\bs \alpha}_{\rm rct} = (\bs A_1^{\rm T} \bs H_1 \bs A_1 + n \gamma_1\bs P_1)^{-1}\bs A_1^{\rm T} \bs H_1 \wh{\bs D}_1$.
It follows that
\bse
{\rm Var}_r(\wh{\tau}_{\rm rct}(\bs x_0)) = \bs \phi(\bs x_0)^{\rm T}{\rm Var}_r(\wh{\bs \alpha}_{\rm rct} ) \bs \phi(\bs x_0), \
{\rm Var}_r(\wh{\tau}_n(\bs x_0)) = \bs \phi(\bs x_0)^{\rm T}{\rm Var}_r(\wh{\bs \alpha}_n ) \bs \phi(\bs x_0).
\ese

We first consider ${\rm Var}_r(\wh{\tau}_{\rm rct}(\bs x_0))$.
\bse
{\rm Var}_r(\wh{\bs \alpha}_{\rm rct}) & = & E_r \left\{(\bs A_1^{\rm T} \bs H_1 \bs A_1 + n \gamma_1\bs P_1)^{-1}\bs A_1^{\rm T} \bs H_1 {\rm Var}_r( \wh{\bs D}_1 \mid \bs X^* ) \bs H_1 \bs A_1 (\bs A_1^{\rm T} \bs H_1 \bs A_1 + n \gamma_1\bs P_1)^{-1}  \right\}
\\
&& + {\rm Var}_r\left\{ (\bs A_1^{\rm T} \bs H_1 \bs A_1 + n \gamma_1\bs P_1)^{-1}\bs A_1^{\rm T} \bs H_1 E_r( \wh{\bs D}_1 \mid \bs X^* )  \right\}
\\
& = & \bs \Xi_1 + \bs \Xi_2,
\ese
where $\bs \Xi$ and $\bs \Xi_2$ are clear from the above equation, furthermore,
\bse
{\rm Var}_r\{\wh{\tau}_{\rm rct}(\bs x_0)\} = \bs \phi(\bs x_0)^{\rm T} \bs \Xi_1\bs \phi(\bs x_0) + \bs \phi(\bs x_0)^{\rm T} \bs \Xi_2 \bs \phi(\bs x_0).
\ese
Next, we show that $\bs \phi(\bs x_0)^{\rm T} \bs \Xi_2 \bs \phi(\bs x_0)$ is dominated by $\bs \phi(\bs x_0)^{\rm T} \bs \Xi_1 \bs \phi(\bs x_0)$.
Let
$
\bs \zeta^\dag(\bs X^*) = (\bs A_1^{\rm T} \bs H_1 \bs A_1 + n \gamma_1\bs P_1)^{-1}\bs A_1^{\rm T} \bs H_1 \bs A_1 (\bs A_1^{\rm T} \bs H_1 \bs A_1 + n \gamma_1\bs P_1)^{-1}
$. By Condition \ref{con15},
\bsee\label{th4_eqa}
\bs \phi(\bs x_0)^{\rm T} \bs \Xi_1\bs \phi(\bs x_0)  =
E_r\left\{ \bs \phi(\bs x_0)^{\rm T} \bs \zeta^\dag(\bs X^*) \bs \phi(\bs x_0) \right\}
 + o_P(1) E_r\left\{ \bs \phi(\bs x_0)^{\rm T} \bs \zeta^\dag(\bs X^*) \bs \phi(\bs x_0) \right\}.
\esee
Under Conditions \ref{con3} and \ref{con16}, we have
\bse
\rho_{\min}(\bs A_1^{\rm T}\bs H_1 \bs A_1)
 \ge  \rho_{\min}(\bs H_1) \rho_{\min}(\bs A_1^{\rm T}\bs A_1)
& \ge & c_1 \rho_{\min}(\bs A^{\rm T}\bs A)
\\
& \ge & c_2 n^{1 - \kappa^*},
\\
\rho_{\max}(\bs A_1^{\rm T}\bs H_1 \bs A_1)
 \le  \rho_{\max}(\bs H_1 ) \rho_{\max}(\bs A_1^{\rm T}\bs A_1)
& \le & c_3 \rho_{\max}(\bs A^{\rm T}\bs A)
\\
& \le & c_4 n^{1 - \kappa^*},
\ese
that is
\bsee\label{th4_eq1}
\rho(\bs A_1^{\rm T}\bs H_1 \bs A_1) \asymp n^{1 - \kappa^*}.
\esee
Under Conditions \ref{con3}, \ref{con16}, and \ref{con17}, it follow from equation (\ref{th4_eq1}) that
\bse
\rho_{\min}\left\{ (\bs A_1^{\rm T}\bs H_1  \bs A_1 + n \gamma_1 \bs P_1 )^{-1} \right\}
& = & \left\{ \rho_{\max}(\bs A_1^{\rm T}\bs H_1 \bs A_1 + n \gamma_1 \bs P_1 ) \right\}^{-1}
\\
& \ge & \left\{ \rho_{\max}(\bs A_1^{\rm T}\bs H_1 \bs A_1) + n \gamma_1 \rho_{\max}(\bs P_1) \right\}^{-1}
\\
& \ge & \left\{ c_5 n^{1 - \kappa^*} + c_6 \gamma_1 n^{1 + \kappa_1^\dag} \right\}^{-1}
\\
& \ge & c_7 n^{\kappa^* - 1 }
\\
\rho_{\max}\left\{ (\bs A_1^{\rm T} \bs H_1 \bs A_1 + n \gamma_1 \bs P_1 )^{-1} \right\}
& = & \left\{ \rho_{\min}(\bs A_1^{\rm T} \bs H_1 \bs A_1 + n \gamma_1 \bs P_1 ) \right\}^{-1}
\\
& \le & \left\{ c_8 n^{1 - \kappa^*} + n \gamma_1 \rho_{\min}(\bs P_1) \right\}^{-1}
\\
&\le& c_9 n^{\kappa^* - 1 },
\ese
that is,
\bsee\label{th4_eq2}
\rho\left\{ (\bs A_1^{\rm T} \bs H_1 \bs A_1 + n \gamma_1 \bs P_1 )^{-1} \right\}
\asymp n^{\kappa^* - 1 }.
\esee
Thus
\bse
\bs \phi(\bs x_0)^{\rm T} \bs \zeta^\dag(\bs X^*) \bs \phi(\bs x_0)
& \le & \bs \phi(\bs x_0)^{\rm T} (\bs A_1^{\rm T} \bs H_1 \bs A_1)^{-1} \bs \phi(\bs x_0)
\\
&\le & \rho_{\max}\left\{ (\bs A_1^{\rm T} \bs H_1 \bs A_1)^{-1} \right\} \bs \phi(\bs x_0)^{\rm T}\bs \phi(\bs x_0)
\\
& = & \left\{\rho_{\min}(\bs A_1^{\rm T} \bs H_1 \bs A_1)\right\}^{-1} \bs \phi(\bs x_0)^{\rm T}\bs \phi(\bs x_0)
\\
& \le & c_{10} n^{\kappa^* - 1} \bs \phi(\bs x_0)^{\rm T}\bs \phi(\bs x_0),
\\
\bs \phi(\bs x_0)^{\rm T} \bs \zeta^\dag(\bs X^*) \bs \phi(\bs x_0)
& \ge & \bs \phi(\bs x_0)^{\rm T} (\bs A_1^{\rm T}\bs H_1 \bs A_1 + n \gamma_1 \bs P_1)^{-1} \bs \phi(\bs x_0)
\\
&\ge & \rho_{\min}\left\{ (\bs A_1^{\rm T} \bs H_1 \bs A_1 + n \gamma_1 \bs P_1 )^{-1} \right\} \bs \phi(\bs x_0)^{\rm T}\bs \phi(\bs x_0)
\\
& \ge &c_{11} n^{\kappa^* - 1} \bs \phi(\bs x_0)^{\rm T}\bs \phi(\bs x_0).
\ese
As a consequence,
\bsee\label{th4_eq4}
\bs \phi(\bs x_0)^{\rm T} \bs \Xi_1\bs \phi(\bs x_0) \asymp n^{\kappa^* - 1} \bs \phi(\bs x_0)^{\rm T}\bs \phi(\bs x_0).
\esee

We next consider $\bs \phi(\bs x_0)^{\rm T} \bs \Xi_2 \bs \phi(\bs x_0)$.
Let $\Gamma(\bs X_i) = \Gamma_1(\bs X_i) + \Gamma_2(\bs X_i)$, $i = 1, \ldots, n_1$, $\bs \Gamma(\bs X^*_1) = (\Gamma(\bs X_1), \ldots, \Gamma(\bs X_{n_1}))^{\rm T}$, and $\bs \tau_0(\bs X^*_1) = (\tau_0(\bs X_1), \ldots, \tau_0(\bs X_{n_1}))$,
then by Lemma \ref{lemma2},
\bsee\label{th4_eq5}
&&\bs \phi(\bs x_0)^{\rm T} \bs \Xi_2 \bs \phi(\bs x_0)
\nonumber
\\
& = & \bs \phi(\bs x_0)^{\rm T} {\rm Var}_r\left[ (\bs A_1^{\rm T} \bs H_1 \bs A_1 + n \gamma_1\bs P_1)^{-1}\bs A_1^{\rm T} \bs H_1 \left\{E_r( \wh{\bs D}_1 - \bs D_1  \mid \bs X^* ) + E_r( \bs D_1 \mid \bs X^* ) \right\}  \right] \bs \phi(\bs x_0)
\nonumber
\\
& = & \bs \phi(\bs x_0)^{\rm T} {\rm Var}_r\left[ (\bs A_1^{\rm T} \bs H_1 \bs A_1 + n \gamma_1\bs P_1)^{-1}\bs A_1^{\rm T} \bs H_1 \left\{\bs \Gamma(\bs X^*_1) + \bs \tau_0(\bs X_1^*) \right\}  \right] \bs \phi(\bs x_0)
\nonumber
\\
& \le & 2 \bs \phi(\bs x_0)^{\rm T} {\rm Var}_r\left[ (\bs A_1^{\rm T} \bs H_1 \bs A_1 + n \gamma_1\bs P_1)^{-1}\bs A_1^{\rm T} \bs H_1 \bs \Gamma(\bs X^*_1)  \right] \bs \phi(\bs x_0)
\nonumber
\\
&& + 2 \bs \phi(\bs x_0)^{\rm T} {\rm Var}_r\left\{ (\bs A_1^{\rm T} \bs H_1 \bs A_1 + n \gamma_1\bs P_1)^{-1}\bs A_1^{\rm T} \bs H_1  \bs \tau_0(\bs X_1^*)  \right\} \bs \phi(\bs x_0)
\nonumber
\\
& = & 2 I^\ddag_{n1} + 2 I_{n2}^\ddag,
\esee
where $I_{n1}^\ddag$ and $I_{n2}^\ddag$ are clear from the above equation.
Let
$\bs w_j$ denote the $j$-th column of $(\bs A_1^{\rm T} \bs H_1 \bs A_1 + n \gamma_1\bs P_1)^{-1}\bs A_1^{\rm T} \bs H_1$, then
$(\bs w_1, \ldots, \bs w_{n_1}) = (\bs A_1^{\rm T} \bs H_1 \bs A_1 + n \gamma_1\bs P_1)^{-1}\bs A_1^{\rm T} \bs H_1$, then by equation (\ref{th4_eq2}), Conditions \ref{con3}, \ref{con10}, and \ref{con16},
\bsee\label{th4_eq5a}
I^\ddag_{n1} & = & {\rm Var}_r \left\{\bs \phi(\bs x_0)^{\rm T} \sum_{j = 1}^{n_1}\bs w_j \Gamma(\bs X_j) \right\}
\nonumber
\\
& = & {\rm Var}_r \left\{ \sum_{j = 1}^{n_1}\Gamma(\bs X_j)\bs \phi(\bs x_0)^{\rm T} \bs w_j  \right\}
\nonumber
\\
& \le & E_r \left\{ \sum_{j = 1}^{n_1}\Gamma(\bs X_j)\bs \phi(\bs x_0)^{\rm T} \bs w_j  \right\}^2
\nonumber
\\
& \le & E_r \left[ \sum_{j = 1}^{n_1} (\Gamma(\bs X_j))^2  \sum_{j = 1}^{n_1} \{ \bs \phi(\bs x_0)^{\rm T} \bs w_j \}^2 \right]
\nonumber
\\
& = & E_r \left[ \sum_{j = 1}^{n_1} (\Gamma(\bs X_j))^2  \sum_{j = 1}^{n_1} \{ \bs \phi(\bs x_0)^{\rm T} \bs w_j \bs w_j^{\rm T} \bs \phi(\bs x_0)  \right]
\nonumber
\\
& = & E_r \bigg[ \left\{\sum_{j = 1}^{n_1} (\Gamma(\bs X_j))^2 \right\} \bs \phi(\bs x_0)^{\rm T} (\bs A_1^{\rm T} \bs H_1 \bs A_1 + n \gamma_1\bs P_1)^{-1}\bs A_1^{\rm T} \bs H_1^2 \bs A_1
\nonumber
\\
&& \times (\bs A_1^{\rm T} \bs H_1 \bs A_1 + n \gamma_1\bs P_1)^{-1}
\bs \phi(\bs x_0)  \bigg]
\nonumber
\\
& \le & \left[\rho_{\max}\left\{ (\bs A_1^{\rm T} \bs H_1 \bs A_1 + n \gamma_1\bs P_1)^{-1} \right\} \right]^2 \rho_{\max}(\bs A_1^{\rm T} \bs A_1 ) \rho_{\max}(\bs H_1^2) \bs \phi(\bs x_0)^{\rm T} \bs \phi(\bs x_0)
\nonumber
\\
&& \times n_1 E_r \{ \Gamma_1(\bs X) + \Gamma_2(\bs X) \}^2
\nonumber
\\
& \le & O(1) n^{2 ( \kappa^* - 1 } ) n^{1 - \kappa^*} \bs \phi(\bs x_0)^{\rm T} \bs \phi(\bs x_0) n_1 o_P(n^{-1})
\nonumber
\\
& = & o_P\left( n^{\kappa^* - 1}\right)  \bs \phi(\bs x_0)^{\rm T} \bs \phi(\bs x_0).
\esee

For $I^\ddag_{n2}$,
\bsee\label{th4_eq5b_a}
I^\ddag_{n2} & = & \bs \phi(\bs x_0)^{\rm T} {\rm Var}_r\left[ (\bs A_1^{\rm T} \bs H_1 \bs A_1 + n \gamma_1\bs P_1)^{-1}\bs A_1^{\rm T} \bs H_1  \{\bs \tau_0(\bs X_1^*) - \bs \tau_n(\bs X_1^*) + \bs \tau_n(\bs X_1^*) \}  \right] \bs \phi(\bs x_0)
\nonumber
\\
& \le & 2\bs \phi(\bs x_0)^{\rm T} {\rm Var}_r\left[ (\bs A_1^{\rm T} \bs H_1 \bs A_1 + n \gamma_1\bs P_1)^{-1}\bs A_1^{\rm T} \bs H_1  \{\bs \tau_0(\bs X_1^*) - \bs \tau_n(\bs X_1^*) \}  \right] \bs \phi(\bs x_0)
\nonumber
\\
&& + 2 \bs \phi(\bs x_0)^{\rm T} {\rm Var}_r\left\{ (\bs A_1^{\rm T} \bs H_1 \bs A_1 + n \gamma_1\bs P_1)^{-1}\bs A_1^{\rm T} \bs H_1  \bs \tau_n(\bs X_1^*)  \right\} \bs \phi(\bs x_0)
\nonumber
\\
& = & 2 I^\ddag_{n2, 1} + 2 I^\ddag_{n2, 2},
\esee
where $I^\ddag_{n2, 1}$ and $I^\ddag_{n2, 2}$ are self-explained from the above equation and
\bse
\bs \tau_n(\bs X_1^*) = (\tau_n(\bs X_1), \ldots, \tau_n(\bs X_{n_1}))^{\rm T}.
\ese
Employing  equations (\ref{th1_eq0}), (\ref{th4_eq2}), and Conditions \ref{con3}, \ref{con16}--\ref{con17}, and
analogous arguments as we did in equation (\ref{th4_eq5a}), we conclude
\bsee\label{th4_eq5b_b}
I^\ddag_{n2, 1} = O_P\left(n^{1 - 2 \kappa_1} \right)n^{\kappa^* - 1 }\bs \phi(\bs x_0)^{\rm T} \bs \phi(\bs x_0)
= o_P\left(n^{\kappa^* - 1 } \right)  \bs \phi(\bs x_0)^{\rm T} \bs \phi(\bs x_0).
\esee
We now deal with $I^\ddag_{n2, 2}$.
Noting that $\tau_n \in \Phi_n$, then $\tau_n(\bs x) $ can be written as $\tau_n(\bs x) = \bs \phi(\bs x)^{\rm T} \bs \alpha_n$ for some $\bs \alpha_n = (\alpha_{n1}, \ldots, \alpha_{nr_1})^{\rm T}$. Thus $\bs \tau_n(\bs X^*_1) = \bs A_1 \bs \alpha_n$, and further by equation (\ref{th1_eq0}) and (\ref{th4_eq2}), Conditions \ref{con16}--\ref{con17},
\bsee\label{th4_eq5b_c}
I^\ddag_{n2, 2} & = & \bs \phi(\bs x_0)^{\rm T} {\rm Var}_r\left\{ (\bs A_1^{\rm T} \bs H_1 \bs A_1 + n \gamma_1\bs P_1)^{-1}\bs A_1^{\rm T} \bs H_1  \bs A_1 \bs \alpha_n  \right\} \bs \phi(\bs x_0)
\nonumber
\\
& = & \bs \phi(\bs x_0)^{\rm T} {\rm Var}_r\big[ \left\{(\bs A_1^{\rm T} \bs H_1 \bs A_1)^{-1} - (\bs A_1^{\rm T} \bs H_1 \bs A_1)^{-1} n \gamma_1 \bs P_1 (\bs A_1^{\rm T} \bs H_1 \bs A_1 + n \gamma_1\bs P_1)^{-1} \right\}
\nonumber
\\
&& \times \bs A_1^{\rm T} \bs H_1  \bs A_1 \bs \alpha_n  \big] \bs \phi(\bs x_0)
\nonumber
\\
& = & \bs \phi(\bs x_0)^{\rm T} {\rm Var}_r\left\{\bs \alpha_n - n \gamma_1 (\bs A_1^{\rm T} \bs H_1 \bs A_1)^{-1} \bs P_1 (\bs A_1^{\rm T} \bs H_1 \bs A_1 + n \gamma_1\bs P_1)^{-1} \bs A_1^{\rm T} \bs H_1  \bs A_1 \bs \alpha_n \right\}\bs \phi(\bs x_0)
\nonumber
\\
& = & n^2 \gamma_1^2 \bs \phi(\bs x_0)^{\rm T}{\rm Var}_r\left\{(\bs A_1^{\rm T} \bs H_1 \bs A_1)^{-1} \bs P_1 (\bs A_1^{\rm T} \bs H_1 \bs A_1 + n \gamma_1\bs P_1)^{-1} \bs A_1^{\rm T} \bs H_1  \bs A_1 \bs \alpha_n \right\}\bs \phi(\bs x_0)
\nonumber
\\
& \le & n^2 \gamma_1^2 \bs \phi(\bs x_0)^{\rm T}E_r \bigg\{(\bs A_1^{\rm T} \bs H_1 \bs A_1)^{-1} \bs P_1 (\bs A_1^{\rm T} \bs H_1 \bs A_1 + n \gamma_1\bs P_1)^{-1} \bs A_1^{\rm T} \bs H_1  ( \bs A_1 \bs \alpha_n )
\nonumber
\\
&& \times (\bs \alpha^{\rm T}\bs A_1^{\rm T}) \bs H_1 \bs A_1 (\bs A_1^{\rm T} \bs H_1 \bs A_1 + n \gamma_1\bs P_1)^{-1} \bs P_1 (\bs A_1^{\rm T} \bs H_1 \bs A_1)^{-1} \bigg\}\bs \phi(\bs x_0)
\nonumber
\\
& \le & O(1) n^2 \gamma_1^2 \rho_{\max}\left\{ \bs \tau_n(\bs X^*) \bs \tau_n(\bs X^*)^{\rm T} \right\}  \left[\rho_{\max}\left\{(\bs A_1^{\rm T} \bs H_1 \bs A_1)^{-1} \right\} \right]^2
\left\{\rho_{\max}(\bs P_1) \right\}^2
\nonumber
\\
&& \times \left[\rho_{\max}\left\{(\bs A_1^{\rm T} \bs H_1 \bs A_1 + n\gamma_1\bs P_1)^{-1} \right\} \right]^2
\rho_{\max}\left(\bs A_1^{\rm T} \bs H_1 \bs A_1 \right)
\bs \phi(\bs x_0)^{\rm T} \bs \phi(\bs x_0)
\nonumber
\\
& \le & O(1) n^2 \gamma_1^2 n n^{2(\kappa^* - 1)} n^{2\kappa^\dag} n^{2(\kappa^* - 1)} n^{1 - \kappa^*}
\bs \phi(\bs x_0)^{\rm T} \bs \phi(\bs x_0)
\nonumber
\\
& = & O\left(n^{1 + 2(\kappa^* + \kappa^\dag)}\gamma_1^2 \right) n^{\kappa^* - 1 } \bs \phi(\bs x_0)^{\rm T} \bs \phi(\bs x_0)
\nonumber
\\
& = & o\left( n^{\kappa^* - 1 } \right) \bs \phi(\bs x_0)^{\rm T} \bs \phi(\bs x_0).
\esee
Plugging equations (\ref{th4_eq5b_b}) and (\ref{th4_eq5b_c}) into equation (\ref{th4_eq5b_a}), we have
\bsee\label{th4_eq5b}
I^\ddag_{n2} = o_P\left( n^{\kappa^* - 1 } \right) \bs \phi(\bs x_0)^{\rm T} \bs \phi(\bs x_0).
\esee
By equations (\ref{th4_eq5}), (\ref{th4_eq5a}), and (\ref{th4_eq5b}),
\bse
\bs \phi(\bs x_0)^{\rm T} \bs \Xi_2 \bs \phi(\bs x_0) = o_P\left( n^{\kappa^* - 1 } \right) \bs \phi(\bs x_0)^{\rm T} \bs \phi(\bs x_0),
\ese
combined with equation (\ref{th4_eq4}), show that $\bs \phi(\bs x_0)^{\rm T} \bs \Xi_2 \bs \phi(\bs x_0)$ is dominated by $\bs \phi(\bs x_0)^{\rm T} \bs \Xi_1 \bs \phi(\bs x_0)$,
moreover, ${\rm Var}_r\{\wh{\tau}_{\rm rct}(\bs x_0) \} \asymp n^{\kappa^* - 1 } \bs \phi(\bs x_0)^{\rm T} \bs \phi(\bs x_0)$ and
\bsee\label{th4_eq6}
&& {\rm Var}_r\left\{\wh{\tau}_{\rm rct}(\bs x_0) \right\}
\nonumber
\\
& = & \bs \phi(\bs x_0)^{\rm T} E_r \left\{ (\bs A_1^{\rm T} \bs H_1 \bs A_1 + n \gamma_1\bs P_1)^{-1}\bs A_1^{\rm T} \bs H_1 \bs A_1 (\bs A_1^{\rm T} \bs H_1 \bs A_1 + n \gamma_1\bs P_1)^{-1} \right\} \bs \phi(\bs x_0)
\nonumber
\\
&& +  o_P\left(n^{\kappa^* - 1 } \right) \bs \phi(\bs x_0)^{\rm T}\bs \phi(\bs x_0).
\esee
By Conditions \ref{con15}--\ref{con16}, and equation (\ref{th4_eq2}),
\bse
&& n \gamma_1 \bs \phi(\bs x_0)^{\rm T} (\bs A_1^{\rm T} \bs H_1 \bs A_1 + n \gamma_1\bs P_1)^{-1} \bs P_1 (\bs A_1^{\rm T} \bs H_1 \bs A_1 + n \gamma_1\bs P_1)^{-1} \bs \phi(\bs x_0)
\\
& \le & n \gamma_1 \rho_{\max}(\bs P_1) \rho_{\max}\left\{ (\bs A_1^{\rm T} \bs H_1 \bs A_1 + n \gamma_1\bs P_1)^{-1} \right\}^2 \bs \phi(\bs x_0)^{\rm T} \bs \phi(\bs x_0)
\\
& \le & O(1) n\gamma_1 n^{\kappa^\dag}  n^{2(\kappa^* - 1 )} \bs \phi(\bs x_0)^{\rm T} \bs \phi(\bs x_0)
\\
& = & o( n^{\kappa^*  - 1 } ) \bs \phi(\bs x_0)^{\rm T} \bs \phi(\bs x_0),
\\
&& n \gamma_1 \bs \phi(\bs x_0)^{\rm T} (\bs A_1^{\rm T} \bs H_1 \bs A_1 )^{-1} \bs P_1 (\bs A_1^{\rm T} \bs H_1 \bs A_1 + n \gamma_1\bs P_1)^{-1} \bs \phi(\bs x_0)
\\
& \le & n \gamma_1 \rho_{\max}(\bs P_1)
\rho_{\max}\left\{ (\bs A_1^{\rm T} \bs H_1 \bs A_1 )^{-1} \right\}
\rho_{\max}\left\{ (\bs A_1^{\rm T} \bs H_1 \bs A_1 + n \gamma_1\bs P_1)^{-1} \right\} \bs \phi(\bs x_0)^{\rm T} \bs \phi(\bs x_0)
\\
& = & o( n^{\kappa^*  - 1 } ) \bs \phi(\bs x_0)^{\rm T} \bs \phi(\bs x_0),
\ese
then equation (\ref{th4_eq6}) is  deduced by
\bsee\label{th4_eq6a}
&& {\rm Var}_r \left\{\wh{\tau}_{\rm rct}(\bs x_0) \right\}
\nonumber
\\
& = & \bs \phi(\bs x_0)^{\rm T} E_r \big\{ (\bs A_1^{\rm T} \bs H_1 \bs A_1 + n \gamma_1\bs P_1)^{-1}\bs A_1^{\rm T} \bs H_1 (\bs A_1 + n\gamma_1 \bs P_1 - n \gamma_1 \bs P_1 )
\nonumber
\\
&& \times (\bs A_1^{\rm T} \bs H_1 \bs A_1 + n \gamma_1\bs P_1)^{-1} \big\} \bs \phi(\bs x_0) +  o_P\left(n^{\kappa^* - 1 } \right) \bs \phi(\bs x_0)^{\rm T}\bs \phi(\bs x_0)
\nonumber
\\
& = & \bs \phi(\bs x_0)^{\rm T} E_r \big \{ (\bs A_1^{\rm T} \bs H_1 \bs A_1 + n \gamma_1\bs P_1)^{-1} \big\} \bs \phi(\bs x_0)
\nonumber
\\
&& - n \gamma_1 \bs \phi(\bs x_0)^{\rm T} E_r \big \{ (\bs A_1^{\rm T} \bs H_1 \bs A_1 + n \gamma_1\bs P_1)^{-1} \bs P_1 (\bs A_1^{\rm T} \bs H_1 \bs A_1 + n \gamma_1\bs P_1)^{-1} \big\} \bs \phi(\bs x_0)
\nonumber
\\
&&+ o_P\left(n^{\kappa^* - 1 } \right) \bs \phi(\bs x_0)^{\rm T}\bs \phi(\bs x_0)
\nonumber
\\
& = & \bs \phi(\bs x_0)^{\rm T} E_r \big \{ (\bs A_1^{\rm T} \bs H_1 \bs A_1 + n \gamma_1\bs P_1)^{-1} \big\} \bs \phi(\bs x_0)+ o_P\left(n^{\kappa^* - 1 } \right) \bs \phi(\bs x_0)^{\rm T}\bs \phi(\bs x_0)
\nonumber
\\
& = & \bs \phi(\bs x_0)^{\rm T} E_r \big \{ (\bs A_1^{\rm T} \bs H_1 \bs A_1)^{-1} \big\} \bs \phi(\bs x_0) + o_P\left(n^{\kappa^* - 1 } \right) \bs \phi(\bs x_0)^{\rm T}\bs \phi(\bs x_0)
\nonumber
\\
&& - n \gamma_1 \bs \phi(\bs x_0)^{\rm T} E_r \big \{ (\bs A_1^{\rm T} \bs H_1 \bs A_1)^{-1} \bs P_1 (\bs A_1^{\rm T} \bs H_1 \bs A_1 + n \gamma_1\bs P_1)^{-1} \big\} \bs \phi(\bs x_0)
\nonumber
\\
& = &\bs \phi(\bs x_0)^{\rm T} E_r \big \{ (\bs A_1^{\rm T} \bs H_1 \bs A_1)^{-1} \big\} \bs \phi(\bs x_0)+ o_P\left(n^{\kappa^* - 1 } \right) \bs \phi(\bs x_0)^{\rm T}\bs \phi(\bs x_0).
\esee

Mimicking the discussion of equation (\ref{th4_eq6a}), we can also conclude that
\bse
&& {\rm Var}_r\left\{\omega_1 \wh{\tau}_n(\bs x_0) + \omega_2 \wh{\lambda}_n(\bs x_0) \right\}
\\
& = & (\omega_1 \bs \phi(\bs x_0)^{\rm T}, \omega_2 \bs \psi(\bs x_0)^{\rm T} ) E_r \left\{(\bs A^{\rm T} \bs H \bs A)^{-1}  \right\}(\omega_1 \bs \phi(\bs x_0)^{\rm T}, \omega_2 \bs \psi(\bs x_0)^{\rm T} )^{\rm T}
\\
&& +  o_P\left(n^{\kappa^* - 1 } \right) (\omega_1 \bs \phi(\bs x_0)^{\rm T}, \omega_2 \bs \psi(\bs x_0)^{\rm T} )(\omega_1 \bs \phi(\bs x_0)^{\rm T}, \omega_2 \bs \psi(\bs x_0)^{\rm T} )^{\rm T}
\ese
for every $\omega_1, \omega_2 \in \mb R$.
Let
\bsee\label{th4_eq7}
 \left(
  \begin{array}{cc}
    \bs \Pi_1 &  \bs \Pi_{12}\\
    \bs \Pi_{12}^{\rm T} & \bs \Pi_{2} \\
  \end{array}
\right)
 =  (\bs A^{\rm T} \bs H \bs A)^{-1} =
\left(
  \begin{array}{cc}
    \bs A_1^{\rm T} \bs H_1 \bs A_1 + \bs A_{21}^{\rm T} \bs H_2 \bs A_{21}
     &  \bs A_{21}^{\rm T} \bs H_2 \bs A_{2} \\
    \bs A_2^{\rm T}\bs H_2 \bs A_{21} & \bs A_2^{\rm T}\bs H_2 \bs A_2 \\
  \end{array}
\right)^{-1}.
\esee
Taking $\omega_1 = 1$ and $\omega_2 = 0 $, we have
\bsee\label{th4_eq9}
{\rm Var}_r \{ \wh{\tau}_n(\bs x_0 ) \}
 =  \bs \phi(\bs x_0)^{\rm T} E_r \left( \bs \Pi_1 \right)\bs \phi(\bs x_0)
 + o_P\left(n^{\kappa^* - 1 } \right) \bs \phi(\bs x_0)^{\rm T}\bs \phi(\bs x_0).
\esee
By equation (\ref{th4_eq7}) and some routine calculation,
\bse
\bs \Pi_1  =  \big\{\bs A_1^{\rm T} \bs H_1\bs A_1 + \bs A_{21}^{\rm T}\bs H_2 \bs A_{21}
- \bs A_{21}^{\rm T}\bs H_2 \bs A_2 ( \bs A_2^{\rm T}\bs H_2 \bs A_2 )^{-1}\bs A_2^{\rm T}\bs H_2\bs A_{21}\big\}^{-1}.
\ese
To prove Theorem 4.4, it is enough to show that
$
\bs \Pi_1 \le (\bs A_1^{\rm T} \bs H_1 \bs A_1 )^{-1}
$,
that is
\bsee\label{th4_eq10}
&& \bs \Pi_1^{-1} - \bs A_1^{\rm T}\bs H_1 \bs A_1
\nonumber
\\
& = &  \bs A_{21}^{\rm T}\bs H_2 \bs A_{21} - \bs A_{21}^{\rm T}\bs H_2 \bs A_2 ( \bs A_2^{\rm T}\bs H_2 \bs A_2)^{-1}\bs A_2^{\rm T}\bs H_2 \bs A_{21}
\nonumber
\\
& \ge &  \bs 0.
\esee
Let $\bs A_2^* = \bs H_2^{1/2} \bs A_2 ( \bs A_2^{\rm T}\bs H_2 \bs A_2)^{-1}\bs A_2^{\rm T}\bs H_2^{1/2}$, then we have $ ( \bs A_2^* )^2 = \bs A_2^* $, which implies that
$\rho_{\max} (\bs A_2^*) = 1$. Therefore, for  every $ \bs \omega = (\omega_1, \ldots, \omega_{r_1})^{\rm T} $,
\bse
&& \bs \omega^{\rm T } \bs A_{21}^{\rm T}\bs H_2 \bs A_2 ( \bs A_2^{\rm T}\bs H_2 \bs A_2)^{-1}\bs A_2^{\rm T}\bs H_2 \bs A_{21} \bs \omega
\\
& = & ( \bs H_2^{1/2} \bs A_{21} \bs \omega )^{\rm T} \bs A_2^* \bs H_2^{1/2} \bs A_{21} \bs \omega
\\
& \le & \rho_{\max} (\bs A_2^*)( \bs H_2^{1/2} \bs A_{21} \bs \omega )^{\rm T}\bs H_2^{1/2} \bs A_{21} \bs \omega
\\
& = & \bs \omega^{\rm T} \bs A_{21}^{\rm T}\bs H_2 \bs A_{21} \bs \omega,
\ese
which proves equation (\ref{th4_eq10}).
Consequently,
\bse
{\rm Var}_r \{ \wh{\tau}_{\rm rct}(\bs x_0) \} \ge {\rm Var}_r \{ \wh{\tau}_n(\bs x_0) \}
+ o_P\left( n^{\kappa^* - 1} \right) \bs \phi(\bs x_0)^{\rm T} \bs \phi(\bs x_0).
\ese
Specifically,
\bse
{\rm Var}_r \{ \wh{\tau}_{\rm rct}(\bs x_0) \} & = & \bs \phi(\bs x_0)^{\rm T} E_r (\bs \Sigma_{\rm rct}) \bs \phi(\bs x_0) + o_P\left( n^{\kappa^* - 1} \right) \bs \phi(\bs x_0)^{\rm T} \bs \phi(\bs x_0),
\\
{\rm Var}_r \{ \wh{\tau}_{n}(\bs x_0) \} & = & \bs \phi(\bs x_0)^{\rm T} E_r(\bs \Sigma_{\rm int}) \bs \phi(\bs x_0) + o_P\left( n^{\kappa^* - 1} \right) \bs \phi(\bs x_0)^{\rm T} \bs \phi(\bs x_0),
\ese
where
\bse
\bs \Sigma_{\rm rct} & = & \left(\bs A_1^{\rm T} \bs H_1 \bs A_1 \right)^{-1},
\\
\bs \Sigma_{\rm int} & = & \left(\bs A_1^{\rm T} \bs H_1\bs A_1 + \bs A_{21}^{\rm T}\bs H_2 \bs A_{21}
- \bs A_{21}^{\rm T}\bs H_2 \bs A_2 ( \bs A_2^{\rm T}\bs H_2 \bs A_2 )^{-1}\bs A_2^{\rm T}\bs H_2\bs A_{21}\right)^{-1},
\ese
and
\bse
\bs \Sigma_{\rm int}^{-1} - \bs \Sigma_{\rm rct}^{-1} = \bs A_{21}^{\rm T}\bs H_2 \bs A_{21} - \bs A_{21}^{\rm T}\bs H_2 \bs A_2 ( \bs A_2^{\rm T}\bs H_2 \bs A_2)^{-1}\bs A_2^{\rm T}\bs H_2 \bs A_{21} \ge \bs 0.
\ese
Thus, we conclude Theorem 4.
\end{proof}

\subsection{Appendix F: Additional simulation results}
\subsubsection{F.1 Simulation results of Case 2 and Case 3 in Section 5.1 of the main paper}
\begin{figure}[H]
\centering
\includegraphics[width = \textwidth]{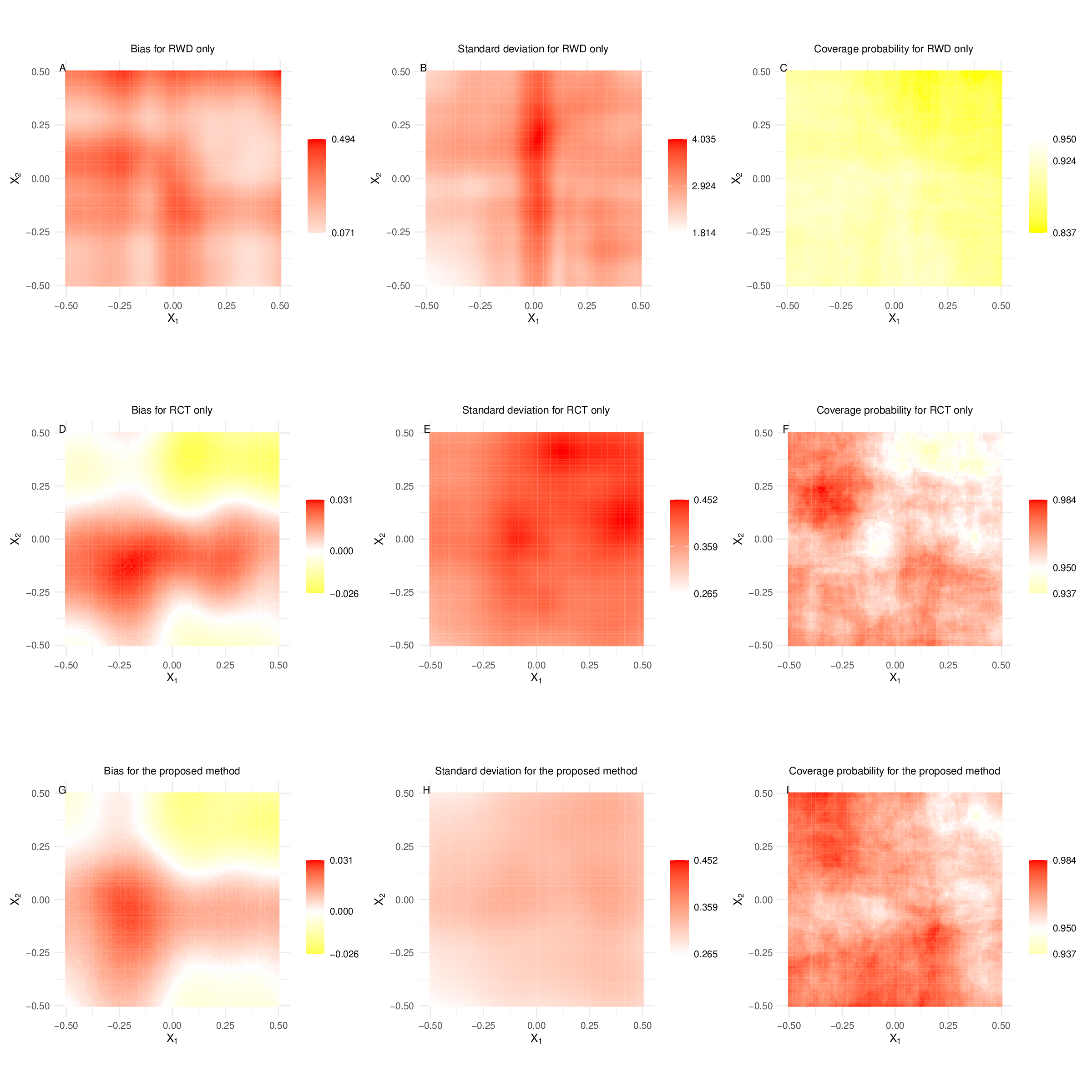}
\caption{The simulation results of Case 2 with $(n_1, n_0) = (500, 1000)$.}
\label{fig-s1}
\end{figure}
\begin{figure}[H]
\centering
\includegraphics[width = \textwidth]{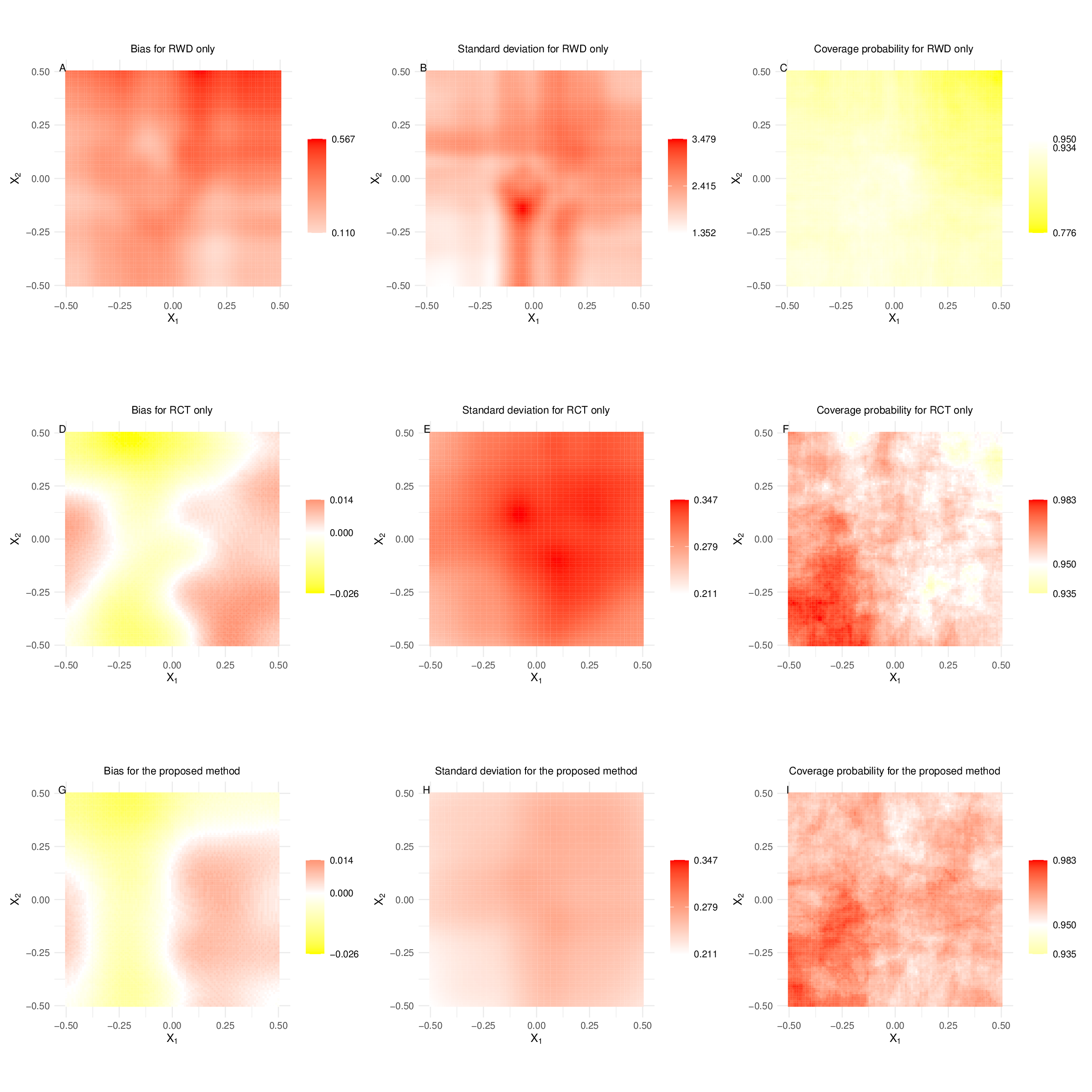}
\caption{The simulation results of Case 2 with $(n_1, n_0) = (1000, 2000)$.}
\label{fig-s2}
\end{figure}
\begin{figure}[H]
\centering
\includegraphics[width = \textwidth]{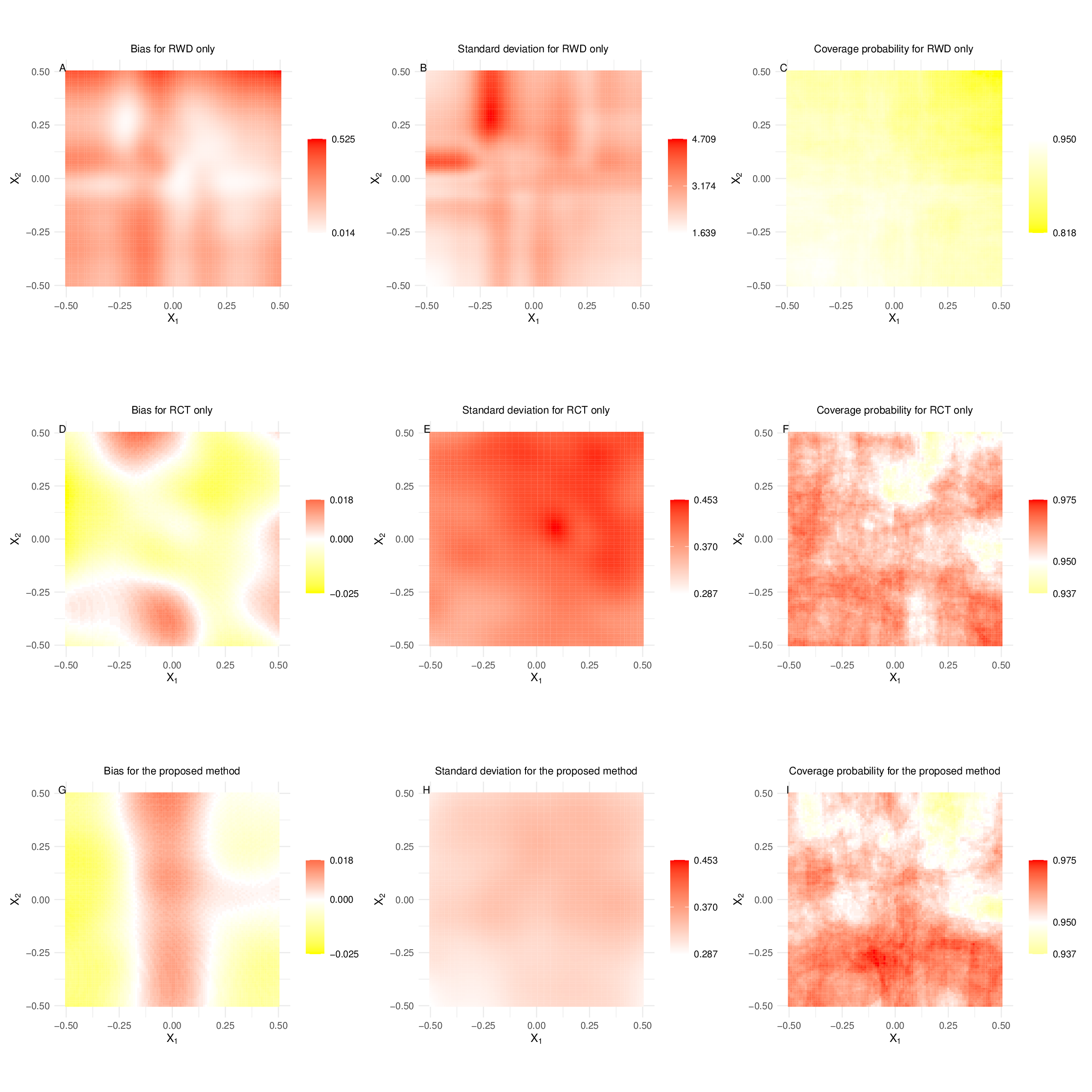}
\caption{The simulation results of Case 3 with $(n_1, n_0) = (500, 1000)$.}
\label{fig-s3}
\end{figure}
\begin{figure}[H]
\centering
\includegraphics[width = \textwidth]{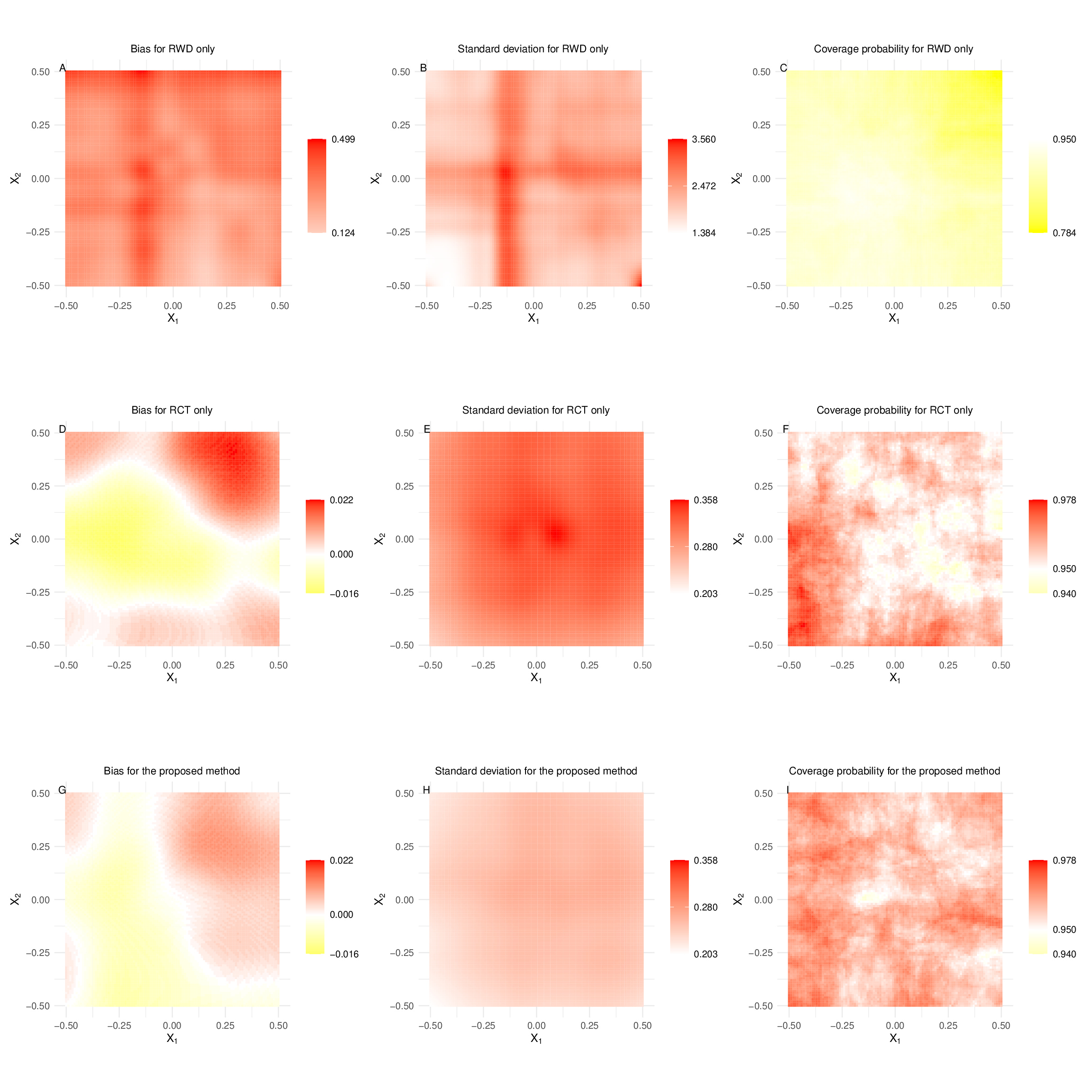}
\caption{The simulation results of Case 3 with $(n_1, n_0) = (1000, 2000)$.}
\label{fig-s4}
\end{figure}

\newpage
\subsubsection*{F.2 Additional simulation study}
Here, we consider the case where $p = 4$, i.e., $\bs X = (X_1, X_2, X_3, X_4)^{\rm T}$, where
$X_1$ and $X_2$ are continuous, and $X_3$ and $X_4$ are binary.
The estimation methodology employed in this additional simulation is the same as the approach detailed in the main paper. Specifically, we stratify the data into four distinct groups based on the values of the covariates $X_3$ and $X_4$.
Given that both $X_3$ and $X_4$ are binary, this stratification results in four possible combinations: $(X_3 = 0, X_4 = 0)$, $(X_3 = 1, X_4 = 0)$, $(X_3 = 0, X_4 = 1)$, and $(X_3 = 1, X_4 = 1)$.
For each of these four groups, applying the same method as described in the main paper,
we can get the results of RCT data-based only method and those of the RWD-based only method.
However, it is important to note that in the implementation of the proposed integrative method, we stratify the RCT data into four groups based on the values of $X_3$ and $X_4$.  In contrast, the RWD is not stratified, instead, the entire dataset is used for combining each of four groups from the RCT data.
Specifically, the information from $X_3$ and $X_4$ in the RWD is not utilized.
For the RWD,  the covariates information we used is only from the continuous covariates $X_1$ and $X_2$. This further demonstrates that the proposed method allows for inconsistency in data structure between the RCT data and the RWD.

Throughout the simulation, $ e(\bs X) = P(A = 1 \mid \bs X, S = 1) = 0.5 $ is known.
The censoring rate is set around $20\%$ for the RCT data and $70\%$ for the RWD.
Corresponding to the scenarios in the simulation study of the main paper, we also consider three cases, and the specific settings are as follows.
\begin{description}
\item[Case S1.]
The generations of $X_u$, $X_1$, and $X_2$ are the same as in Case 1 of Section 5.1 in the main paper.
For the RCT data, $X_3$ and $X_4$ are generated from a Bernoulli distribution with a success probability of $0.5$. In contrast, for the RWD, $X_3$ is generated with a success probability of $0.6$, and $X_4$ with a success probability of $0.4$.
For both the RCT data and the RWD,  $T$ is generated from the survival function
\bse
&&G_{T}(t \mid X_1, X_2, X_3, X_4, X_u, A )\\
& = & (1 + 0.02t)\exp\bigg[ -0.1X_u t -0.845t \exp\big\{ -0.5A + (0.6A - 0.3)\exp(1.5X_1)
\\
&&+ (0.3 - 0.6A)\exp(1.5X_2) + 0.5X_3 - 0.5X_4 \big\}\bigg],
\ese
and the censoring time $C$ is generated from a Cox proportional hazards model with the conditional hazard function taking the form
$
h_C(t \mid X_1, X_2, X_3, X_4) = h_{0C}(t) \exp( 0.5 X_1 + 0.5 X_2 - 0.5X_3 -0.5X_4)
$,
where $h_{0C}(t)$ is the baseline hazard function.
We set $ h_{0C}(t) = 0.052 $ with a study duration of $4.5$ for the RCT data,
and $ h_{0C}(t) = 3.164 $ with a study duration of $3.5$ for the RWD to achieve the preset censoring rates. We consider the restricted time point $L = 2$, the sample size $(n_1, n_0) = (500, 1000)$ and $(1000, 2000)$. For each configuration, 1000 simulations are repeated.

\item[Case S2.]
For the RCT data, the failure time $T$ is generated from a Cox model
\bse
&& h_{T}(t \mid X_1, X_2, X_3, X_4, A )
\\
& = & 0.845 \exp\big\{ -0.5A + (0.6A - 0.3)\exp(1.5X_1)
\\
&& + (0.3 - 0.6A)\exp(1.5X_2) + 0.5X_3 - 0.5X_4 \big\}.
\ese
In contrast, for the RWD, $X_u$ is generated from a standard normal distribution and $T$ is generated from a Cox model
\bse
&& h_{T}(t \mid X_1, X_2, X_3, X_4, X_u, A )
\\
& = & 0.845 \exp\big\{ -0.5A + (0.6A - 0.3)\exp(1.5X_1)
\\
&& + (0.3 - 0.6A)\exp(1.5X_2) + 0.5X_3 - 0.5X_4 + X_u \big\}.
\ese
We set $ h_{0C}(t) = 0.055$ for the RCT data and $ h_{0C}(t) = 3.587 $ for the RWD.
The remaining setups are the same as in Case S1.

\item[Case S3.]
For both the RCT data and the RWD, $T$ is generated from a Cox model
\bse
&& h_{T}(t \mid X_1, X_2, X_3, X_4, X_u, A )
\\
& = & 0.845 \exp\big\{ -0.5A + (0.6A - 0.3)\exp(1.5X_1)
\\
&& + (0.3 - 0.6A)\exp(1.5X_2) + 0.5X_3 - 0.5X_4 + X_u \big\}.
\ese
We set $ h_{0C}(t) = 0.02 $ with a study duration of $5.5$ for the RCT data.
The remaining setups are the same as in Case S2.
\end{description}

Case S1 corresponds to the situation where the failure time follows Cox type (non-standard Cox proportional hazards model), and $ E(T_L \mid \bs X, A, S = 1) = E(T_L \mid \bs X, A, S = 0) $.
Case S2 corresponds to the situation where the failure time follows Cox type (non-standard Cox proportional hazards model), but $ E(T_L \mid \bs X, A, S = 1) \neq E(T_L \mid \bs X, A, S = 0) $.
Case S3 addresses the situation where the failure time does not follow Cox type, and $ E(T_L \mid \bs X, A, S = 1) = E(T_L \mid \bs X, A, S = 0) $.
The simulation results are summarized in Figures \ref{fig-5}--\ref{fig-28}, , which allow us to draw conclusions similar to those in Section 5.1 of the main paper.
However, the coverage performance in this additional simulation is not as favorable as that in the main paper. This may be attributed to two main reasons.
Firstly, the HTE is identified based on the RCT data, which plays a crucial role in the estimation process. In this additional simulation, we stratify the RCT data into four groups and apply the estimation method to each subgroup separately. As a result, the sample size for each subgroup is significantly reduced, which directly impacts the coverage probability.
Secondly, in this additional simulation study, the standard Cox proportional hazards model is used to fit the survival data. Thus,
the failure time model is always misspecified, which impacts the coverage performance.
Even though, the coverage probability still reaches above $90\%$ at almost all evaluation points. This is a commendable result for a fully nonparametric method dealing with four-dimensional covariates, including two discrete variables, as well as a very complex true function form. It demonstrates the robustness and effectiveness of the proposed integrative method even under less favorable conditions.

\begin{figure}
\centering
\includegraphics[width = \textwidth]{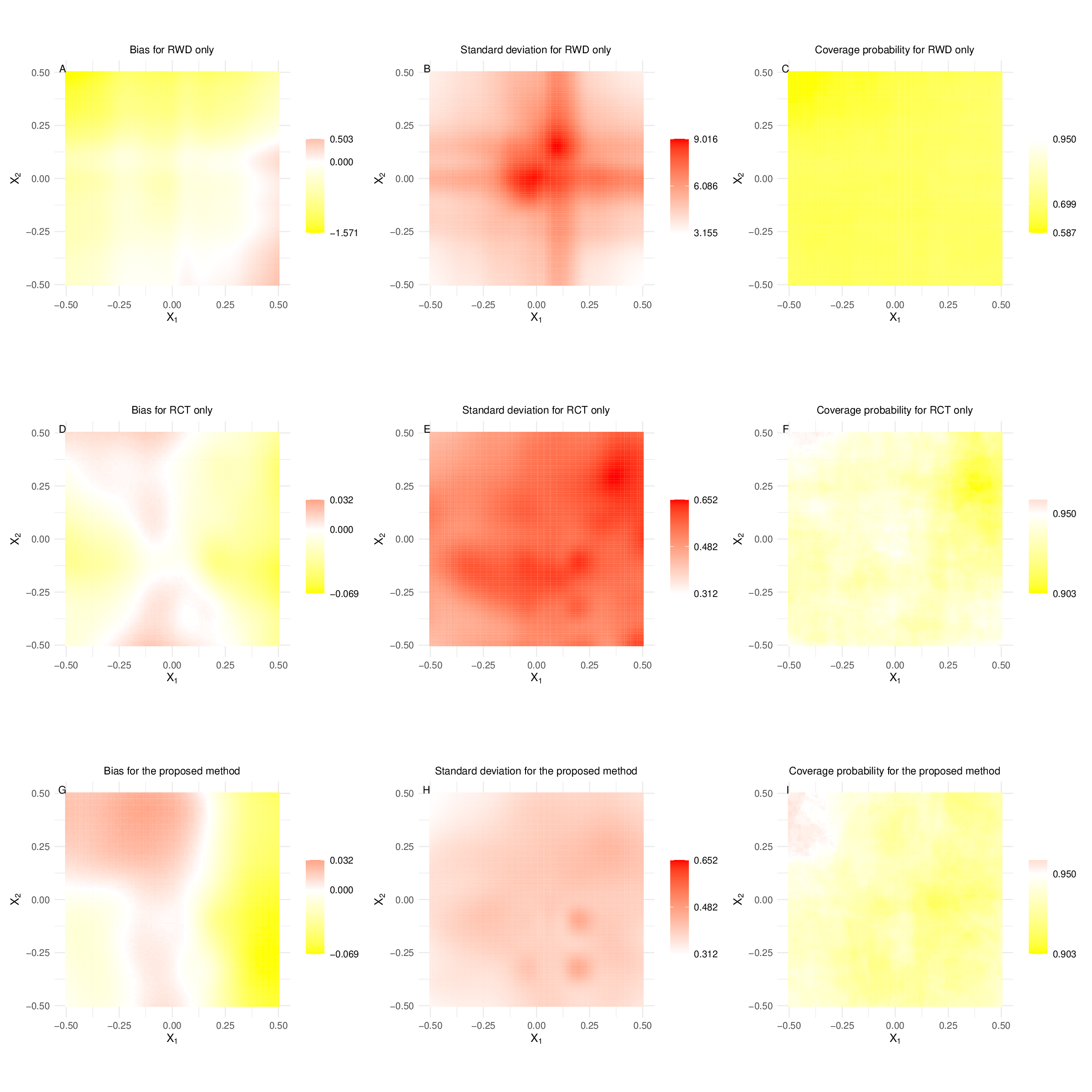}
\caption{The simulation results of Case S1 with $(n_1, n_0) = (500, 1000)$, $X_3 = 0 $, and $X_4 = 0$.}
\label{fig-5}
\end{figure}
\begin{figure}
\centering
\includegraphics[width = \textwidth]{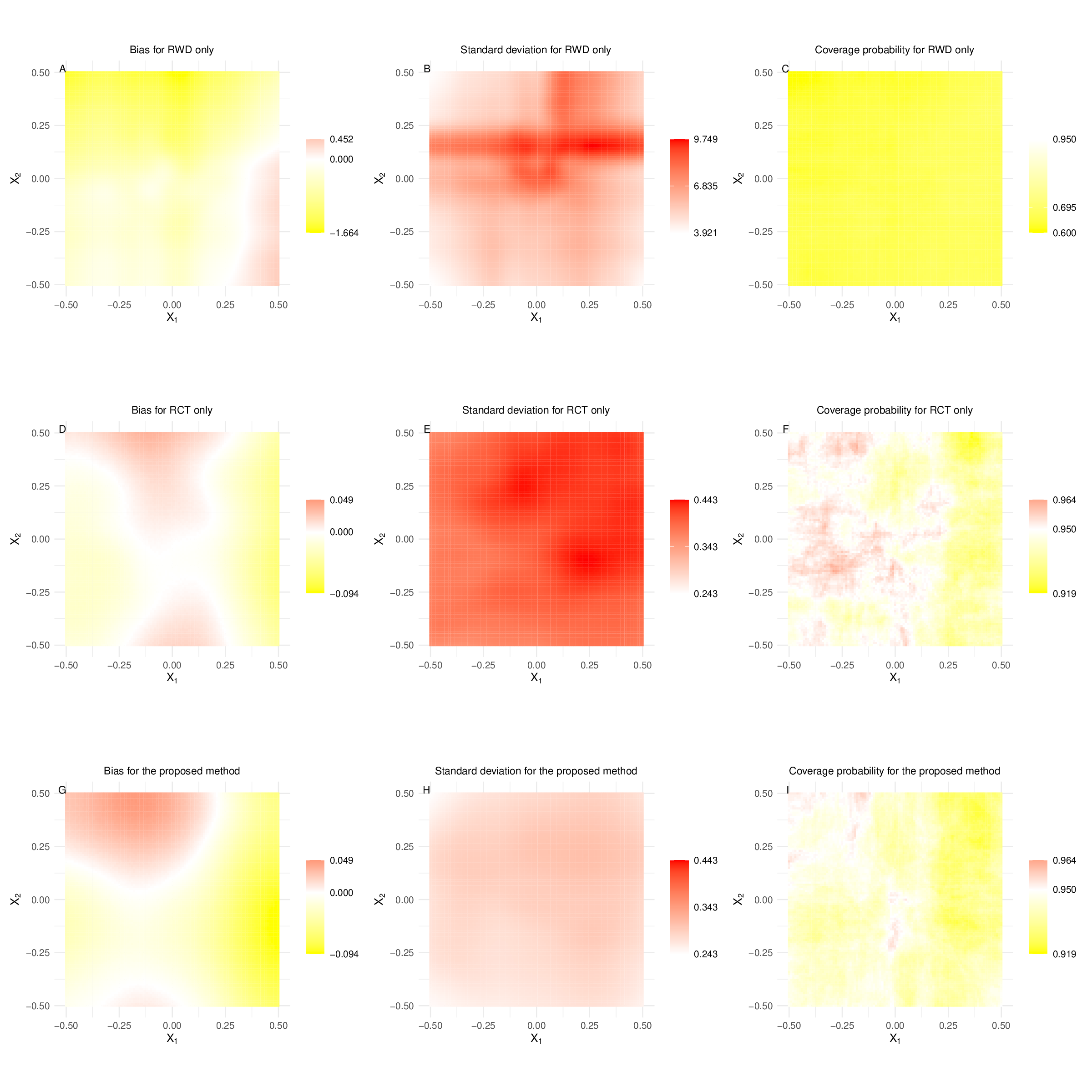}
\caption{The simulation results of Case S1 with $(n_1, n_0) = (1000, 2000)$, $X_3 = 0 $, and $X_4 = 0$.}
\label{fig-6}
\end{figure}
\begin{figure}
\centering
\includegraphics[width = \textwidth]{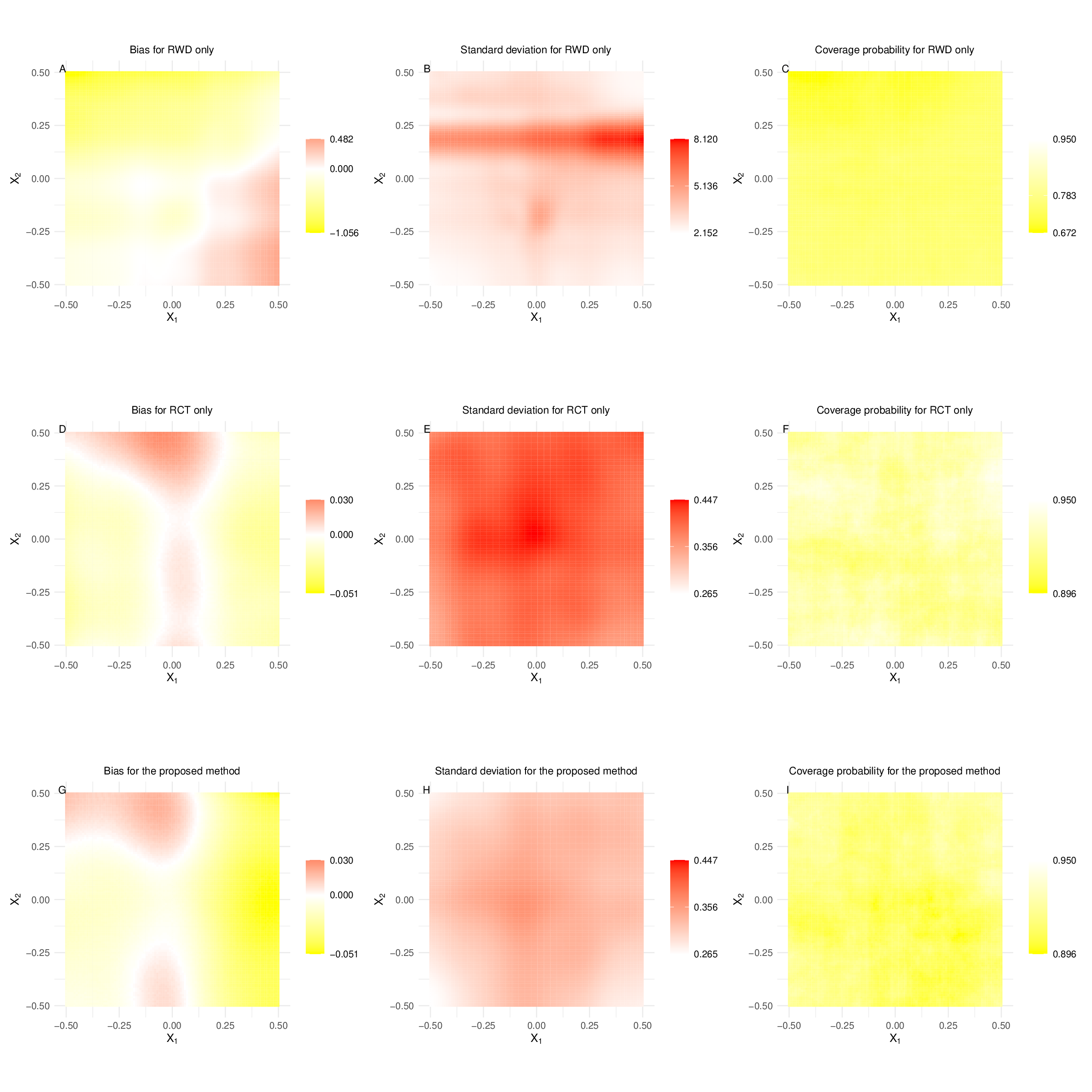}
\caption{The simulation results of Case S1 with $(n_1, n_0) = (500, 1000)$, $X_3 = 1 $, and $X_4 = 0$.}
\label{fig-7}
\end{figure}
\begin{figure}
\centering
\includegraphics[width = \textwidth]{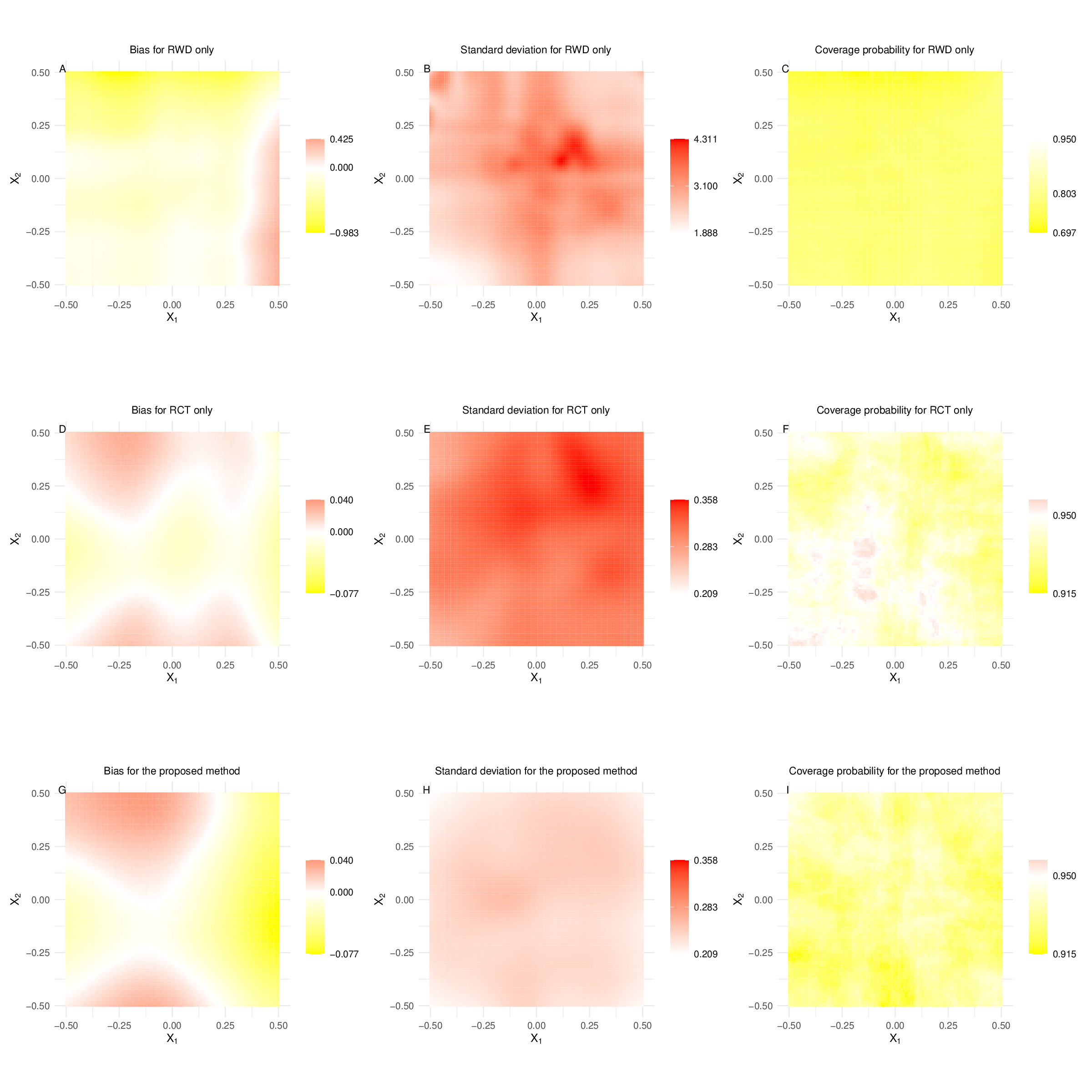}
\caption{The simulation results of Case S1 with $(n_1, n_0) = (1000, 2000)$, $X_3 = 1 $, and $X_4 = 0$.}
\label{fig-8}
\end{figure}
\begin{figure}
\centering
\includegraphics[width = \textwidth]{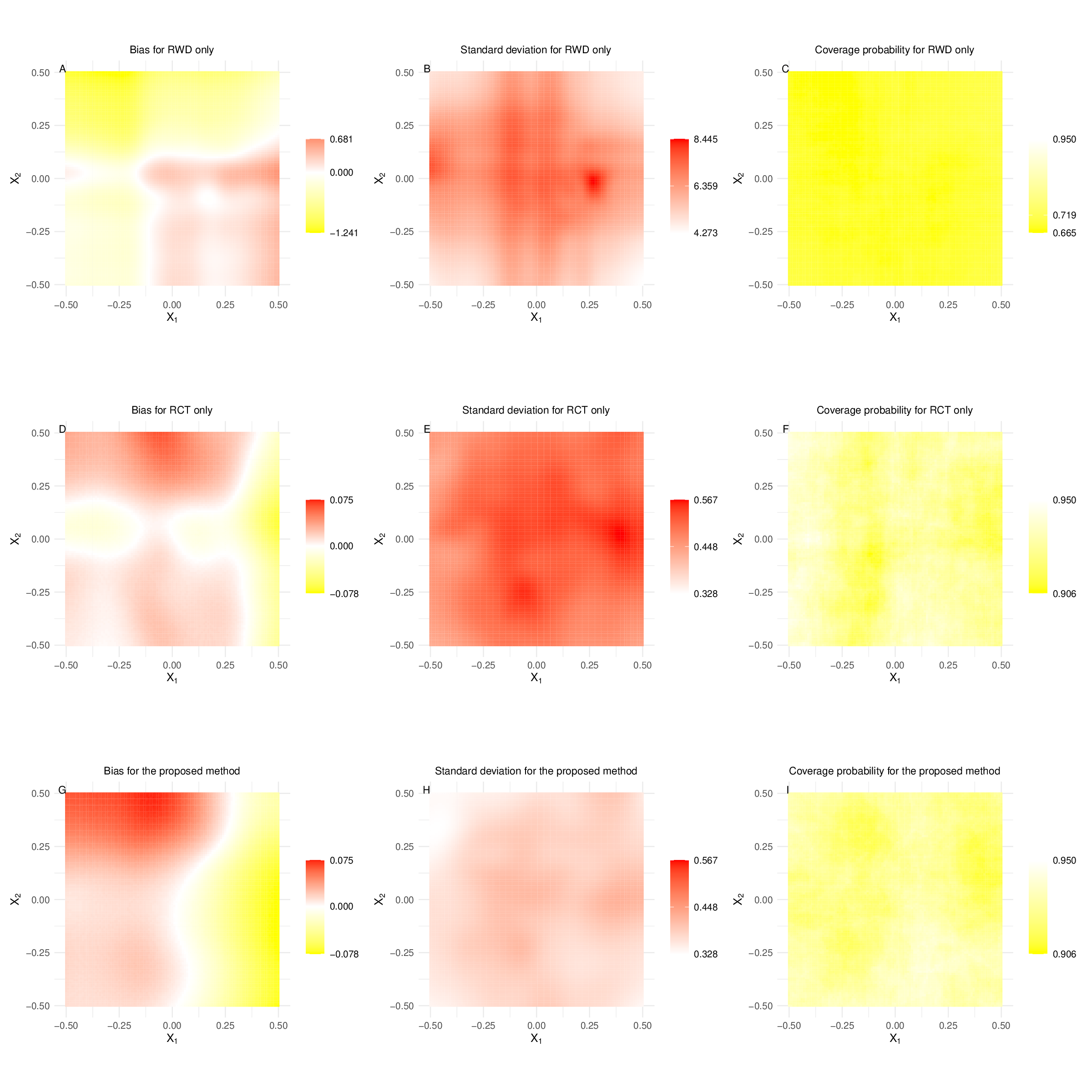}
\caption{The simulation results of Case S1 with $(n_1, n_0) = (500, 1000)$, $X_3 = 0 $, and $X_4 = 1$.}
\label{fig-9}
\end{figure}
\begin{figure}
\centering
\includegraphics[width = \textwidth]{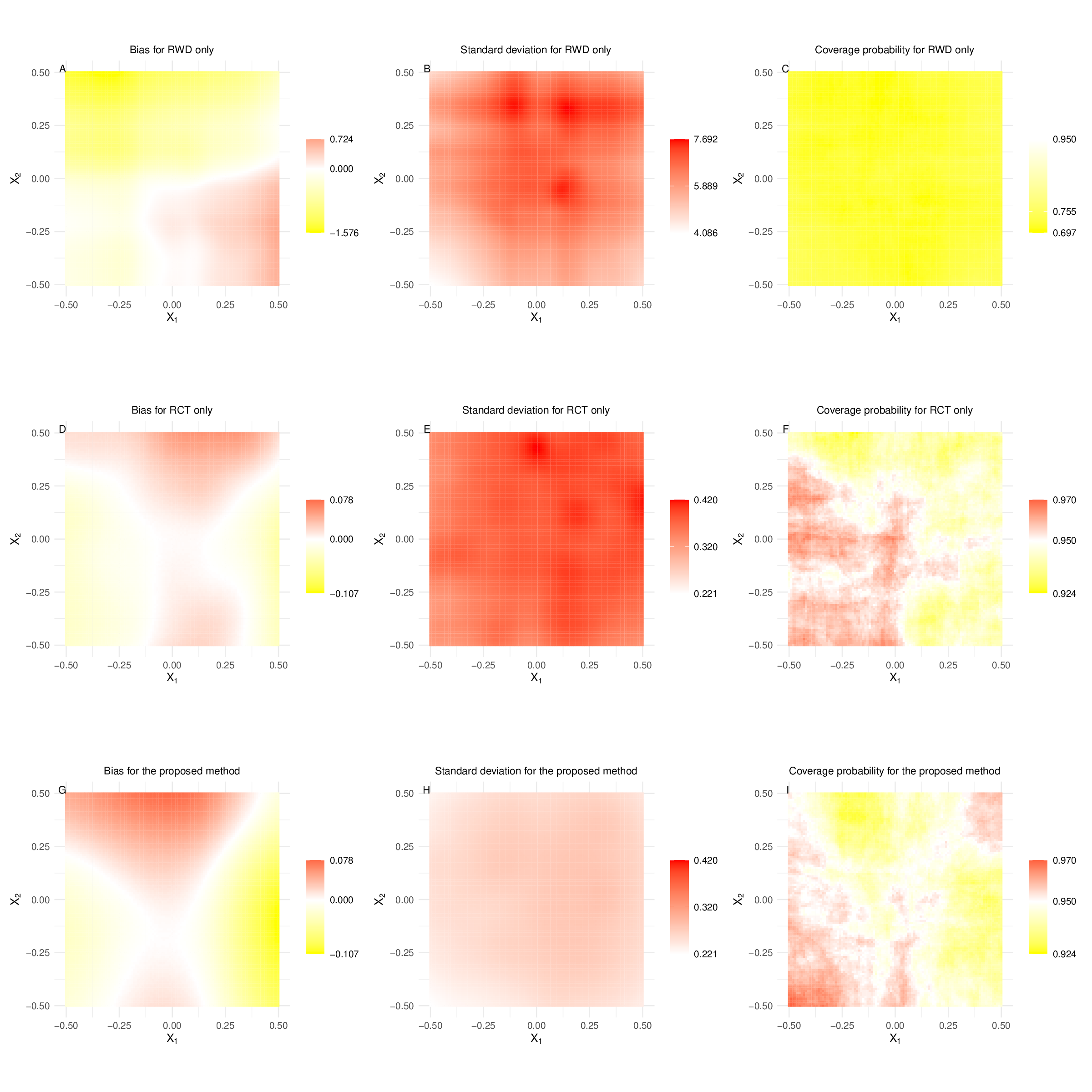}
\caption{The simulation results of Case S1 with $(n_1, n_0) = (1000, 2000)$, $X_3 = 0 $, and $X_4 = 1$.}
\label{fig-10}
\end{figure}
\begin{figure}
\centering
\includegraphics[width = \textwidth]{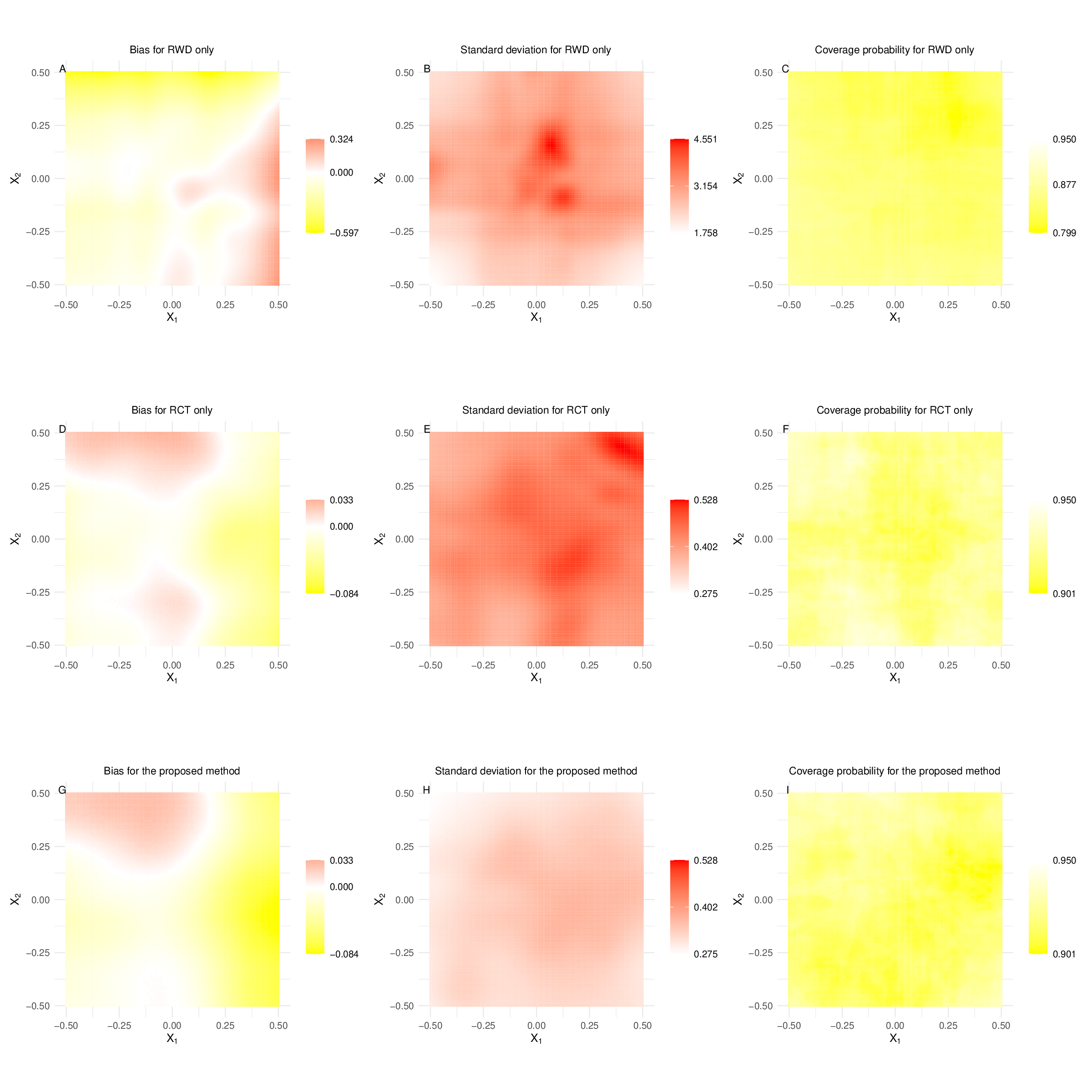}
\caption{The simulation results of Case S1 with $(n_1, n_0) = (500, 1000)$, $X_3 = 1 $, and $X_4 = 1$.}
\label{fig-11}
\end{figure}
\begin{figure}
\centering
\includegraphics[width = \textwidth]{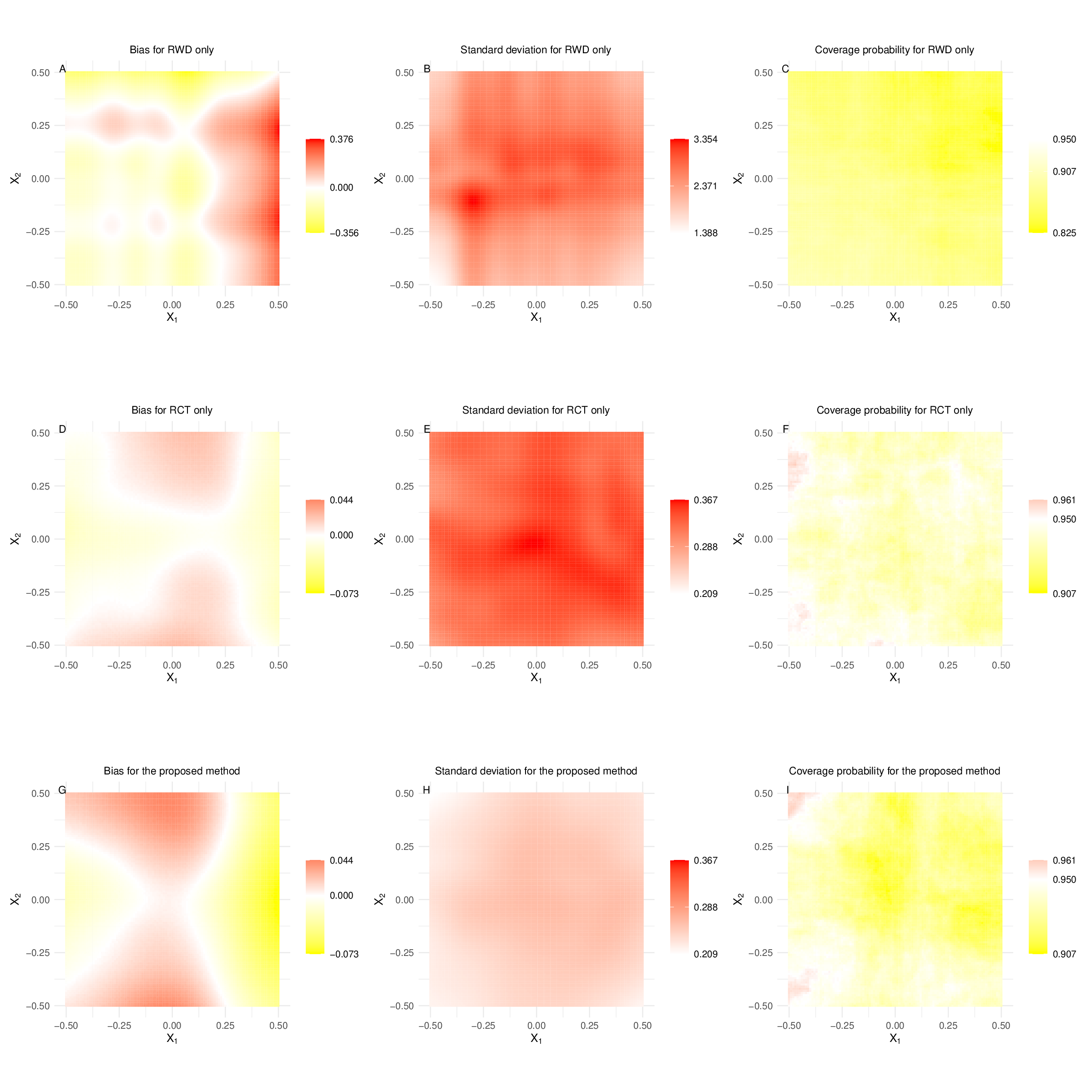}
\caption{The simulation results of Case S1 with $(n_1, n_0) = (1000, 2000)$, $X_3 = 1 $, and $X_4 = 1$.}
\label{fig-12}
\end{figure}

\begin{figure}
\centering
\includegraphics[width = \textwidth]{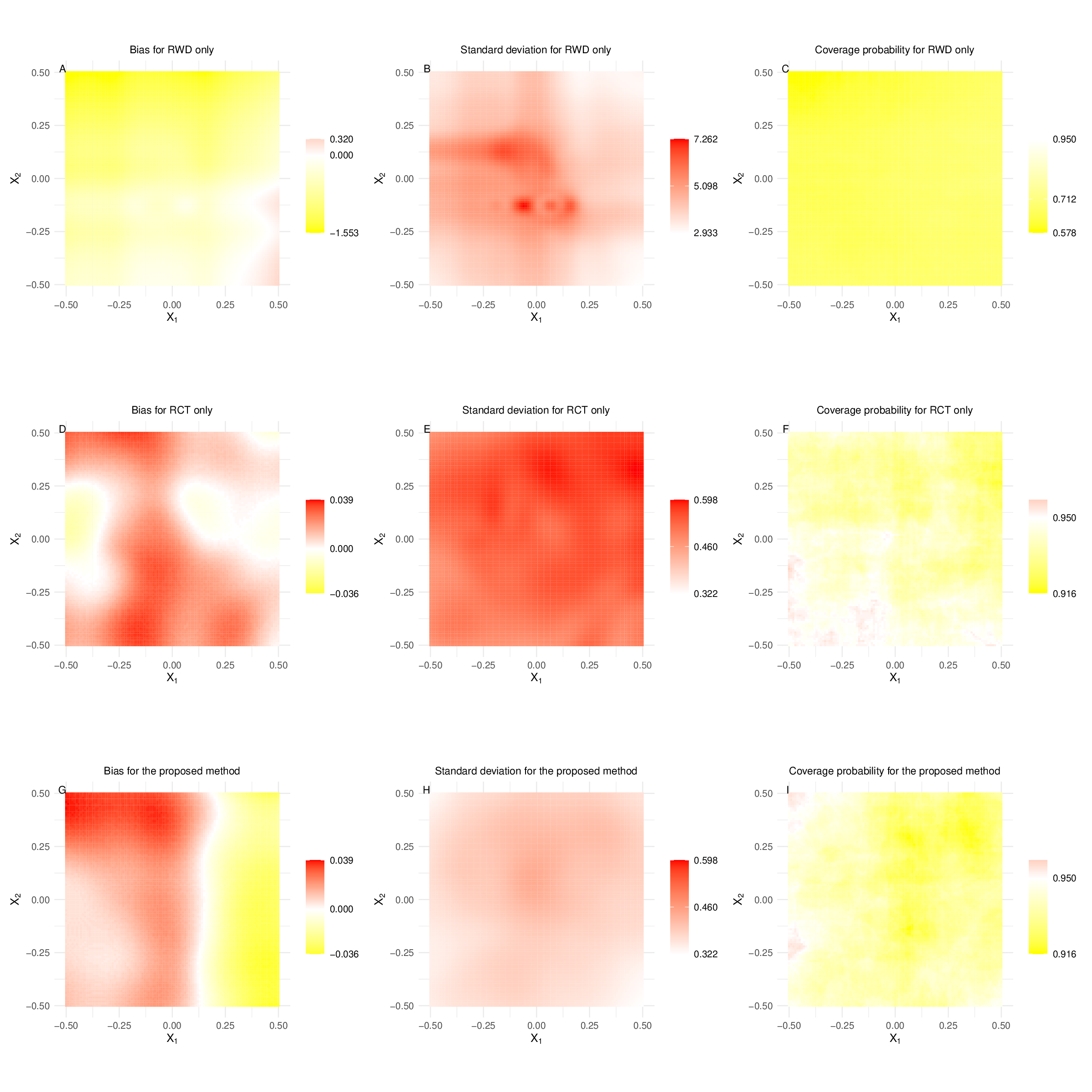}
\caption{The simulation results of Case S2 with $(n_1, n_0) = (500, 1000)$, $X_3 = 0 $, and $X_4 = 0$.}
\label{fig-13}
\end{figure}
\begin{figure}
\centering
\includegraphics[width = \textwidth]{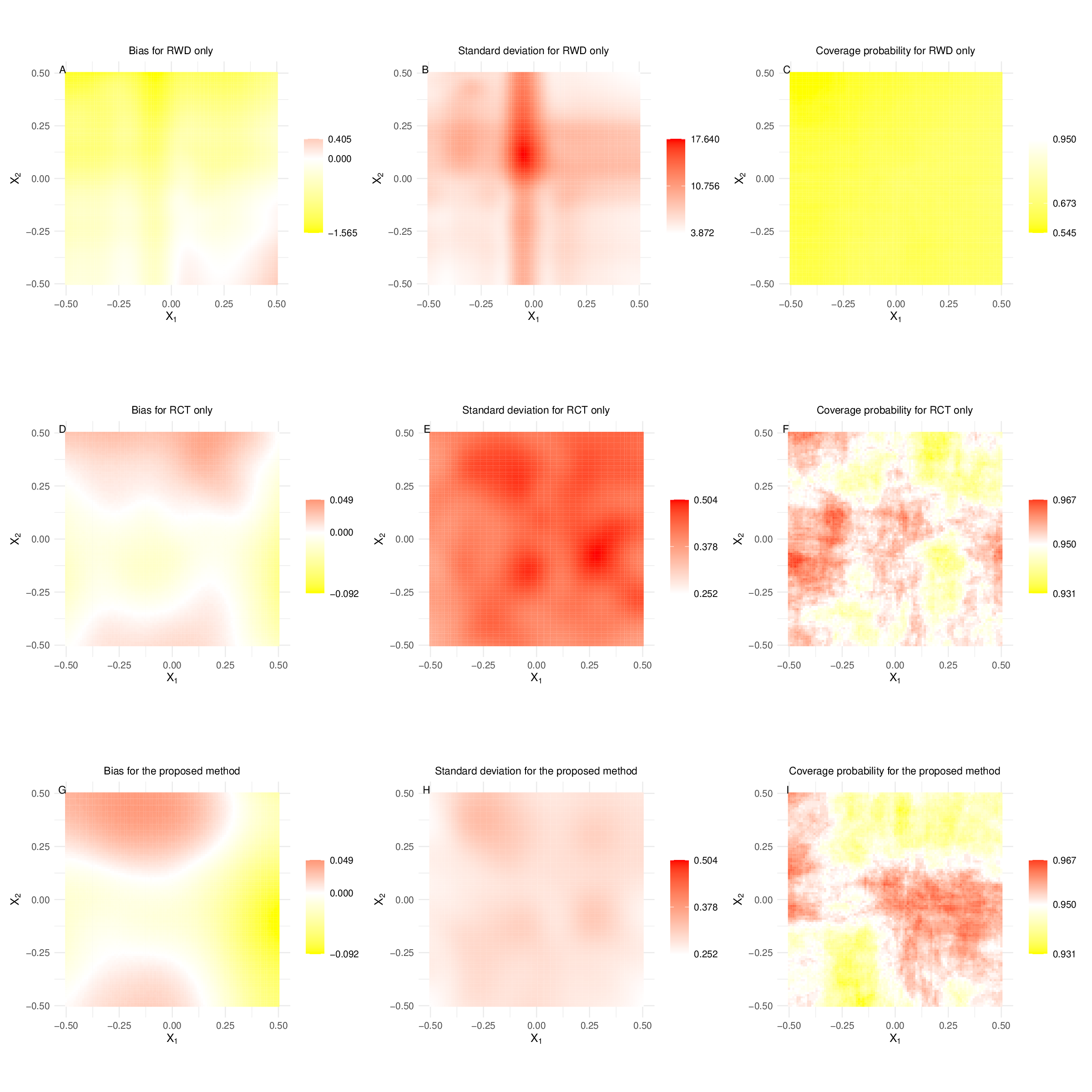}
\caption{The simulation results of Case S2 with $(n_1, n_0) = (1000, 2000)$, $X_3 = 0 $, and $X_4 = 0$.}
\label{fig-14}
\end{figure}
\begin{figure}
\centering
\includegraphics[width = \textwidth]{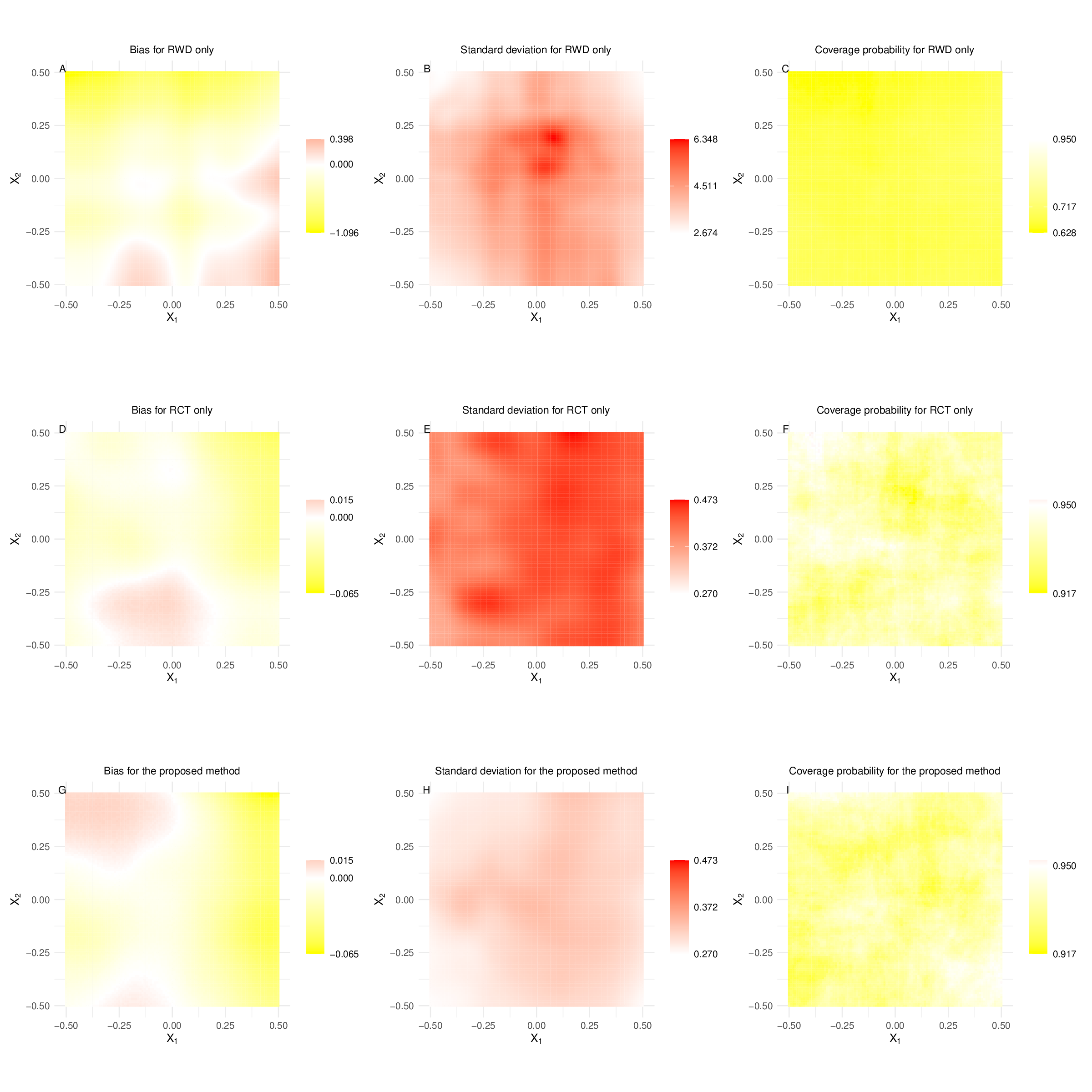}
\caption{The simulation results of Case S2 with $(n_1, n_0) = (500, 1000)$, $X_3 = 1 $, and $X_4 = 0$.}
\label{fig-15}
\end{figure}
\begin{figure}
\centering
\includegraphics[width = \textwidth]{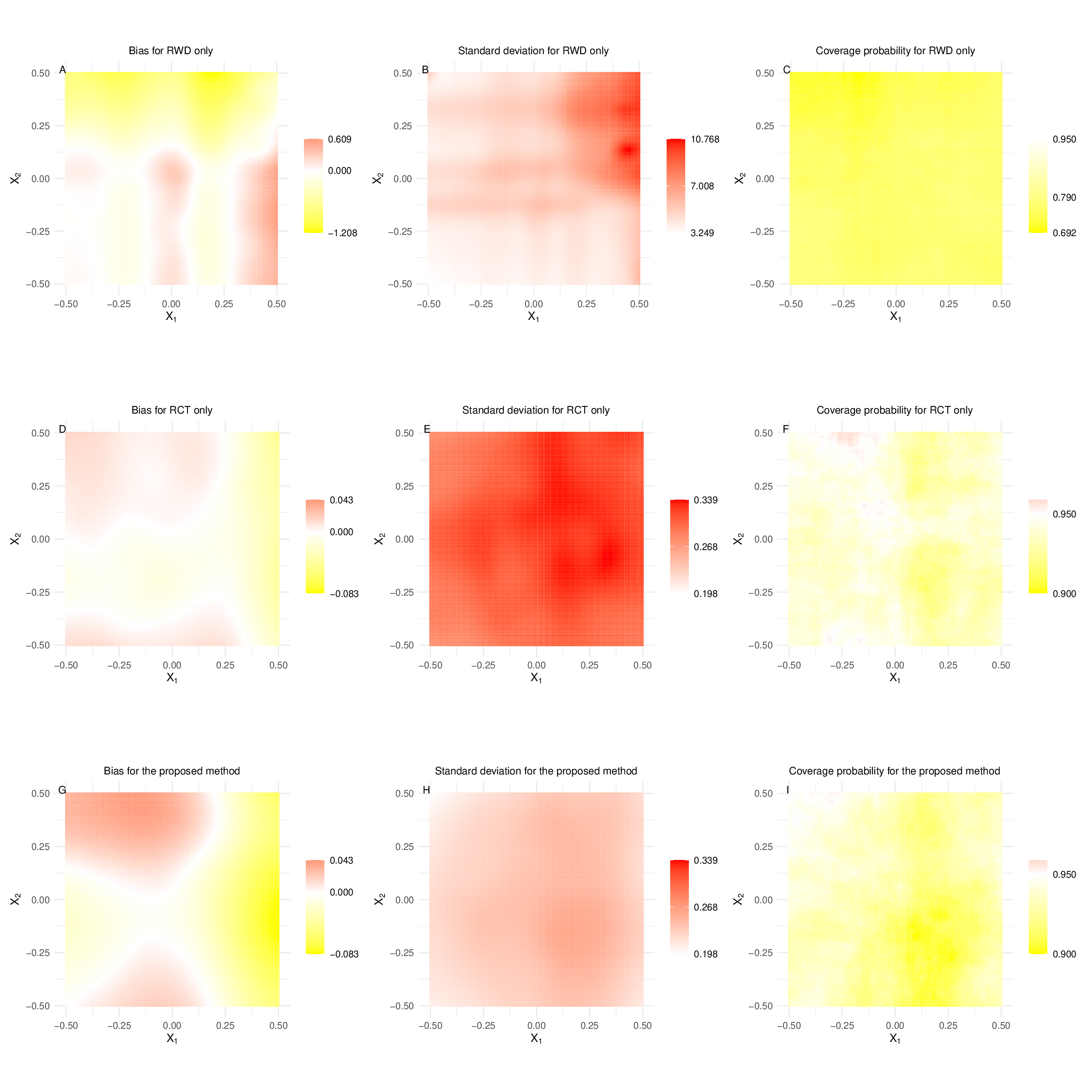}
\caption{The simulation results of Case S2 with $(n_1, n_0) = (1000, 2000)$, $X_3 = 1 $, and $X_4 = 0$.}
\label{fig-16}
\end{figure}
\begin{figure}
\centering
\includegraphics[width = \textwidth]{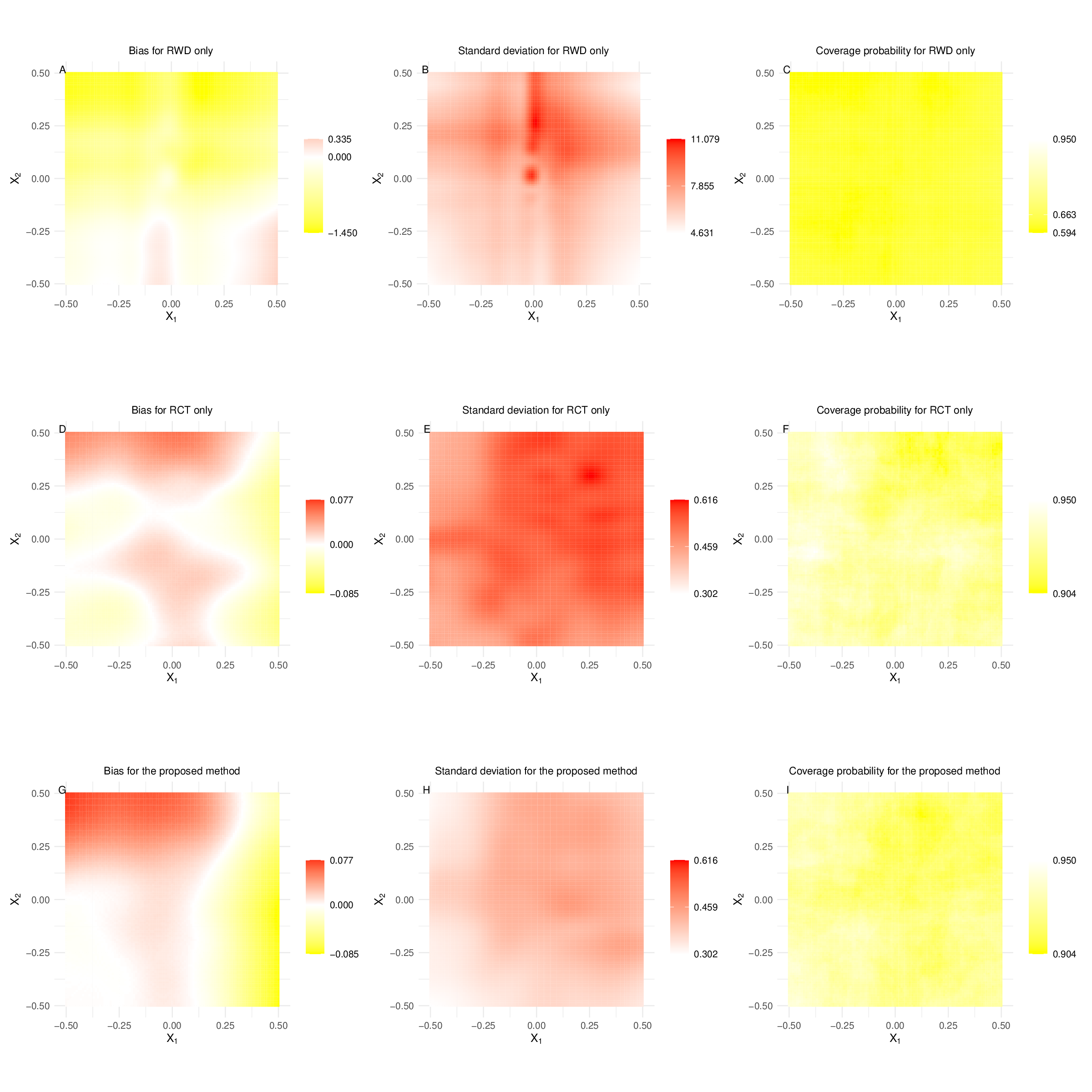}
\caption{The simulation results of Case S2 with $(n_1, n_0) = (500, 1000)$, $X_3 = 0 $, and $X_4 = 1$.}
\label{fig-17}
\end{figure}
\begin{figure}
\centering
\includegraphics[width = \textwidth]{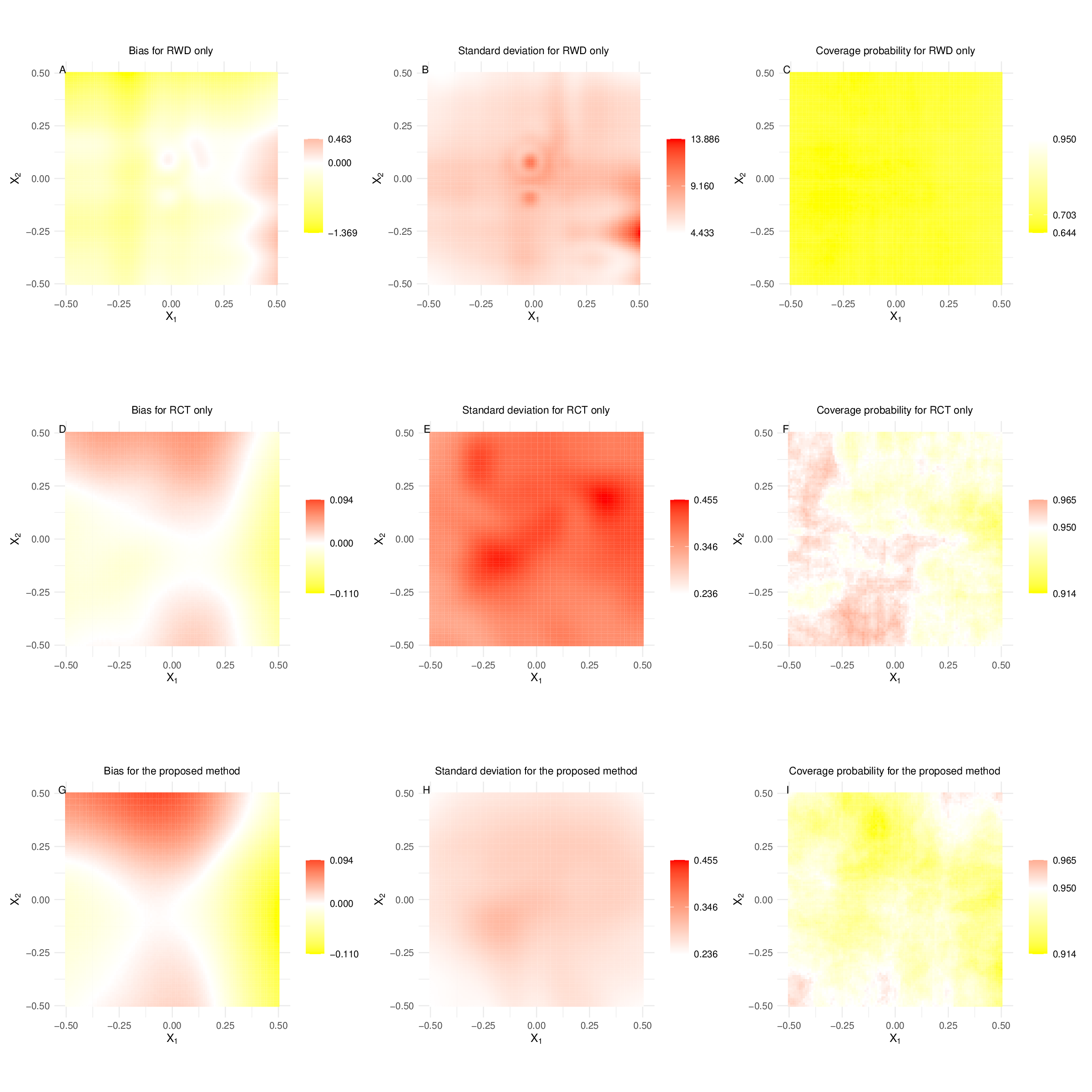}
\caption{The simulation results of Case S2 with $(n_1, n_0) = (1000, 2000)$, $X_3 = 0 $, and $X_4 = 1$.}
\label{fig-18}
\end{figure}
\begin{figure}
\centering
\includegraphics[width = \textwidth]{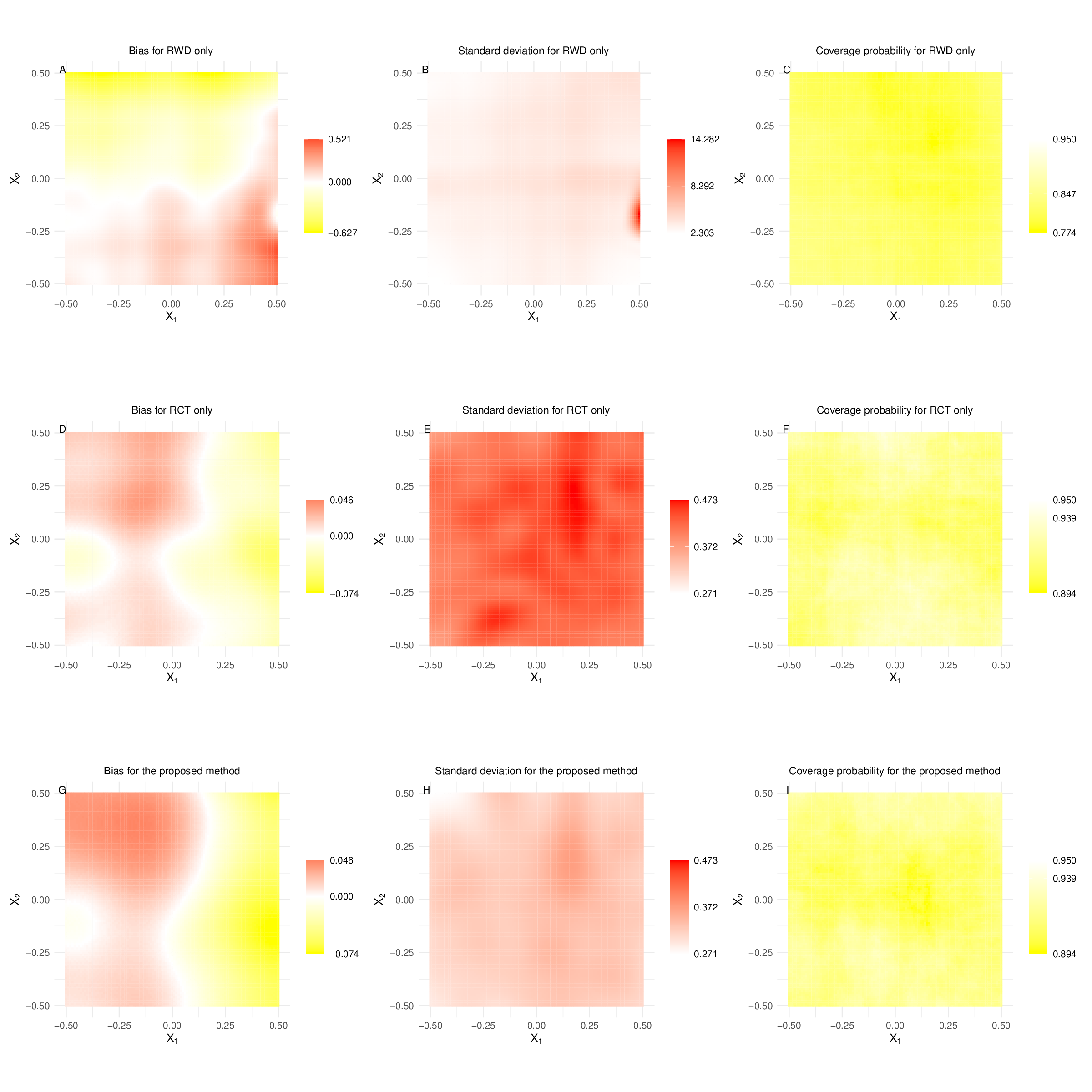}
\caption{The simulation results of Case S2 with $(n_1, n_0) = (500, 1000)$, $X_3 = 1 $, and $X_4 = 1$.}
\label{fig-19}
\end{figure}
\begin{figure}
\centering
\includegraphics[width = \textwidth]{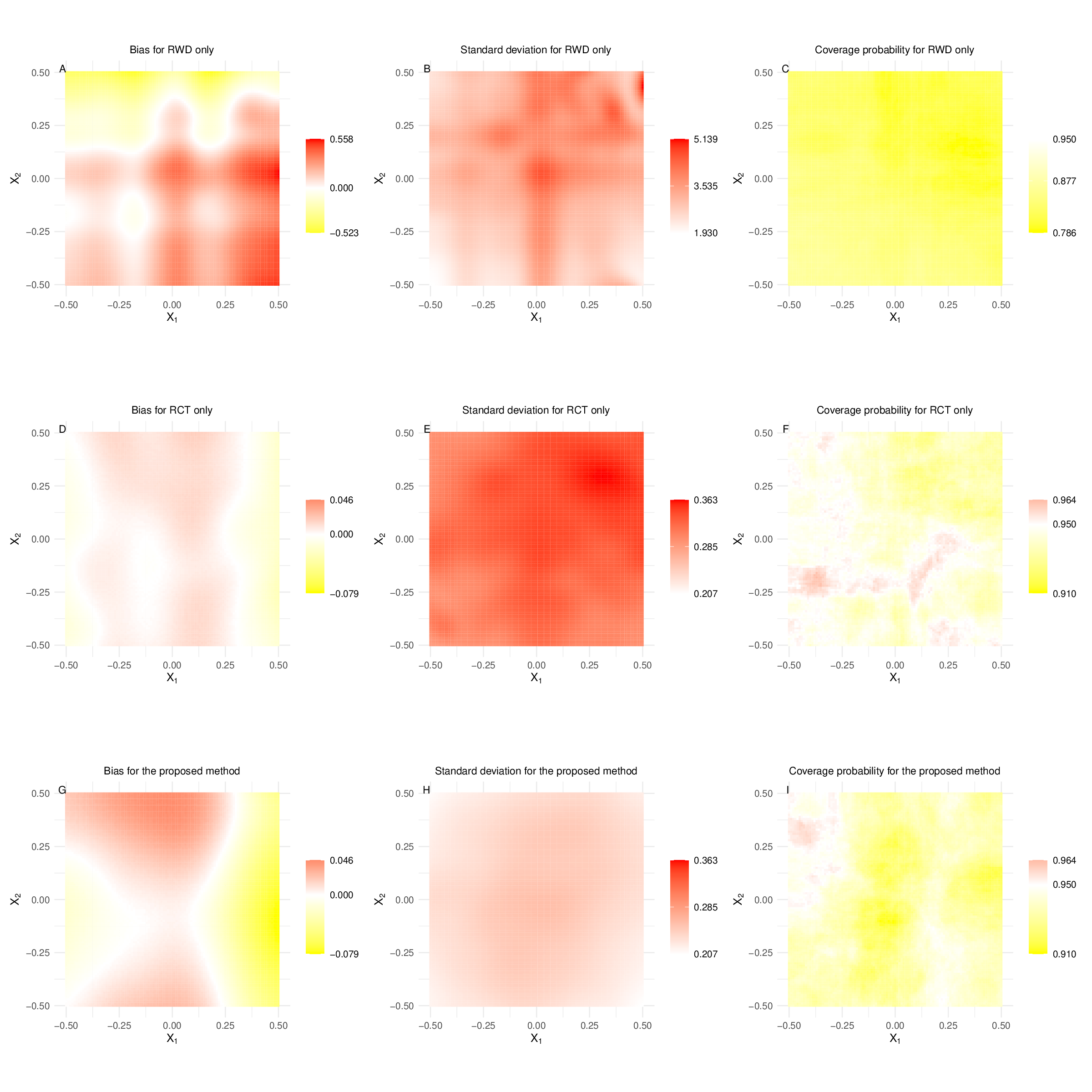}
\caption{The simulation results of Case S2 with $(n_1, n_0) = (1000, 2000)$, $X_3 = 1 $, and $X_4 = 1$.}
\label{fig-20}
\end{figure}
\begin{figure}
\centering
\includegraphics[width = \textwidth]{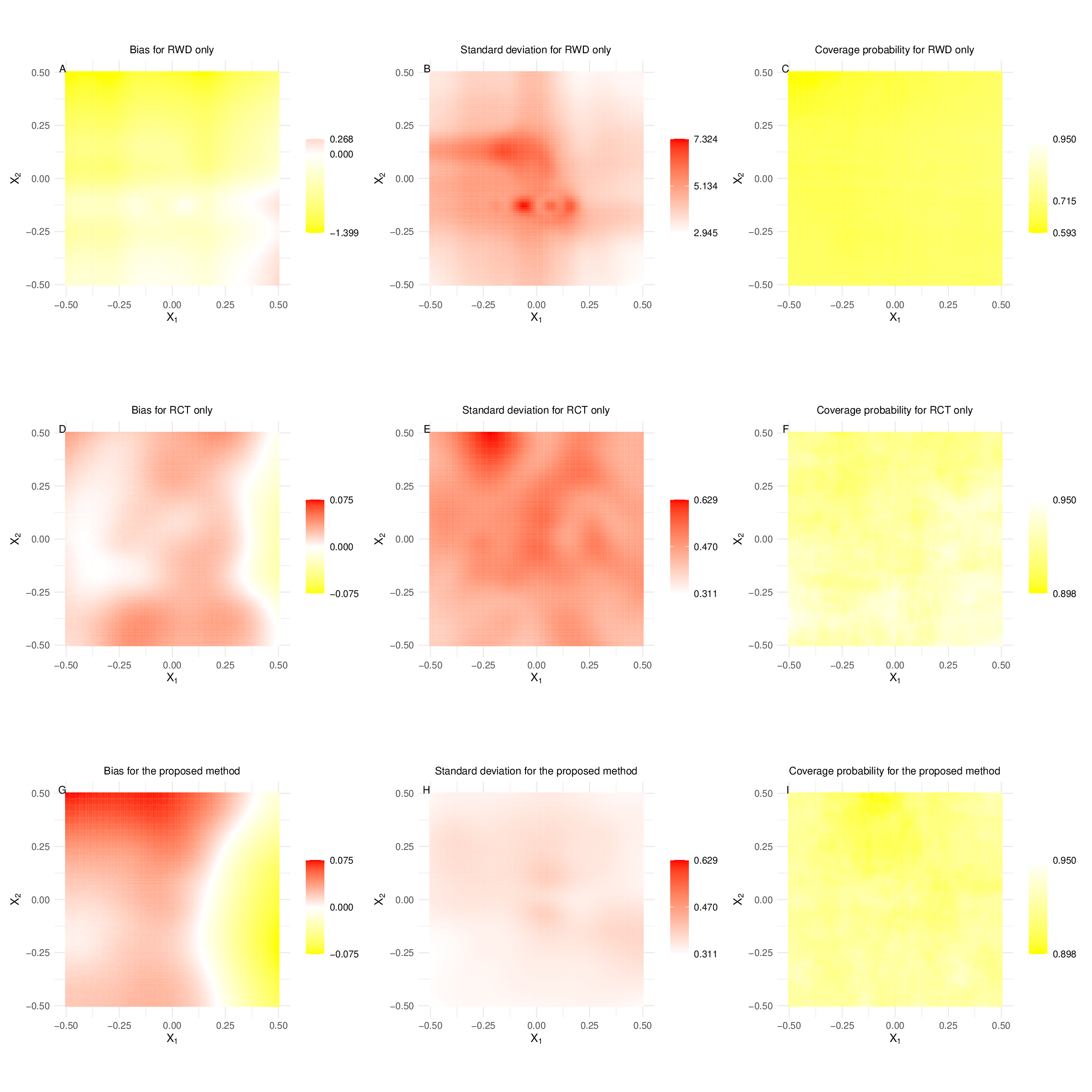}
\caption{The simulation results of Case S3 with $(n_1, n_0) = (500, 1000)$, $X_3 = 0 $, and $X_4 = 0$.}
\label{fig-21}
\end{figure}
\begin{figure}
\centering
\includegraphics[width = \textwidth]{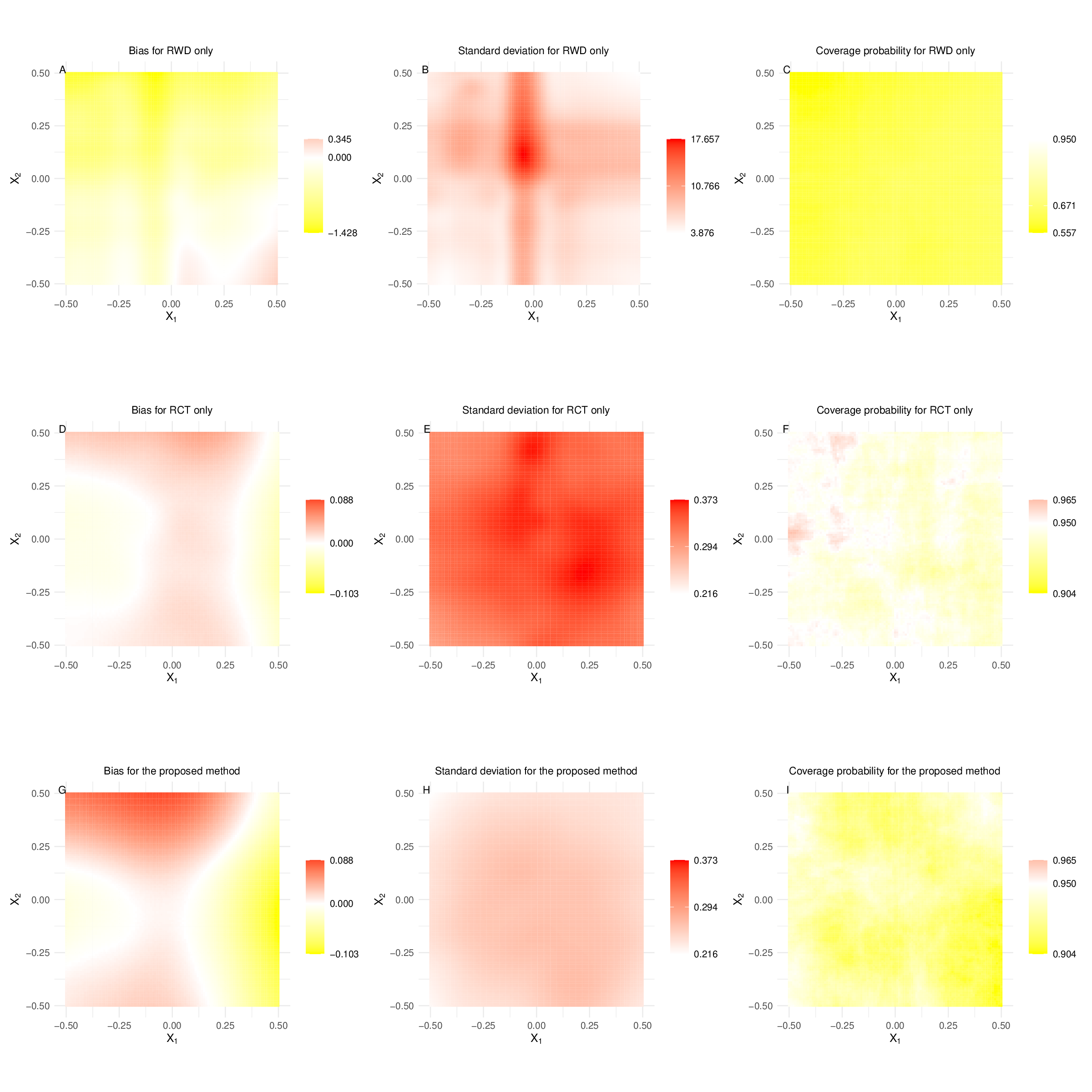}
\caption{The simulation results of Case S3 with $(n_1, n_0) = (1000, 2000)$, $X_3 = 0 $, and $X_4 = 0$.}
\label{fig-22}
\end{figure}
\begin{figure}
\centering
\includegraphics[width = \textwidth]{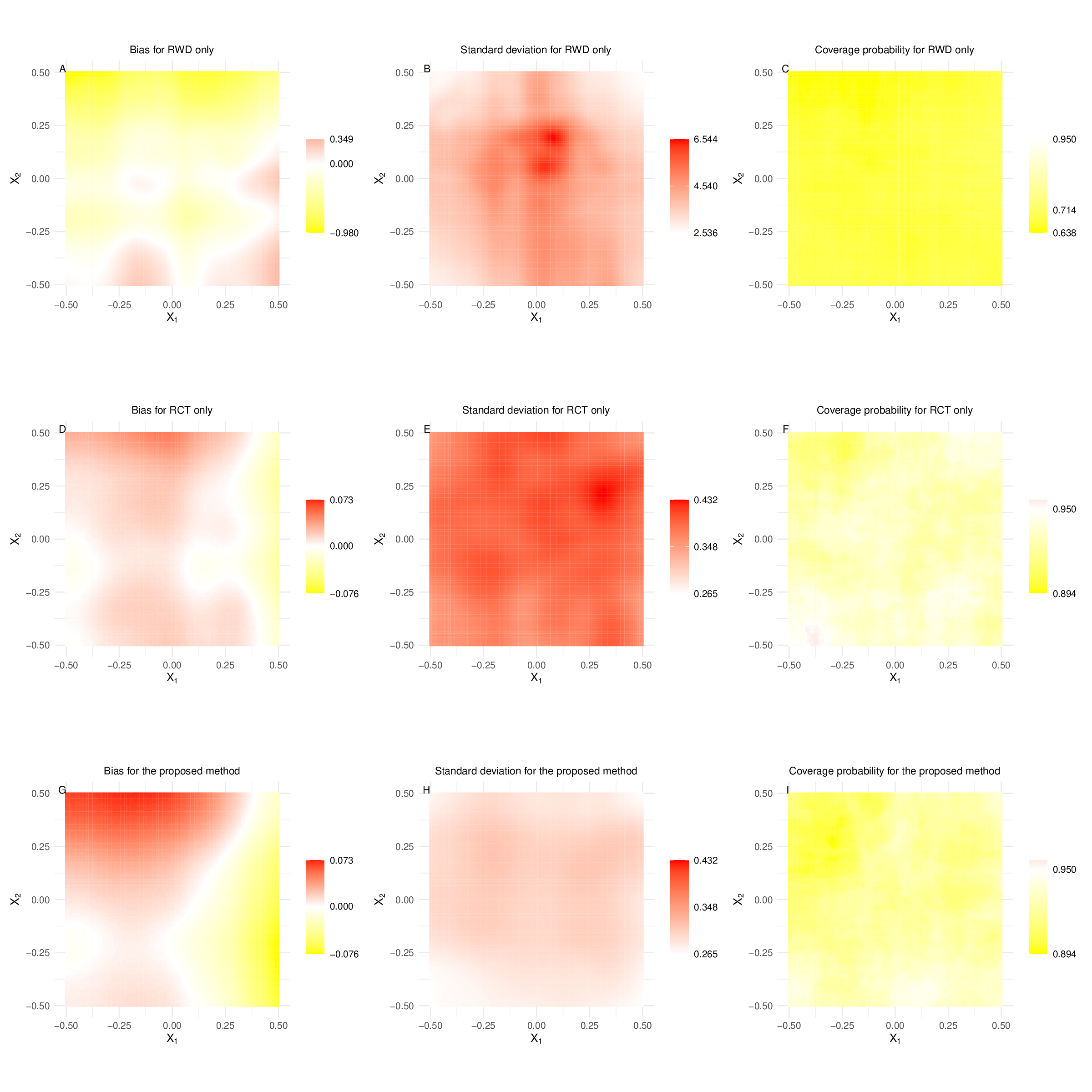}
\caption{The simulation results of Case S3 with $(n_1, n_0) = (500, 1000)$, $X_3 = 1 $, and $X_4 = 0$.}
\label{fig-23}
\end{figure}
\begin{figure}
\centering
\includegraphics[width = \textwidth]{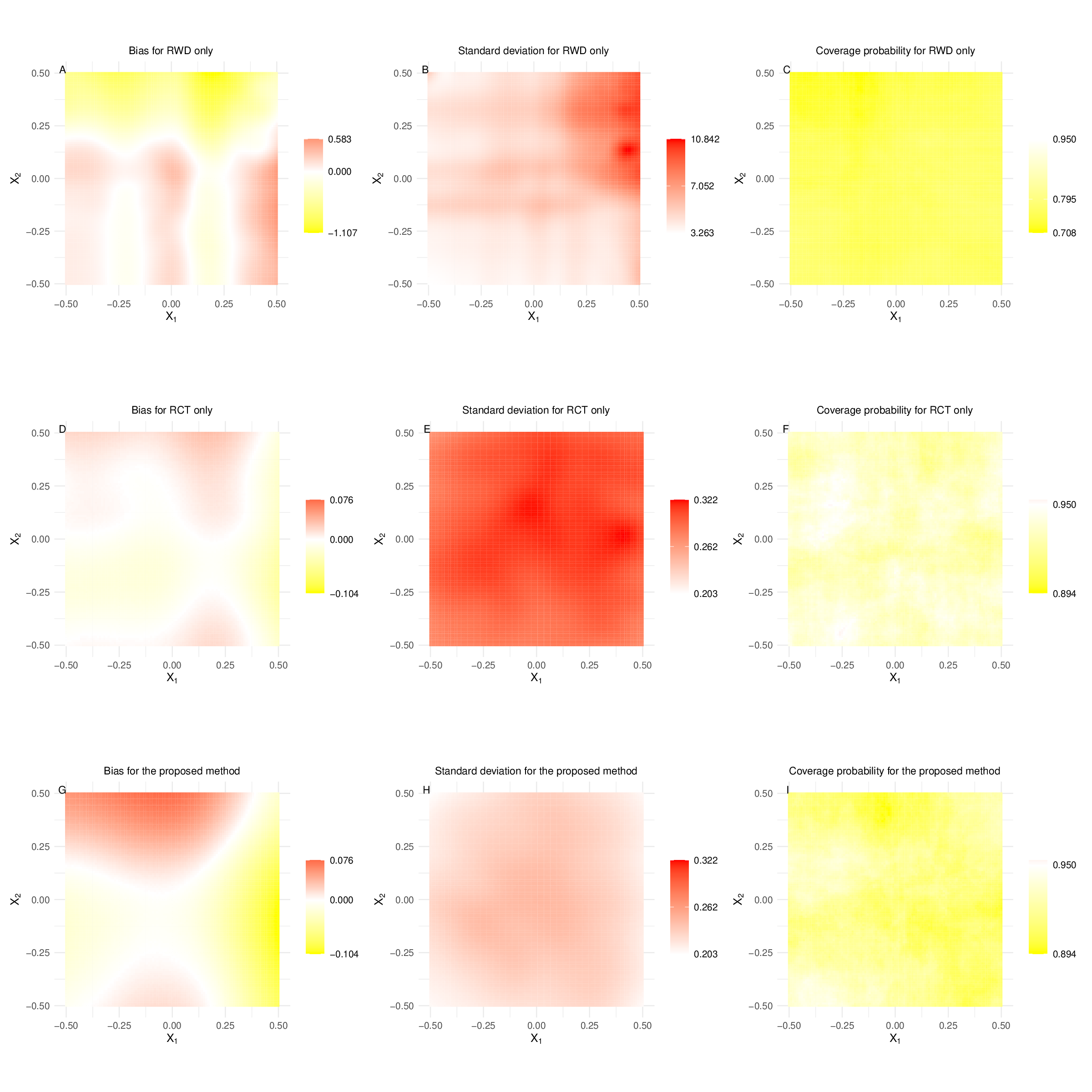}
\caption{The simulation results of Case S3 with $(n_1, n_0) = (1000, 2000)$, $X_3 = 1 $, and $X_4 = 0$.}
\label{fig-24}
\end{figure}
\begin{figure}
\centering
\includegraphics[width = \textwidth]{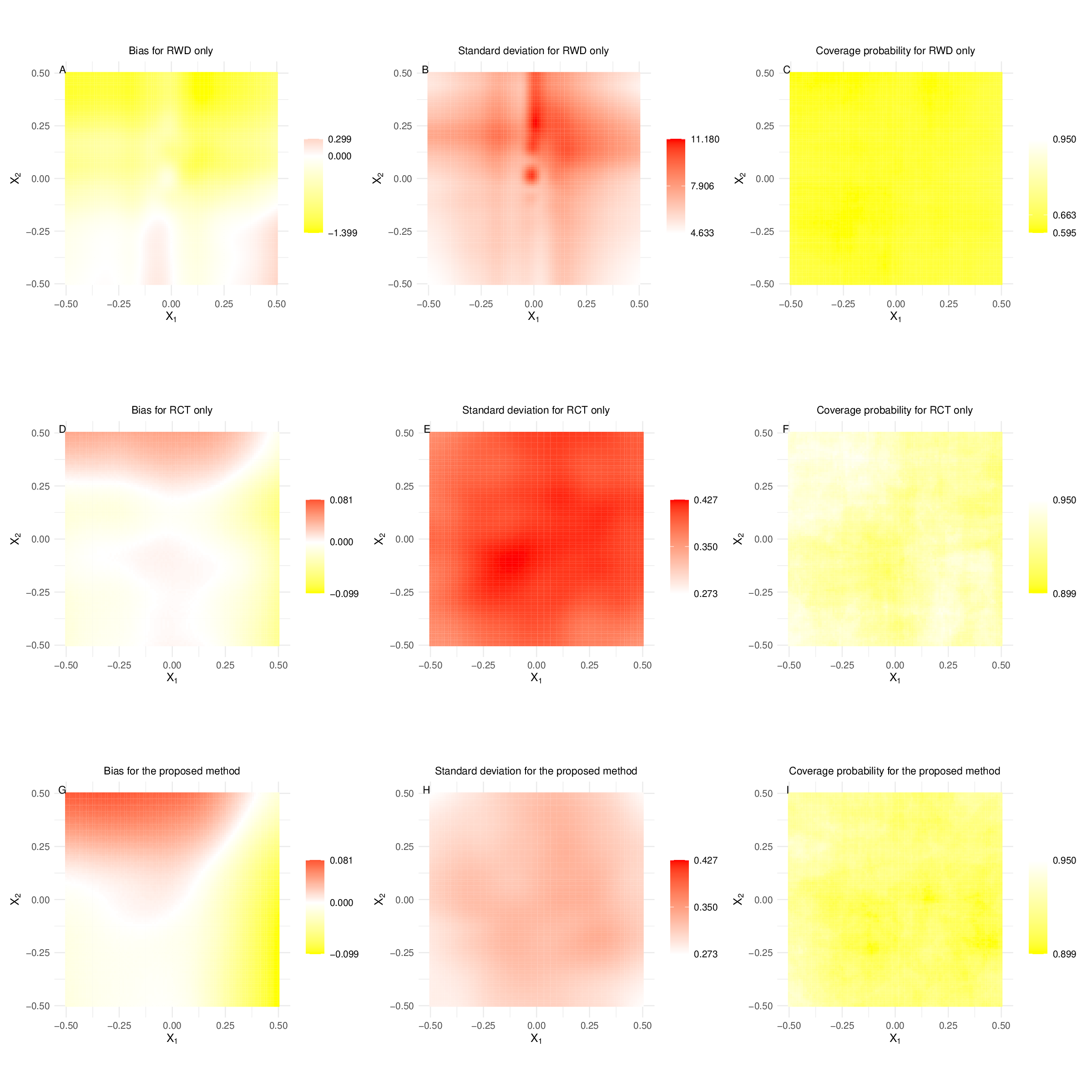}
\caption{The simulation results of Case S3 with $(n_1, n_0) = (500, 1000)$, $X_3 = 0 $, and $X_4 = 1$.}
\label{fig-25}
\end{figure}
\begin{figure}
\centering
\includegraphics[width = \textwidth]{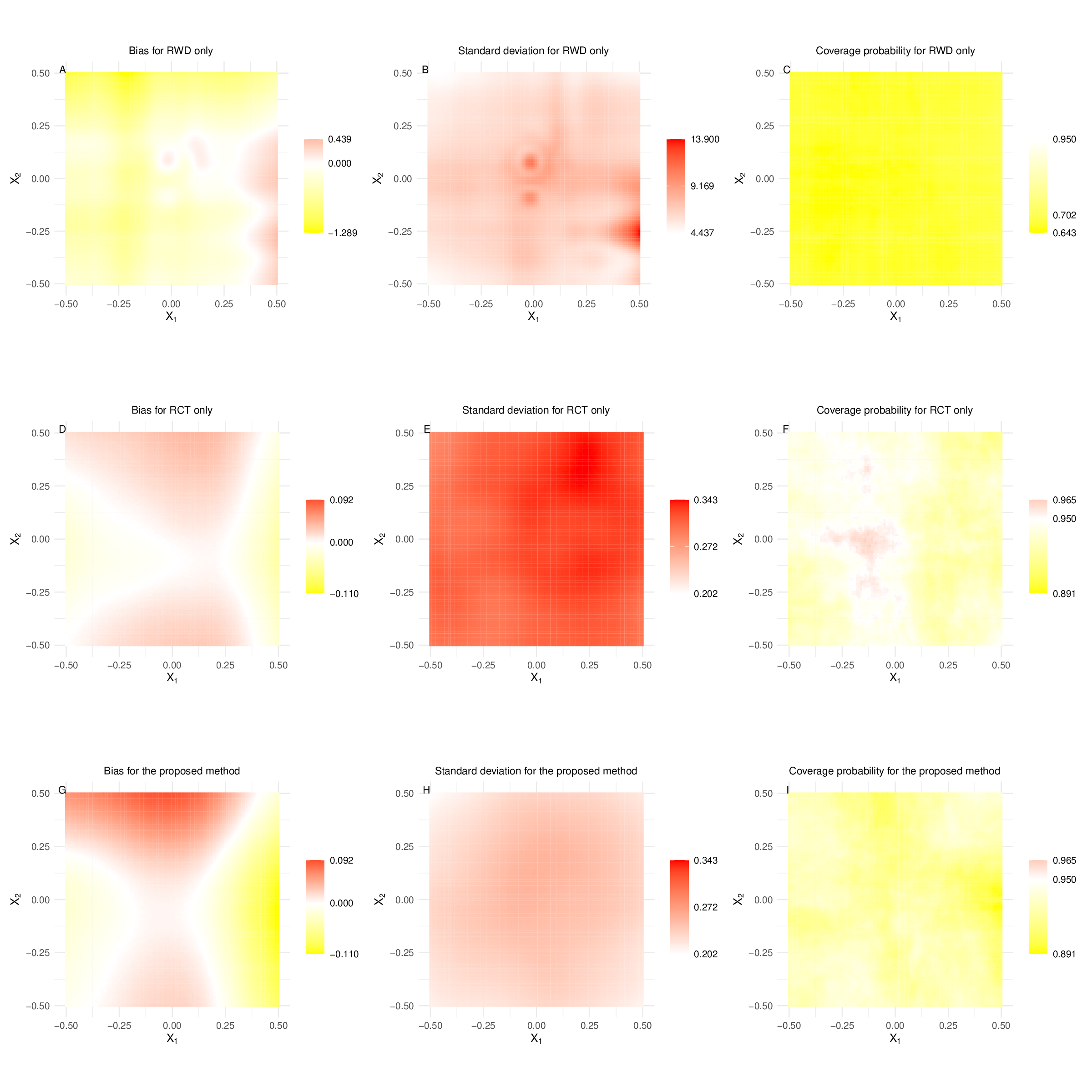}
\caption{The simulation results of Case S3 with $(n_1, n_0) = (1000, 2000)$, $X_3 = 0 $, and $X_4 = 1$.}
\label{fig-26}
\end{figure}
\begin{figure}
\centering
\includegraphics[width = \textwidth]{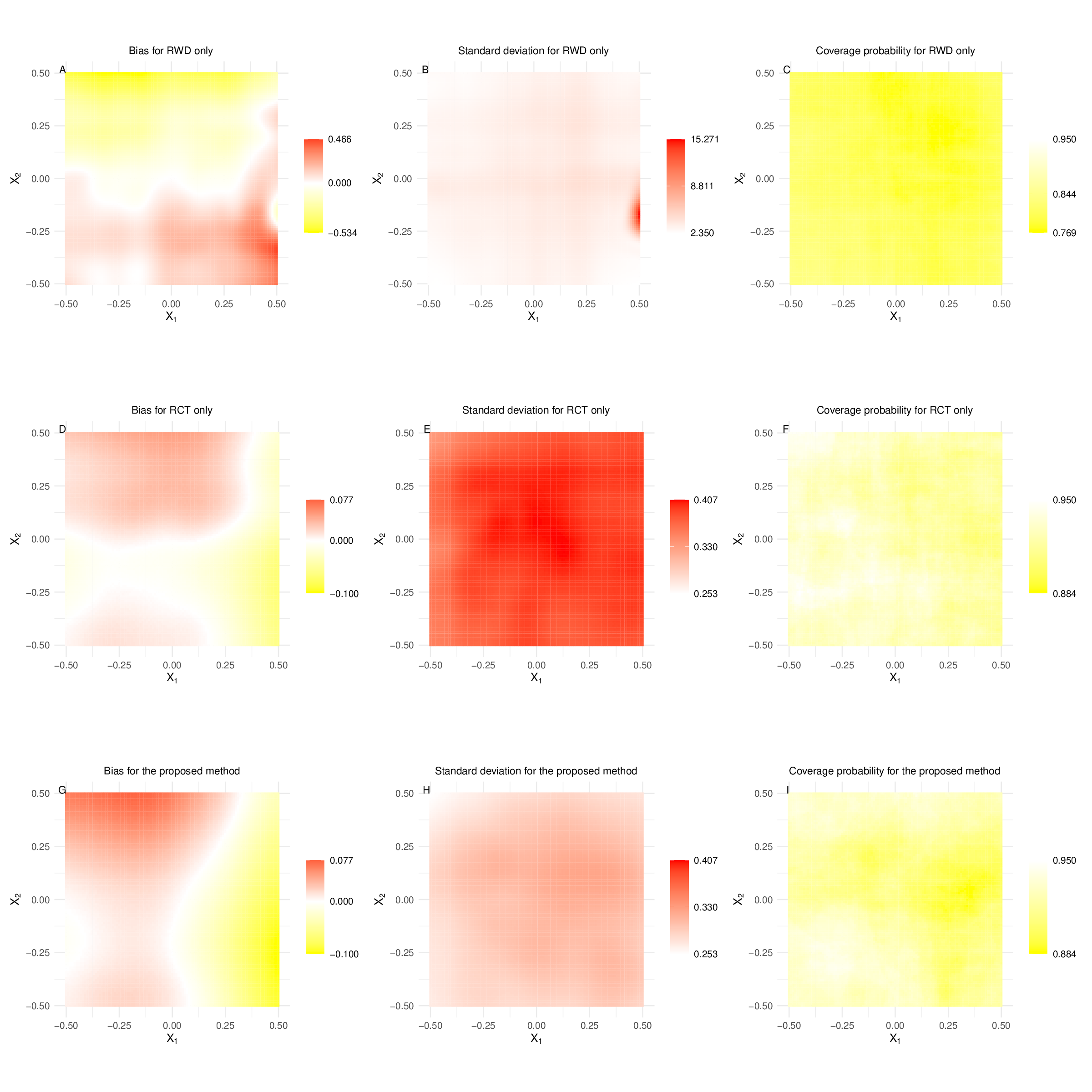}
\caption{The simulation results of Case S3 with $(n_1, n_0) = (500, 1000)$, $X_3 = 1 $, and $X_4 = 1$.}
\label{fig-27}
\end{figure}
\begin{figure}
\centering
\includegraphics[width = \textwidth]{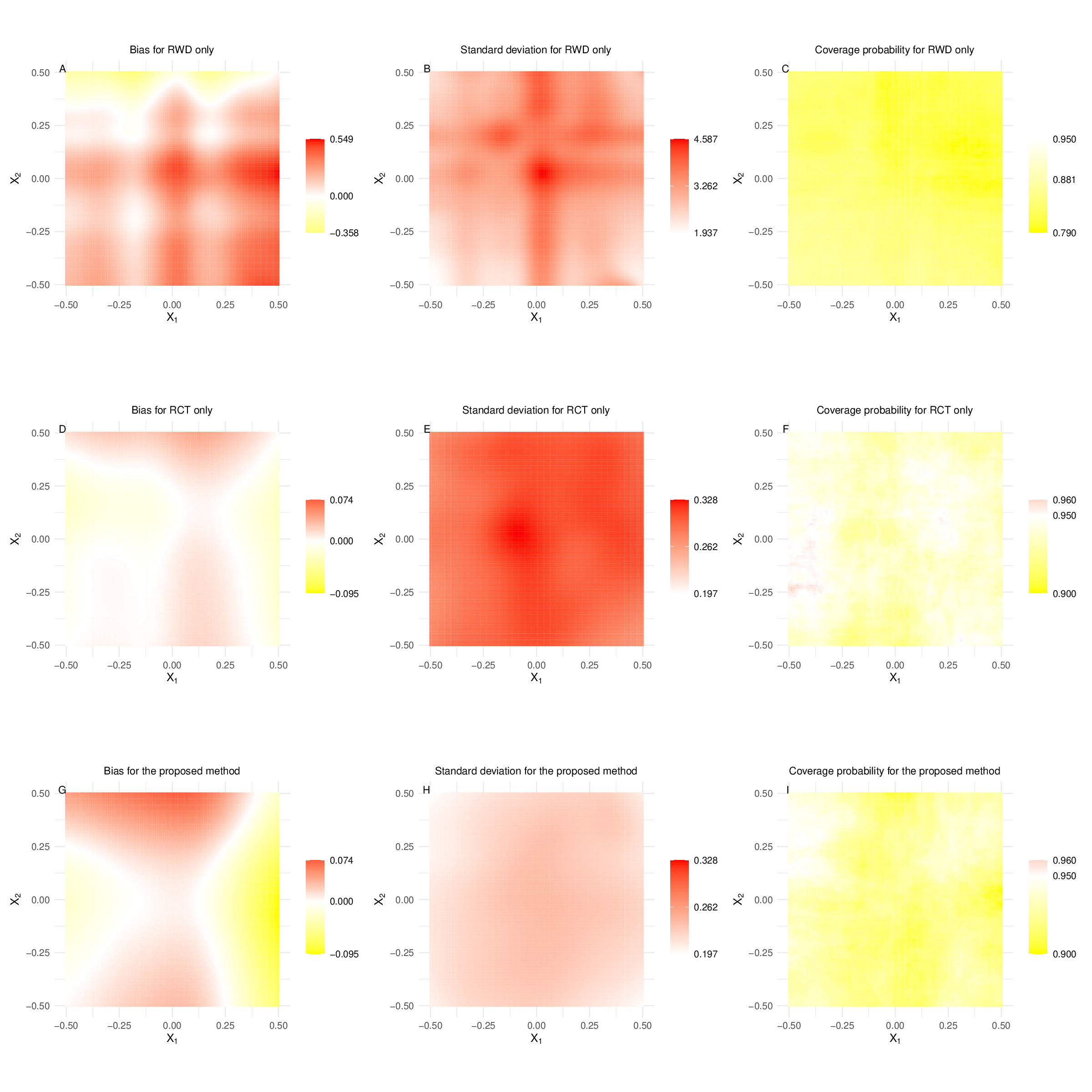}
\caption{The simulation results of Case S3 with $(n_1, n_0) = (1000, 2000)$, $X_3 = 1 $, and $X_4 = 1$.}
\label{fig-28}
\end{figure}

\end{document}